\documentclass[acmsmall,review=false,screen, nonacm=true]{acmart}

\usepackage{booktabs} %

\setcitestyle{authoryear}

\usepackage[colorinlistoftodos]{todonotes}
\usepackage{amsmath,amsthm}

\usepackage{mathtools}
\usepackage{bbm}
\usepackage{graphicx} %
\usepackage{xcolor,colortbl}
\usepackage{pifont}%

\usetikzlibrary{patterns}
\usepackage{booktabs}
\usepackage{multirow}
\usepackage{makecell}

\usepackage{enumitem}

\usepackage{xfrac}

\usepackage[ruled, linesnumbered]{algorithm2e} 
\SetAlFnt{\small}
\SetAlCapFnt{\small}
\SetAlCapNameFnt{\small}
\SetAlCapHSkip{0pt}
\IncMargin{-\parindent}
\usepackage{placeins}
\SetKwInOut{Input}{Input} 
\SetKwInOut{Output}{Output}
\DontPrintSemicolon

\usepackage[nameinlink]{cleveref} %

\usepackage{thmtools}
\usepackage{thm-restate}

\theoremstyle{plain}
\newtheorem{theorem}{Theorem}[section]
\newtheorem{definition}[theorem]{Definition}
\newtheorem{proposition}[theorem]{Proposition}

\newtheorem{lemma}[theorem]{Lemma}

\newtheorem{observation}[theorem]{Observation}

\newtheorem{example}[theorem]{Example}

\newtheorem{corollary}[theorem]{Corollary}

\DeclarePairedDelimiter\norm{\lVert}{\rVert}%
\DeclarePairedDelimiter\abs{\lvert}{\rvert}%
\makeatletter
\let\oldabs\abs
\def\abs{\@ifstar{\oldabs}{\oldabs*}}
\let\oldnorm\norm
\def\norm{\@ifstar{\oldnorm}{\oldnorm*}}
\makeatother

\newcommand{\pjr}{\text{PJR}}
\newcommand{\jr}{\text{JR}}
\newcommand{\ejr}{\text{EJR}}
\newcommand{\ejrp}{\text{EJR+}\xspace}

\newcommand{\cmark}{\ding{51}}
\newcommand{\xmark}{\ding{55}}

\newcommand{\algstyle}[1]{\textsc{#1}\xspace}
\newcommand{\gjcr}{\algstyle{GJCR}}
\newcommand{\noisygjcr}{\algstyle{Softmax-GJCR}}
\newcommand{\noisygjcrslack}{\algstyle{Softmax-GJCR-Slack}}
\newcommand{\noisygjcrcap}{\algstyle{Softmax-GJCR-Capped}}
\newcommand{\noisygreedy}{\algstyle{Softmax-GreedyJR}}
\newcommand{\noisypav}{\algstyle{Softmax-PAV}}

\newcommand{\dynamicgjcr}{\algstyle{Dynamic-GJCR}}

\newcommand{\lm}{\ensuremath{\ell_{\text{max}}}\xspace}

\newcommand{\cA}{\mathcal{A}}
\newcommand{\cD}{\mathcal{D}}
\newcommand{\cE}{\mathcal{E}}
\newcommand{\cI}{\mathcal{I}}
\newcommand{\cL}{\mathcal{L}}
\newcommand{\cM}{\mathcal{M}}
\newcommand{\cN}{\mathcal{N}}
\newcommand{\cR}{\mathcal{R}}
\newcommand{\cS}{\mathcal{S}}
\newcommand{\cT}{\mathcal{T}}
\newcommand{\W}{\mathcal{W}}
\newcommand{\cW}{\W}

\newcommand{\D}{\cD}

\newcommand{\eps}{\varepsilon}
\newcommand{\LW}{\cL}

\newcommand{\qq}{\qquad}

\newcommand{\qqqq}{\qquad\qquad}
\newcommand{\defeq}{\coloneq}

\newcommand{\activatebar}{%
  \begingroup\lccode`~=`| \lowercase{\endgroup\def~}{\,\delimsize\vert\,}%
  \mathcode`|="8000
}
\DeclareMathOperator*{\expectation}{\mathbb{E}}
\newcommand{\expect}{\expectation\expectarg}
\DeclarePairedDelimiterX{\expectarg}[1]{[}{]}{%
	\ifnum\currentgrouptype=16 \else\begingroup\fi
	\activatebar#1
	\ifnum\currentgrouptype=16 \else\endgroup\fi
}

\newcommand{\probability}{\Pr}
\newcommand{\prob}[1]{\probability\left[#1\right]}
\newcommand{\probover}[2]{\probability_{#1}\left[#2\right]}
\newcommand\given{\;\middle\vert\;}

\newcommand{\pav}[1]{{h_{#1}}}
\newcommand{\pavp}[1]{{h'_{#1}}}
\newcommand{\dtv}[2]{\operatorname{d}_{TV}\left( #1 ,\: #2 \right)}

\DeclareMathOperator{\harm}{Harm}
\newcommand{\alg}{\textsc{Alg}}
\DeclareMathOperator{\supp}{supp}

\usepackage{dsfont}

\usepackage{soul}

\setcitestyle{authoryear}

\copyrightyear{2025}
\acmYear{2025}
\setcopyright{cc}
\setcctype{by}
\acmConference{}{}{}
\acmBooktitle{The 26th ACM Conference on Economics and Computation (EC '25), July 7--10, 2025, Stanford, CA, USA}
\acmDOI{10.1145/3736252.3742676}
\acmISBN{979-8-4007-1943-1/25/07}
\acmYear{2025}\copyrightyear{2025}
\setcopyright{acmlicensed}

\title[Robust Committee Voting]{\texorpdfstring{Robust Committee Voting, \\ or The Other Side of Representation}{Robust Committee Voting, or The Other Side of Representation}}

\author{Gregory Kehne}
\email{kehne@wustl.edu}
\orcid{0009-0002-3375-5576}
\affiliation{%
  \institution{Washington University in Saint Louis}
  \city{Saint Louis}
  \state{Missouri}
  \country{USA}
}
\author{Ulrike Schmidt-Kraepelin}
\email{u.schmidt.kraepelin@tue.nl}
\orcid{0000-0002-9213-7746}
\affiliation{
  \institution{TU Eindhoven}
  \city{Eindhoven}
  \country{The Netherlands}
}
\author{Krzysztof Sornat}
\email{sornat@agh.edu.pl}
\orcid{0000-0001-7450-4269}
\affiliation{%
  \institution{AGH University}
  \city{Kraków}
  \country{Poland}
}

\begin{abstract}
We study approval-based committee voting from a novel perspective. 
While extant work largely centers around proportional representation of the voters, we shift our focus to the candidates while preserving proportionality.
Intuitively, candidates supported by similar voter groups should receive comparable representation.
Since deterministic voting rules cannot achieve this ideal, we develop randomized voting rules that satisfy ex-ante neutrality, monotonicity, and continuity, while maintaining strong ex-post proportionality guarantees.  

Continuity of the candidate selection probabilities proves to be the most demanding of our ex-ante desiderata.
We provide it via voting rules that are algorithmically stable, a stronger notion of robustness which captures the continuity of the committee distribution under small changes.
First, we introduce \textsc{Softmax-GJCR}, a randomized variant of the \emph{Greedy Justified Candidate Rule (GJCR)} \cite{brill23robust}, which carefully leverages slack in GJCR to satisfy our ex-ante properties. 
This polynomial-time algorithm satisfies EJR+ ex post, assures ex-ante monotonicity and neutrality, and provides $O(k^3/n)$-stability (ignoring $\log$ factors). 
Building on our techniques for \textsc{Softmax-GJCR}, we further show that stronger stability guarantees can be attained by (i) allowing exponential running time, (ii) relaxing EJR+ to an approximate $\alpha$-EJR+, and (iii) relaxing EJR+ to JR.  

We finally demonstrate the utility of stable voting rules in other settings.
In online dynamic committee voting, we show that stable voting rules imply dynamic voting rules with low expected recourse, and illustrate this reduction for \textsc{Softmax-GJCR}.
Our voting rules also satisfy a stronger form of stability that coincides with differential privacy, suggesting their applicability in privacy-sensitive domains.
\end{abstract}

\begin{document}

\hypersetup{colorlinks=true,
   linkcolor=teal,
   citecolor=teal, 
   filecolor=teal,
   urlcolor=teal,
}

\maketitle

\newpage
\section{Introduction}
Over the past decade, approval-based committee voting has attracted increasing attention in social choice theory. Here, each voter approves a subset of candidates, and a voting rule selects a $k$-sized subset (the \emph{committee}) using voters' preferences. Thanks to its simplicity, approval-based committee voting has multiple real-world applications: in proof-of-stake blockchain systems, committee elections are held on a daily basis to select validators \cite{cevallos2021verifiably,boehmer2024approval}; committee elections constitute a fundamental special case of participatory budgeting \cite{aziz2021participatory,rey2023computational}; and scholars have proposed using committee voting rules in online civic participation platforms, such as \emph{Polis}, to automatically aggregate large numbers of statements into a representative selection \cite{halpern23representation,fish2024generative}.

The academic literature has largely focused on designing committee voting rules that achieve the ideal of \emph{proportional representation}—that is, any $\alpha$-fraction of the voters should be represented by an $\alpha$-fraction of the committee. While proportional representation ensures that voters are adequately represented in the outcome, it does not prevent discrimination among candidates. This becomes particularly problematic when "candidates" are not objects but people or stakeholders who are invested in being part of the committee. 
For example, in blockchain applications, validators receive monetary compensation, making committee inclusion financially significant. Similarly, in participatory budgeting—a civic participation initiative where voters decide which projects (the candidates) should be funded in their municipality---projects are often proposed by nonprofit associations that are strongly invested in their implementation. 
In such scenarios, how can we justify why one candidate was selected while another candidate with similar support was not? 

Motivated by these considerations, we formulate the ideal of \emph{candidate fairness}: 
\begin{center}
    \emph{Candidates with similar sets of supporters should receive similar representation.}
\end{center}
Clearly, deterministic committee voting rules cannot satisfy candidate fairness, as "similar representation" is inherently impossible in a binary world where candidates are either \emph{included} in or \emph{excluded} from the committee.
We therefore turn to randomized committee voting rules and impose three well-established desiderata on the selection probabilities that, together, capture the essence of candidate fairness: (i) \emph{Neutrality}: Selection probabilities should be independent of candidate identities, ensuring that candidates with identical supporter sets have equal chances of being selected. (ii) \emph{Monotonicity}: Increasing support for a candidate should only increase its selection probability. (iii) \emph{Continuity}: Small changes in the election should result in only small changes in selection probabilities. We refer to these as \emph{ex-ante} properties. Beyond these, our randomized committee voting rules guarantee proportionality ex post by randomizing only over committees that satisfy well-established proportionality axioms, such as \emph{Extended Justified Representation} (\ejr{}) \cite{aziz17justified, lackner23multi}.

Continuity stands out among our three ex-ante guarantees, both conceptually and technically.
As any constant voting rule illustrates, continuity alone does not confer any fairness guarantees. 
However, existing axioms are strengthened when paired with continuity.
For instance, the combination of neutrality and continuity promotes the ideal of candidate fairness because if two candidates in a profile $P$ have almost identical, or even merely almost symmetric, sets of approving voters, then continuity stipulates that their selection probabilities are nearly the same as in a profile $P'$ where their approving voter sets are exactly symmetric. 
Neutrality then guarantees they are equally likely to be chosen in $P'$, and so their selection probabilities in $P$ must be similar to one another.
Moreover, combining continuity with ex-post proportionality confers a stronger form of proportionality.
For instance, consider proportionality axioms in the \jr{} family. 
If a randomized rule satisfies \jr{} ex post, then, in any outcome of the randomization, it provides a proportionality guarantee to any group of voters above a certain size threshold. However, if a group $G$ is barely below the threshold then \jr{} makes no promises; but if the rule is also continuous, then there is a high probability it provides the guarantee to $G$, since its output must be close to the output in a nearby profile under which \jr{} deterministically provides the guarantee. If $G$ supports a single candidate (or a small group of candidates), then 
these probabilistic guarantees are also directly inherited by the candidate(s).

Among our three ex-ante guarantees, continuity is also by far the most challenging to achieve. We ensure continuity through algorithmic stability, a stronger condition that guarantees continuity in the probability distribution over committees rather than over the selection probabilities. We use the term \textit{robustness} as a conceptual umbrella term capturing both continuity and stability. Stability, in turn, is of independent interest, with implications beyond our setting. We illustrate this by showcasing two applications. First, we show that our rules satisfy a stronger form of stability equivalent to \emph{differential privacy} \cite{dwork06calibrating}, a concept recently explored in voting theory \cite{mohsin22learning,david23local} and committee elections in particular \cite{li24differentially}. Intuitively, differential privacy in voting limits the information an adversary can infer about a group of voters' ballots, even given access to the  voting rule's output and all other ballots. Second, we leverage stability guarantees to design algorithms for \emph{dynamic committee voting}, where the goal is to minimize \emph{recourse}---the extent of changes made to the committee---as election data evolves over time. Dynamic settings are closely related to temporal voting, a topic of growing interest \cite{lackner2020perpetual,elkind2024temporal}.

\subsection{Our Approach and Results} 
After introducing preliminaries in \Cref{sec:prelims}, we present the following results: 

    In \Cref{sec:rando-rules}, we present our main result: a carefully constructed randomized variant of the \emph{Greedy Justified Candidate Rule (GJCR)}\footnote{
    As noted by \citet{SanchezFernandez25}, GJCR was previously known as Simple~EJR \cite{FernandezEL17}.
    We adopt the name GJCR, as its concise formulation serves as a natural basis for our enhancements.
    }
    \cite{brill23robust} satisfies all of our candidate-fairness desiderata as well as ex-post EJR+. We achieve the candidate continuity guarantee by showing that the same guarantee holds for the stronger notion of TV-stability. Our stability proof involves an analysis of the sequential random choices made by our rule, \noisygjcr{}, for which we use the technique of statistical coupling.

    In \Cref{sec:lower-bounds}, we establish a lower bound on the continuity guarantee for any rule satisfying the ex-post proportionality notion of \emph{justified representation} (JR). Intuitively, this implies that for any such rule, a single voter’s change can alter a candidate’s selection probability by up to \(\frac{k}{n}\) (committee size divided by the number of voters). While our guarantees for \noisygjcr{} do not tightly match this bound, we provide stronger lower bounds for our algorithm, suggesting that achieving significantly better guarantees requires a fundamentally different approach.

    In \Cref{sec:more-stability}, we introduce \noisygjcr{} variants that enhance continuity at the cost of relaxed proportionality or polynomial-time computability. \Cref{sec:gjcr-slack} shows that we can improve continuity a factor of \( k \) while maintaining a \((1-\frac{1}{\log k})\)-approximation of \ejrp{}. Meanwhile, \Cref{sec:noisy-greedy-JR} demonstrates that, with our approach, replacing EJR+ with JR improves continuity by at most a factor of \( \log k \). In \Cref{sec:exp-PAV}, we introduce a randomized variant of PAV that yields improved continuity for some realms of parameters while guaranteeing all other desiderata besides polynomial-time computability.
    A summary of our contributions from \Cref{sec:rando-rules,sec:more-stability} can be found in \Cref{tab:results}.

    \begin{table}[htb]
    \centering
    \resizebox{\columnwidth}{!}{
    \begin{tabular}{@{} l c cc cc @{}}
        \toprule
        Prop. & Randomized Rule & \makecell{Monotonicity\\and Neutrality} & \makecell{Poly. Time\\ Comput.} & \makecell{Continuity \\ ($\eps$-TV-stability)} & \makecell{Continuity \\ Lower Bounds} \\
        \midrule
        \midrule
        \multirow{1}{*}{\jr{} } 
        & \makecell{\noisygreedy \\ (\Cref{alg:gjcr-noisy-cap} for $\ell_{\max}=1$)} & \cmark & \cmark & $O\left(\frac{k^3}{n} \cdot \log \left(mn\right)\right)$ & 
        \multirow{1}{*}{$\Omega\left(\frac{k}{n}\right)$ } \\
        \midrule
        \multirow{3}{*}{\ejrp{} } 
        & \makecell{\noisygjcr \\ (\Cref{alg:gjcr-noisy})} & \cmark & \cmark & $O\left(\frac{k^3}{n} \log(k) \cdot \log \left(mn\right)\right)$ & 
        \multirow{3}{*}{$\Omega\left(\frac{k}{n}\right)$ } \\
        &  \makecell{\noisypav \\ (\Cref{alg:exp-PAV})}
        & \cmark & \xmark & $O\left(\frac{k^3}{n} \log (k) \cdot \log \left(m \right)\right)$ \\
        \midrule
        \multirow{1}{*}{$\alpha$-\ejrp{} } 
        & \makecell{\noisygjcrslack \\ (\Cref{alg:gjcr-noisy-slack})} & \cmark & \cmark & $O\left(\frac{k^2}{n (1-\alpha + \frac{2-\alpha}{k})} \cdot \log (k) \cdot \log (mn)\right)$ & $\Omega \left(\frac{1}{n}\right)$ \\
        \bottomrule
    \end{tabular}
    }
    \caption{Overview of our results. Continuity bounds describe the maximum change of a candidate's selection probability when some voter changes all of their approvals. Later, we present fine-grained bounds, parameterized by the number of approvals changed.
    The guarantee for \noisypav{} is stated under the assumption that $n = O(m^k)$.
    For \noisygjcrslack, note that even $\alpha = 1-1/\log k$ yields $\widetilde O\left(k^2/n\right)$ continuity.
    }
    \label{tab:results}
    \vspace*{-10pt}
\end{table}

    In \Cref{sec:dp-abc}, we consider the extent to which our rules satisfy not only additive but additive-multiplicative bounds on the change in selection probability of any subset of committees.
    This corresponds to the concept of $(\eps,\delta)$-differential privacy, and implies strong privacy guarantees for individual voters.
    We demonstrate that \noisypav{} is better able to achieve $(\eps,\delta)$-differential privacy for small $\delta$ than our greedy rules, identify lower bounds for any rule that satisfies proportionality ex post, discuss prior work on differentially private committee voting, and explore the extent to which relaxations of proportionality admit stronger differential privacy guarantees.

    In \Cref{sec:dynamic}, we demonstrate that TV-stable rules can be used to design dynamic voting rules that maintain a winning committee through a sequence of incremental changes to the approval profile, with minimal alterations to the winning committee.
    We illustrate this reduction for \noisygjcr{}, which also satisfies EJR+ ex post for every committee in an output sequence.

\subsection{Related Work}

\paragraph{Proportional representation}
We build upon \emph{(extended) justified representation ((E)JR)}, a family of proportional representation notions first proposed by \citet{aziz17justified}, which have since received significant attention (see, e.g., \cite{lackner23multi}). In particular, we study \emph{EJR+} \cite{brill23robust}, a strengthening of EJR that is satisfied by virtually all rules that satisfy EJR. Additionally, we consider relaxations of EJR+, which align with approximations of EJR \cite{halpern23representation} and approximations of core stability \cite{jiang20approximately,MSW+22a}. %

\paragraph{Monotonicity}
The literature distinguishes between two main types of monotonicity. The first, \textit{committee monotonicity}, requires that increasing the committee size \( k \) should not cause any candidate to lose their seat. Whether this property is compatible with \ejr{}(+) remains a major open problem \cite{lackner23multi}. The second type, introduced by \citet{fernandez19monotonicity}, concerns scenarios where a voter provides additional approval to a candidate (or a set of candidates). A rule is said to be \textit{candidate monotone} (\textit{support monotone}, respectively) if these additional approvals do not result in that candidate (or set of candidates) losing a seat. Here we focus on a randomized generalization of candidate monotonicity. 
The deterministic compatibility of support monotonicity with \ejr{}(+) remains open \cite[Q13]{lackner23multi}.   

\paragraph{Randomized committee voting}
Recently, several works studied randomized committee voting, thereby focusing on ex-ante proportionality, without analyzing candidate selection probabilities. \citet{suzuki2024maximum} introduce a rule satisfying a new ex-ante proportionality notion (\emph{group resource proportionality}) alongside ex-post proportionality axioms. However, their approach extends a deterministic rule fractionally when fewer than $k$ candidates are selected. Consequently, this approach does not yield any continuity guarantees. Similarly, \citet{aziz2023best} closely build on the method of equal shares~\cite{PetersS20equalshares}, which does not easily yield continuity guarantees. Moreover, \citet{cheng19group} introduce \emph{stable lotteries} for committee voting and prove their existence. The stability of \citeauthor{cheng19group} generalizes core stability to a randomized setting, and should not be confused with the distributional stability we use, which concerns a rule's robustness to small input changes. %

\paragraph{Robustness} To our knowledge, the stability of randomized approval-based committee rules, as studied here, is novel, though similar notions for deterministic rules, termed \emph{robustness}, have been explored. \citet{gawron2019robustness} showed that PAV can entirely change its outcome if a single voter approves one additional candidate. In contrast, we introduce a randomized variant of PAV in \Cref{sec:more-stability} that mitigates this issue. \citet{boehmer2023robustness} experimentally analyzed robustness in participatory budgeting, finding that proportional rules tend to be less robust. \citet{kraiczy23properties} studied a related robustness concept for a local search variant of PAV, showing that an outcome of $O(n / k^2)$-LS-PAV still satisfies \ejr{} even if an $\Omega(1 / k^2)$ fraction of voters are added or removed. Finally, our stability questions can be seen as dual to those in the smoothed analysis of voting rules \cite{xia20smoothed,flanigan23smoothed}. While smoothed analysis asks whether a given desirable property typically holds under small perturbations of a given profile $E$, we ask how similar a rule can remain under perturbations of $E$, while ensuring a desirable property for $E$ itself.

We discuss prior work on differentially private and online voting in \Cref{sec:dp-abc,sec:dynamic}.

\section{Preliminaries} \label{sec:prelims}

In the first part of this section, we introduce approval-based committee elections, randomized voting rules and our desiderata for these rules. In \Cref{subsec:prelimTV}, we introduce the concept of stability for randomized algorithms and the technique of statistical coupling. 

\subsection{Elections and Randomized Voting Rules}
Let $C$ be a set of $m$ candidates and $N$ be a set of $n$ voters. \emph{An approval set} of voter $i \in N$ is a subset $A_i \subseteq C$, which we interpret to be the set of candidates that are ``approved'' by voter $i$. 
A collection of all voters' approval sets is denoted by $A = (A_i)_{i \in N}$ and called \emph{an approval profile}.
The \emph{set of supporters} of a  candidate $c \in C$ is the set $N_c = \{ i \in N: c \in A_i\}$, i.e., the set of voters approving $c$.
We call $k\in \mathbb{N}$ a \emph{committee size} and we define the set of committees to be all subsets of size at most $k$, i.e., $\W \defeq \{W \subseteq C: |W| \leq k \}$. \footnote{While allowing committee rules to select fewer than $k$ candidates is slightly non-standard, it has been done, e.g., in the context of the method of equal shares \cite{PetersS20equalshares}. For proportional representation, this is mild because
any smaller committee can be arbitrarily extended to one of size $k$ with the same proportionality guarantee. In our case, we have to be a bit more careful to also maintain our ex-ante properties. We can maintain neutrality and continuity, for example, by 
sampling candidates uniformly at random to fill the remaining seats.}
An \emph{approval-based committee election}, or simply \emph{an election}, is a tuple $E = (C,N,A,k)$.

We study \emph{randomized voting rules} $f$ (later simply \emph{rules}), which take an election $E=(C,N,A,k)$ as input and output a random committee $W \in \mathcal{W}$.
We denote the probability distribution of a random variable $X$ by $\cD_{X}$, and let $\cD(x) \defeq \probover{X \sim \cD}{X = x}$.
Each committee distribution $\cD_{f(E)}$ induces 
candidate \emph{selection probabilities} $\pi(f(E)) = (\pi_c)_{c\in C}$, where $\pi_c \defeq \Pr[c \in f(E)]$ and $\sum_c \pi_c \leq k$.

\subsubsection{Ex-ante Properties}
We capture candidate fairness through three properties based on \emph{selection probabilities}: \emph{neutrality}, \emph{monotonicity}, and \emph{continuity}, defined below. Intuitively, neutrality ensures candidates are treated equally, independent of their identities. %

\begin{definition}[{Neutrality \cite[Section~3.1]{lackner23multi}}] \label{def:candidate-neut}
    A randomized voting rule $f$ is \emph{neutral}, if
    for every pair of elections $E = (C,N,A,k), E' = (C',N,A',k)$ that differ by renaming the candidates, i.e.,
    there exists a bijection $\sigma \colon C \to C'$ such that
    $(\sigma(A_i))_{i \in N} = (A'_i)_{i \in N}$,
    where $\sigma(A_i) = \{\sigma(c) : c \in A_i \}$,
    we have $\pi_{\sigma(c)}(f(E)) = \pi_c(f(E'))$ for all $c \in C$.
\end{definition}

\paragraph{Small Changes to the Approval Profile} \label{sec:small-approval-changes}
In order to formalize monotonicity and continuity, 
we must make the notion of `small changes in the election' precise.
There are a few choices for what it means for approval profiles $A$ and $A'$ to be close.
The first is that they differ on only one approval; that $c \not\in A_v$ and $c \in A'_v$ for some candidate $c$ and voter $v$, and $A$ and $A'$ agree on all other approvals.
The second is that $A$ and $A'$ differ on only one voter; that is, $A_v \neq A_v'$ for some voter $v$, and $A_{v'} = A_{v'}'$ for all other voters. 
We will use this first change for monotonicity, since it lets us consider a candidate $c$ receiving one additional approval in isolation.
For continuity and stability, we will use the second and let $A$ and $A'$ differ on multiple approvals of a single voter. 
More precisely, we will use $\Delta_v \defeq \abs{A_v \oplus A_v'}$ to denote the number of candidates for which $v$ changes approval, where $S \oplus T$ denotes the symmetric difference of sets. 
In pursuit of beyond-worst-case guarantees, we assume throughout that $\Delta_v \leq \Delta$ for some known $\Delta$. 
Then the single-approval-change setting corresponds to $\Delta=1$, and the setting of arbitrary single-voter changes corresponds to $\Delta = m$.

Monotonicity requires that increasing the support of a candidate never decreases its probability of inclusion in the winning committee.
If we want to differentiate this form of monotonicity from others in the literature, we also refer to it as \emph{candidate monotonicity}. 

\begin{definition}[Monotonicity] \label{def:candidate-mon}
    The randomized voting rule $f$ is \emph{monotone} if,
    for every pair of elections $E = (C,N,A,k)$ and $ E' = (C,N,A',k)$ such that $A$ and $A'$ differ on only one voter $v$ who approves one additional candidate $c^*$, i.e., $A'_v = A_v \cup \{c^*\}$, then we have 
        $\pi_{c^*}(f(E')) \geq \pi_{c^*}(f(E))$.
\end{definition}

Continuity requires that the change in the selection probabilities of any candidate is bounded after the change of one voter's approvals. If we want to differentiate this form of continuity from stability of the committee distributions, we also refer to it as \emph{candidate continuity}. 
We define continuity according to our parameter $\Delta$, but omit $\Delta$ when it is clear from context.
\begin{definition}[Continuity] \label{def:candidate-cont}
    The randomized voting rule $f$ is \emph{$L$-$\Delta$-continuous} if,
    for every committee size $k$ and for every pair of elections $E = (C,N,A,k)$ and $ E' = (C,N,A',k)$ such that $A$ and $A'$ differ on only one voter $v$ for which $\abs{A_v \oplus A_v'} \leq \Delta$, for all $c \in C$ we have 
        $\abs{\pi_c(f(E')) - \pi_c(f(E))} \leq L$.
\end{definition}

\subsubsection{Ex-post Proportionality}
We now define the standard proportionality axioms Justified Representation (\jr{}) and Extended Justified Representation+~(\ejrp{})~\cite{aziz17justified, brill23robust}.
To this end, we call a subset of the voters $N' \subseteq N$ \emph{$\beta$-large} for some $\beta \in \mathbb{R}$, if $|N'| \geq  \frac{\beta n}{k}.$

We directly introduce $\alpha$-approximate generalizations of the axioms, wherein a group must be larger by a factor of $\alpha$ in order to demand representation under the axiom; the standard definitions then correspond to $\alpha=1$. 
While there are a number of ways to relax these proportionality axioms, our choice is arguably the most common
\cite{jiang20approximately,do22online,halpern23representation}.

\begin{definition}[$\alpha$-\jr{}] \label{def:JR}
    For $\alpha \in (0,1]$, a committee $W$ provides \emph{$\alpha$-\jr{}} if there is no candidate $c \in C \setminus W$ and $\frac{1}{\alpha}$-large subset of voters $N' \subseteq N$ such that
    \[
        c \in \bigcap_{i \in N'} A_i
        \quad \text{and} \quad
        \max_{i \in N'} (\abs{A_i \cap W}) = 0.
    \]
\end{definition}

\begin{definition}[$\alpha$-\ejrp{}] \label{def:EJRp}
    For $\alpha \in (0,1]$, a committee $W$ provides \emph{$\alpha$-\ejrp{}} if
    there is no candidate $c \in C \setminus W$, $\ell \in [k]$ and some
    $\frac{\ell}{\alpha}$-large subset of voters $N' \subseteq N$ such that
    \[
        c \in \bigcap_{i \in N'} A_i
        \quad \text{and} \quad
        \max_{i \in N'} (\abs{A_i \cap W}) < \ell.
    \]
\end{definition}

We say that a randomized voting rule $f$ satisfies $\alpha$-JR ($\alpha$-EJR+, respectively) \emph{ex post}, if for every election $E$, every committee in the support of $\cD_{f(E)}$ satisfies $\alpha$-JR ($\alpha$-EJR+, respectively).

\subsection{Total Variation Distance and Stability} \label{subsec:prelimTV}

We define \emph{the total variation (TV) distance} between two probability distributions that measures the largest absolute difference between the probabilities that the distributions assign to the same event.

\begin{definition}[Total Variation Distance] \label{def:tv-dist}
    For probability distributions $\mu$, $\nu$ over $\Omega$, the total variation distance between $\mu$ and $\nu$ is $\dtv{\mu}{\nu} \defeq \max_{A \subseteq \Omega} (\mu(A) - \nu(A))$.
\end{definition}
More generally, $d_{TV}$ is a metric on the space of probability distributions over a common measurable space (identified up to sets of measure zero).
For countable $\Omega$ like ours, this admits a standard and useful characterization in terms of the $L^1$ norm (see, e.g., \citet[Prop. 4.2]{levin17markov}).
\begin{proposition}
    If $\Omega$ is countable, then $\dtv{\mu}{\nu} = \frac{1}{2} \sum_{a \in \Omega} \abs{\mu(a) - \nu(a)}$.
\end{proposition}

We present rules where small changes in the approval profile \( A \) lead to only bounded shifts in the winning committee distribution. This aligns with algorithmic stability, a key concept in differential privacy, machine learning, and dynamic algorithms \cite{beimel2022dynamic}. Following \citet{bassily21algorithmic}, who use both notions for adaptive data analysis, we begin with the weaker stability concept.
We adapt these definitions to our specific setting to accommodate beyond-worst-case guarantees parameterized by $\Delta$, our assumed upper bound on the number of approvals a single voter changes.
\begin{definition}[TV stability] \label{def:tv-stability} 
    A randomized rule $f:\cE \rightarrow \cW$ is \emph{$\eps$-$\Delta$-TV-stable} if for all $E, E' \in \cE$ such that $E$ and $E'$ differ on one voter $v$ and $\abs{A_v \oplus A_v'} \leq \Delta$, 
    \[
        \dtv{\cD_{f(E)}}{\cD_{f(E')}} \leq \eps.
    \]
\end{definition}
\begin{definition}[Max-KL stability] \label{def:kl-stability} 
    A randomized rule $f:\cE \rightarrow \cW$ is \emph{$(\eps,\delta)$-$\Delta$-max-KL stable} if for all $E, E' \in \cE$ such that $E$ and $E'$ differ on exactly one voter $v$ for which $\abs{A_v \oplus A_v'} \leq \Delta$, and for all $R \subseteq \cW$,
    \[
        \prob{f(E) \in R} \leq e^\eps \cdot \prob{f(E') \in R} + \delta.
    \]
\end{definition}
\noindent
When $\Delta = m$, max-KL stability coincides with $(\eps,\delta)$-differential privacy  ($(\eps,\delta)$-DP) \cite{dwork06calibrating}, where $\delta=0$ is known as ``pure'' differential privacy. 
We discuss DP in \Cref{sec:dp-abc}.

One key observation is that TV stability implies candidate continuity. Continuity relates to the selection probabilities, $\pi_c$, which correspond to expectations of projections of $f(E)$, and so $\abs{\pi_c' - \pi_c}$ is at most the TV distance between $\cD_{f(E)}$ and $\cD_{f(E')}$. The following can be seen as a consequence of intuitive closure property of $\eps$-TV-stability called \emph{post-processing} \citep{bassily21algorithmic}; we also provide a direct proof in the appendix.
\begin{restatable}{proposition}{etvdistanceToContinuity} \label{obs:marginalseasier}
    If a rule $f$ is $\eps$-$\Delta$-TV-stable, then $f$ is $\eps$-$\Delta$-continuous.
\end{restatable}

Lastly, we introduce statistical coupling, which we use in our stability proofs and in \Cref{sec:dp-abc}.

\begin{definition}[Coupling] \label{def:coupling}
    Consider probability distributions $\mu$, $\nu$ over $\Omega$. A probability distribution $\pi$ over $\Omega \times \Omega$ is a \emph{coupling} of $\mu$ and $\nu$ if it preserves the marginals of $\mu$ and $\nu$; that is, if for $(X, Y) \sim \pi$ the distribution of $X$ is $\mu$ and the distribution of $Y$ is $\nu$.
\end{definition}

\begin{lemma}[Coupling Lemma \cite{aldous1983random}] \label{lem:coupling}
For probability distributions $\mu$, $\nu$ over a ground set $\Omega$ and a coupling $\pi$ of them,
    \[
        \probability_{(X,Y) \sim \pi}[X \neq Y] \geq \dtv{\mu}{\nu}, \text{ and there exists a coupling } \pi^* \text{ with }
        \probability_{(X,Y) \sim \pi^*}[X \neq Y] = \dtv{\mu}{\nu}.\]
\end{lemma}

\section{A Proportional and Candidate-Fair Voting Rule}\label{sec:rando-rules}

In this section, we present a randomized voting rule that is both ex-post proportional and ex-ante candidate-fair. 
Before presenting our rule, we explore natural approaches to designing such a rule and demonstrate why they fall short of meeting at least one of our desiderata.

\paragraph{Uniform Selection} A natural idea is to sample a committee uniformly at random from the set of all \ejrp{} committees. We call this voting rule \textsc{Uniform\ejrp{}}. Indeed, this rule is monotone. All missing proofs can be found in the appendix.

\newcommand{\Wf}{\ensuremath{\mathcal{E}}}
\newcommand{\Ws}{\ensuremath{\mathcal{E}'}}
\newcommand{\Wc}{\ensuremath{\mathcal{W}_c}}
\newcommand{\uejr}{\textsc{Uniform\ejrp{}}\xspace}

\begin{restatable}{theorem}{thmuniformEJR}\label{thm:uniformEJR}
    \textsc{Uniform\ejrp{}} is ex-ante monotone. 
\end{restatable}

However, \uejr is far from being continuous, as the following example illustrates. 

\begin{example} \label{ex:continuity}
    Let $n, k \in \mathbb{N}$ be such that $n$ is divisible by $k$. We divide the voters into $k$ equal-sized groups $N = N_1 \cup \dots \cup N_k$. The set of candidates is $C = \{c_1, \dots, c_k\} \cup \{d_1, \dots, d_{m-k}\}$, where for each $i \in [k]$, the candidate $c_i$ is approved exactly by those voters in $N_i$. All candidates $d_i, i \in [m-k]$ are dummy candidates and are not approved. We call the resulting election $E$. There exists exactly one committee satisfying \ejrp{}, namely, the committee $\{c_1, \dots, c_k\}$. Thus, the selection probabilities for \uejr for the election $E$ are $1$ for all candidates in $\{c_1, \dots, c_k\}$ and $0$ for all dummy candidates.

    
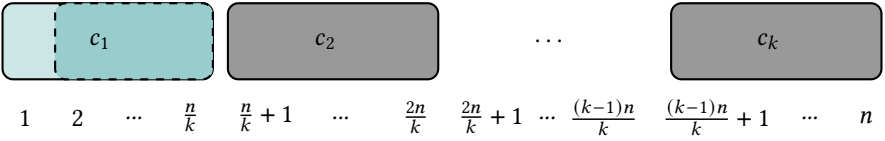
\vspace{.2cm}

    \begin{center}
    \begin{tikzpicture}
    \node at (0.8,0) {\(1\)};
    \node at (1.5,0) {\(2\)};
    \node at (2.25,0) {...};
    \node at (3,0) {\(\frac{n}{k}\)};
    \node at (4,0) {\(\frac{n}{k} + 1\)};
    \node at (5,0) {...};
    \node at (6,0) {\(\frac{2n}{k}\)};
    \node at (7,0) {\(\frac{2n}{k} + 1\)};
    \node at (7.75,0) {...};
    \node at (8.5,0) {\(\frac{(k-1)n}{k}\)};
    \node at (10,0) {\(\frac{(k-1)n}{k} + 1\)};
    \node at (11.245,0) {...};
    \node at (12,0) {\(n\)};

    \draw[thick, rounded corners, fill=teal!20] (0.5,0.5) rectangle (3.3,1.5);
    \draw[thick, dashed, rounded corners, fill=teal!40] (1.2,0.5) rectangle (3.3,1.5);
    \node at (1.8,1){$c_1$};
    \draw[thick, rounded corners,fill=black!40] (3.5,0.5) rectangle (6.3,1.5);
    \node at (4.8,1){$c_2$};
    \node at (7.8,1){$\dots$};
    \draw[thick, rounded corners,fill=black!40] (9.4,0.5) rectangle (12.2,1.5);
    \node at (10.7,1){$c_k$};
\end{tikzpicture}
    \end{center}

    Now, we construct the election $E'$ by deleting the approval of voter $1$ for candidate $c_1$. The set of \ejrp{} committees consists of any committee $W$ with $\{c_2, \dots, c_k\} \subseteq W$. Thus, the selection probability for $c_1$ drops from $1$ to $\frac{1}{m-k+2}$.
\end{example}

\Cref{ex:continuity} demonstrates that uniformly randomizing over \ejrp{} committees violates continuity and that the same issue arises with weaker proportionality notions like \jr{}. It also shows that a continuous algorithm must distinguish between nearly crucial and non-crucial candidates for \ejrp{}.

A natural approach is to randomize over \ejrp{} committees, weighting probabilities by a suitable measure of proportionality. In \Cref{sec:exp-PAV}, we show that this idea indeed gives rise to an ex-ante monotone and continuous rule, however, this rule is not polynomial-time.

We explore a second approach: introducing randomness in deterministic polynomial-time committee rules that satisfy \ejrp{}. A natural candidate would be the \emph{method of equal shares (MES)} \cite{PetersS20equalshares}, however, since MES is very sensitive to small changes in the input, we are so far unable to design a continuous randomized variant of it. %
A second candidate would be a local search variant of PAV; however, this has been shown to even violate candidate monotonicity \cite[Proposition 12]{kraiczy23properties}. In contrast, we will demonstrate that the third contender, the greedy justified candidate rule of \citet{brill23robust}, indeed allows for a carefully crafted randomized variant that we show to satisfy all of our desiderata.

\subsection{Softmax Greedy Justified Candidate Rule (\noisygjcr)} \label{sec:noisy-GJCR}

\newcommand{\scl}{\ensuremath{n_{c\ell}}}
\newcommand{\sclp}{\ensuremath{n'_{c\ell}}}
\newcommand{\Scl}{\ensuremath{N_{c\ell}}}
\newcommand{\scsl}{\ensuremath{n_{c^*\ell}}}
\newcommand{\scslp}{\ensuremath{n'_{c^*\ell}}}
\newcommand{\sdl}{\ensuremath{n_{d\ell}}}
\newcommand{\sdlp}{\ensuremath{n'_{d\ell}}}

The so-called \emph{Greedy Justified Candidate Rule (\gjcr{})} is a simple and elegant algorithm that returns \ejrp{} committees. The algorithm starts with an empty committee and then iterates over $k$ epochs. For each epoch $\ell \in [k]$, the algorithm checks whether there exists a candidate that witnesses an \ejrp{} violation for some $\ell$-large group of voters. If so, the algorithm adds such a candidate to the committee. %
To formalize GJCR, we define for a (current) committee $W$ 
\begin{equation} \label{eq:gjcr-deltajell-defn}
    \Scl \defeq \{i \in N_c \colon \abs{A_i \cap W} < \ell\}\qquad \text{ and } \qquad \scl \defeq |\Scl|, 
\end{equation}
capturing the number supporters of candidate $c$ who are less than $\ell$-represented by $W$.

\begin{algorithm}
\caption{Greedy Justified Candidate Rule (\gjcr{}) \cite{brill23robust}}
\label{alg:gjcr}
$W \gets \emptyset$\;
\For{$\ell$ in $(k, k-1, \dots, 1)$}{
    \While{there exists $c\in C \setminus W$ with $\scl \geq \frac{\ell n}{k}$ \label{line:GJCR-thresh}}{
        $W \gets W \cup \{c\}$ \;
    }
}
\Return $W$\;
\end{algorithm}

We modify the algorithm in two ways. To ensure neutrality, our algorithm randomizes over all witnesses instead of picking an arbitrary one. 
To achieve continuity, the algorithm must interpolate between cases where a candidate is and is not a witness. To address this, we observe that \gjcr has slack: as long as the threshold $\frac{\ell n}{k}$ in \cref{line:GJCR-thresh} remains strictly greater than $\frac{\ell n}{k+1}$, the returned committee satisfies \ejrp{}. This follows from \citet[Proposition 7]{brill23robust} and our proof of \Cref{thm:gjcr-noisy-correct}. 

We refer to candidates $c$ with $\scl \geq \frac{n \ell}{k}$
as true \emph{witnesses} and to candidates $c$ with $\scl \in (\frac{n\ell}{k+1},\frac{n\ell}{k})$
as \emph{quasi-witnesses}. While witnesses have to be included unless another candidate’s inclusion nullifies their witness status, quasi-witnesses do not directly witness an \ejrp{} violation, but including them does not hinder \ejrp{} compliance. Based on this, we refine \Cref{alg:gjcr} by sampling from the set of quasi-witnesses and witnesses, 
weighted with help of the following function
\begin{equation} \label{eq:gjcr-g-defn}
    g_\ell(x) \defeq e^{a_\ell x},
\end{equation}
for factors $a_\ell$ for $\ell \in [k]$ to be specified later. Any positive and increasing $g_\ell$ suffices for \Cref{alg:gjcr-noisy} to satisfy \ejrp{} (\Cref{thm:gjcr-noisy-correct}) and monotonicity (\Cref{lem:gjcr-mon-next}); $a_\ell$ will be chosen to optimize our continuity guarantee (\Cref{lem:noisy-GJCR-sequence-stable}). 
We are now ready to formalize \noisygjcr{} in \Cref{alg:gjcr-noisy}. We write $c = \bot$ to indicate that no candidate is selected. This may transpire if the set of {(quasi-)} witnesses (indicated by $\LW$) is empty, or if $\LW \neq \emptyset$ but the sum of the $g_\ell(\scl)$ does not reach $g_\ell(\frac{n\ell}{k})$; this can happen if 
no true witnesses exist.

\begin{algorithm}
\caption{\noisygjcr{}}
\label{alg:gjcr-noisy}
$W \gets \emptyset$\;
\For{epoch $\ell$ in $(k, k-1, \dots, 1)$}{ \label{line:alg:noisy-gjcr:outer-loop}
    \For{round $r \in [k]$} { \label{line:GJCR-noisy-inner-loop}
        $\LW \leftarrow \left\{c \in C \setminus W: \scl > \frac{n \ell}{k+1} \right\}$\;\label{line:alg:noisy-gjcr:defining-l} 
        Sample $c\sim \LW$ with probability $P_c \defeq \frac{g_\ell(\scl)}{\max\left(\sum_{d \in \LW} g_\ell(\sdl), \: g_\ell(\frac{n\ell}{k})\right)}$, and $c \leftarrow \bot$ otherwise\; \label{line:GJCR-noisy-sample}
        \If{$c \neq \bot$}{
            $W \gets W \cup \{c\}$\;
        }
    }
}
\Return $W$\;
\end{algorithm}

Aside from differences in thresholds and deterministic vs. randomized sampling, a difference between \Cref{alg:gjcr} and \Cref{alg:gjcr-noisy} is that for the former, the inner loop is a ``while'' loop, while the latter employs a loop that takes $k$ rounds. This is because \Cref{alg:gjcr-noisy} may continue to select candidates even though no witness remains. 

We show that \noisygjcr{} indeed satisfies \ejrp{}. The proof is a modification of the proof of \citet[Proposition 7]{brill23robust}.

\begin{restatable}{theorem}{gcjrEJRplus} \label{thm:gjcr-noisy-correct}
    \noisygjcr{} satisfies ex-post \ejrp{}.
\end{restatable}

We now move on to proving our ex-ante properties for \noisygjcr{}.

\subsection{\noisygjcr{} is Neutral and Monotone}
We first observe that \noisygjcr{} is neutral, which follows immediately from its definition as all decisions made by the algorithm depend only on the approval sets of the voters. %

\begin{observation}
    \noisygjcr{} is ex-ante neutral. 
\end{observation}

We now turn towards proving ex-ante monotonicity of \noisygjcr{}, which requires more effort. 
Let $E$ and $E'$ be two elections that differ only in the ballot of one voter, call them $v$, who approves one additional candidate $c^*$ in $E'$.
We aim to show that the selection probability for $c^*$ is at least as high in $E'$ as in $E$. Our proof proceeds via an analysis of the intermediate states of \Cref{alg:gjcr-noisy}. 
First, we ``flatten'' the two for-loops by defining for a given epoch $\ell$ and round $r$ the index $T\defeq k(k-\ell) + r$. Then, we define $\alg(E,T)$ to be the sequence of the first $T$ samples of \noisygjcr{}, done in \cref{line:GJCR-noisy-sample}. Note that this sequence typically contains multiple $\bot$. For any sequence $s$, we also define the corresponding (partial) committee by $W_s \defeq \{c \in s : c \neq \bot\}$. Lastly, we define $\alg(E)$ to be the sequence of length $k^2$ of the algorithm's selections upon its termination. In particular, $\alg(E)_{T}$ is then the element selected at iteration $T$. 

We can now formulate the following lemma, which states that under the assumption that \noisygjcr{} has made the same decisions for election $E$ as for election $E'$ up to some point, monotonicity for $c^*$ holds for the next sample of the algorithm. The proof carefully argues about the potential changes of the witness set $\mathcal{L}$, the variable $\scsl$, and the denominator in \Cref{line:GJCR-noisy-sample}.

\begin{restatable}{lemma}{lemgcjrmonnext} \label{lem:gjcr-mon-next}
    For any $T\in \{0, \ldots, k^2-1\}$ and any sequence $s$ of length $T$ with $c^* \not\in s$, it holds that
    \begin{align}
        \prob{d = \alg(E')_{T+1}\given \alg(E',T)=s} &\leq \prob{d = \alg(E)_{T+1}\given \alg(E,T)=s} \text{for all } d \neq c^*, \notag \\
        \prob{\bot = \alg(E')_{T+1}\given \alg(E',T)=s} &\leq \prob{\bot = \alg(E)_{T+1}\given \alg(E,T)=s}, \text{and} \notag \\
        \prob{c^* = \alg(E')_{T+1}\given \alg(E',T)=s} &\geq \prob{c^* = \alg(E)_{T+1}\given \alg(E,T)=s}. \notag 
    \end{align}
\end{restatable}

Applying \Cref{lem:gjcr-mon-next} inductively over the decision tree of \noisygjcr{} proves it is monotone.

\begin{restatable}{theorem}{gcjrmonotone} \label{thm:gjcr-mon}
    \noisygjcr{} is ex-ante monotone. 
\end{restatable}

We now turn to continuity guarantees for \noisygjcr{}. 

\subsection{\noisygjcr{} is Continuous} \label{sec:gjcr-stable}

Now, let $E$ and $E'$ be two elections that differ in the approvals of voter $v$, and parameterize the difference by $\Delta_v \defeq \abs{A_v \oplus A'_v}$. Consider a sequence $s$ of length $T$. We define $\cN_s$ to be the \emph{next candidate distribution}, i.e. the probability distribution over candidates $c \in C$, defined by 
\[
    \cN_s(c) = \prob{c = \alg(E)_{T+1} \given \alg(E,T) = s}, \text{ and}
\] 
define $\cN_s'$ analogously for election $E'$. The next lemma bounds the TV distance of $\cN_s$ and $\cN'_s$.

\begin{lemma} \label{lem:gjcr-stability-per-step}
    Let $E$ and $E'$ be two elections that differ only in the ballot of voter $v$, and let $\cN_s, \cN_s'$ and $\Delta_v$ be defined as described above. 
    For any partial sequence of samples $s$ in \Cref{alg:gjcr-noisy}, it holds that  
    \begin{equation} \label{eq:gjcr-noisy-round-continuity}
        \dtv{\cN_s}{\cN_s'} \leq \gamma \cdot a_\ell + 4 \cdot \exp\left(\log(\Delta_v) -a_{\ell}\frac{n\ell}{k(k+1)} + a_{\ell} \right).
    \end{equation}
    for some universal constant $\gamma$ (in particular, $\gamma = 4(2e-1)$).
\end{lemma}

\noindent
This proof (and others) invoke a particular consequence of the triangle inequality:
\begin{observation} \label{fac:quotientdiff}
    For any $a, a' \geq 0$ and any $b, b' > 0$, it holds that $\abs{\frac{a'}{b'} - \frac{a}{b}} \leq a' \abs{\frac{1}{b'} - \frac{1}{b}} + \frac{1}{b} \abs{a' - a}$.
\end{observation}

\begin{proof}[Proof of \Cref{lem:gjcr-stability-per-step}]
    We begin with an observation about the regime of $a_\ell$ for which this statement requires proof.
    \begin{observation} \label{obs:a-ell-small-wlog}
        If $a_\ell \geq 1$, then \Cref{lem:gjcr-stability-per-step} is trivially satisfied (e.g., with $\gamma=1$).
    \end{observation}
    This holds because $\dtv{\cdot}{\cdot} \leq 1$. Therefore, we will assume that $a_{\ell} \leq 1$ holds in the following.

    We first establish a pair of useful bounds and notation. In the following we assume that the sequence $s$ of length $T$ is fixed and all of the following notation is under the assumption that the algorithm is in step $T+1$ and $\alg(E,T) = \alg(E',T) = s$, e.g., we denote by $\scl$ ($\sclp$, respectively) the variable in step $T+1$ of \Cref{alg:gjcr-noisy} when running it for input $E$ ($E'$, respectively). 
    Then, 
    \begin{equation} \label{eq:gjcr-delta-vals-general}
        \scl - 1 \leq \sclp \leq \scl + 1, 
    \end{equation}
    which holds because the only voter that can change membership in $\Scl$ when going from $E$ to $E'$ is voter $v$. 
    From \eqref{eq:gjcr-delta-vals-general}, the definition of $g_\ell$, and $a_\ell \leq 1$ (\Cref{obs:a-ell-small-wlog}) we have that for $c\in C$,
    \begin{equation} \label{eq:gjcr-gdelta-bound}
        \abs{g_\ell(\sclp)- g_\ell(\scl)} \leq  g_\ell(\scl) \cdot \max(1 - e^{-a_\ell}, \:e^{a_\ell} - 1)  \leq g_\ell(\scl) \cdot a_\ell \cdot (e-1).
    \end{equation}
    We will also apply the fact that 
    \begin{equation}
    |\LW \oplus \LW'| \leq \Delta_v, \label{eq:delta}
    \end{equation}
     which is due to the fact that $c \in \LW \oplus \LW'$ implies $\scl \neq \sclp$, which in turn implies that voter $v$ votes differently for $c$ in the two elections $E$ and $E'$. For notational convenience, we define 
    \begin{equation}
    G \defeq \max\left(\sum_{d \in \LW} g_\ell(\sdl), \: g_\ell\left(\frac{n\ell}{k}\right)\right)
    \end{equation}
    to be the denominator of $P_c$ in \cref{line:GJCR-noisy-sample}, and define $G'$ analogously for $\sdlp$ and $\LW'$. We observe that \begin{equation}
        |G' - G| \leq \sum_{c \in \LW' \cap \LW} \abs{g_\ell(\sclp) - g_\ell(\scl)} + \sum_{c \in \LW' \oplus \LW} \max(g_\ell(\sclp),g_\ell(\scl)). \label{eq:Gdifference}
    \end{equation}

    {We are now ready to bound the difference between $\cN_s$ and $\cN_s'$. Since $P_\bot = 1 - \sum_{c \in C} P_c$, to bound $\dtv{\cN_s}{\cN_s'}$ we may incur a factor of two and focus on $c\in C$. Using this, we have

    \begin{align}
        \frac{1}{2}\dtv{\cN_s}{\cN_s'}
        &\leq  \sum_{c \in C} \abs{P_c' - P_c} %
        =  \sum_{c \in \LW' \cup \LW} \abs{P_c' - P_c} \notag \\
        &=  \sum_{c \in \LW' \cap \LW} \abs{\frac{g_\ell(\sclp)}{G'} - \frac{g_\ell(\scl)}{G}} +  \sum_{c \in \LW' \oplus \LW} \max(P_c', \: P_c). \notag
        \intertext{We now apply \Cref{fac:quotientdiff} to the first sum, which yields} 
        &\leq \sum_{c \in \LW' \cap \LW} g_\ell(\sdlp) \abs{\frac{1}{G'} - \frac{1}{G}} + \frac{1}{G} \sum_{c \in \LW' \cap \LW} \abs{g_\ell(\sclp) - g_\ell(\scl)} +  \sum_{c \in \LW' \oplus \LW} \max(P_c', \: P_c). \notag \\ 
        \intertext{For the next inequality, we rewrite the first sum as $\sum_{c \in \LW' \cap \LW} \frac{g_\ell(\sclp)}{G'} \frac{\abs{G' - G}}{G}$, which is at most $ \frac{\abs{G' - G}}{G}$ since $\sum_{c \in \LW' \cap \LW} g_\ell(\sclp) \leq G'$. We also upper bound the second sum, which will become helpful later on. We get }
        &\leq \frac{\abs{G' - G}}{G} + \frac{1}{G} \sum_{c \in \LW' \cap \LW} \abs{g_\ell(\sclp) - g_\ell(\scl)} + \sum_{c \in \LW' \oplus \LW} \frac{\max(g_\ell(\scl),g_\ell(\sclp))}{\min(G,G')}. \notag \\ 
        \intertext{Next, we apply our previously established upper bound on $|G'-G|$ from \eqref{eq:Gdifference} and directly combine these terms with the right-hand-side of the expression.}
        &\leq 2\left(\frac{1}{G} \cdot \sum_{c \in \LW' \cap \LW} \abs{g_\ell(\sclp) - g_\ell(\scl)} + \sum_{c \in \LW' \oplus \LW} \frac{\max(g_\ell(\scl),g_\ell(\sclp))}{\min(G,G')} \right). \notag \\ 
        \intertext{Then, we apply \eqref{eq:gjcr-gdelta-bound} to the first sum. Moreover, we upper bound each summand of the second sum by the same upper bound, using the following  observation: When $c \in \LW' \oplus \LW$, then $\max(\scl,\sclp) \leq \frac{n\ell}{k+1} + 1$, which holds since \scl{} and \sclp{} can only differ by $1$. Also, $\min(G,G') \geq g_\ell(\frac{n\ell}{k})$. Lastly, by \eqref{eq:delta}, we know that the number of summands is bounded by $|\LW'\oplus\LW| \leq \Delta_v$. This yields }
        &\leq 2 \left((e-1) \cdot a_\ell \cdot \sum_{c \in \LW'\cap \LW} \frac{g_{\ell}(\scl)}{G} + \Delta_v \frac{g_{\ell}\left(\frac{n\ell}{k+1}+ 1\right) }{g_{\ell}(\frac{n\ell}{k})} \right), \notag \\ 
        \intertext{where we can upper bound the sum by $1$ again and rewrite the right-hand-side by applying the definition of $g_{\ell}$. This yields}
        & \leq 2\left( (e-1) \cdot a_{\ell} + \exp\left(\log(\Delta_v) - a_{\ell} \frac{n\ell}{k(k+1)} + a_{\ell} \right) \right), \notag
    \end{align}    
    which, after multiplying the entire inequality by the factor of $2$, yields the claim.}
    \end{proof}

We now discuss the optimal choice of $a_\ell$. 
We have to carefully balance 
the terms in \Cref{lem:gjcr-stability-per-step} to provide the smallest continuity upper bound within each step. Larger choices of $a_\ell$ increase the first term but decrease the second, and vise versa. Within epoch $\ell$ of \Cref{alg:gjcr-noisy}, we choose
\begin{equation} \label{eq:gjcr-a-defn}
    a_\ell \defeq \frac{k(k+1)}{n\ell}\log\left(n \cdot \Delta \right),
\end{equation} 
where $\Delta$ is our assumed upper bound on $\Delta_v$ for all single-voter approval changes.

Recall that by \Cref{obs:a-ell-small-wlog}, the lemma is trivially true for all $a_\ell \geq 1$. Therefore, we assume in the following that $a_\ell \leq 1$. The second term of \eqref{eq:gjcr-noisy-round-continuity} then becomes  
\begin{align*}
    \exp\left(\log(\Delta_v) -a_\ell\frac{n\ell}{k(k+1)} + a_\ell \right) & \leq \exp \left(\log(\Delta) -a_\ell\frac{n\ell}{k(k+1)} + a_\ell \right) \\
    &=\exp \left( \log(\Delta) - \log\left(n \cdot \Delta \right) + a_\ell \right) \\
    &= \exp\left(\log\left(\frac{1}{n} \right) + a_\ell \right) = \frac{1}{n} \cdot e^{a_\ell}\\
    & \leq \frac{e}{n} \leq a_\ell \cdot e.
\end{align*}
Plugging this into \Cref{lem:gjcr-stability-per-step} yields a bound of $4(2e -1) a_{\ell}$  on the TV-distance of $\cN_s$ and $\cN'_s$.

 Now, for some election $E$, let $\cS_E$ be the distribution over the random (complete) sequence $\alg(E)$.

\begin{restatable}{theorem}{gcjrsequencestable} \label{lem:noisy-GJCR-sequence-stable}
The sequence distribution induced by \noisygjcr{} for $a_{\ell}$ defined in \Cref{eq:gjcr-a-defn} satisfies 
$\eps$-$\Delta$-TV-stability for
\[
    \eps = \gamma \cdot \frac{k^3}{n} \cdot \log k \cdot \left(\log n + \log \Delta \right),\qquad \text{for some universal constant }\gamma.
\]
\end{restatable}

\begin{proof}[Proof Sketch of \Cref{lem:noisy-GJCR-sequence-stable}]
    We translate the per-step upper bound from \Cref{lem:gjcr-stability-per-step}, which is $\mathcal{O}(a_{\ell})$ for our choice of $a_\ell$ into an upper bound for $\dtv{\mathcal{S}_E}{\mathcal{S}_{E'}}$. 
    The per-step upper bound implies that for each $s$ and given that the instantiations of the algorithm with inputs $E$ and $E'$ have both arrived in state $s$, $\mathcal{O}(a_\ell)$ is an upper bound for the probability that the sequences will diverge at this point. 
    In order to upper bound the overall probability that the sequences diverge at some point, we union bound over the $k^2$ steps of \noisygjcr{}. 
    To make this argumentation precise we deploy statistical coupling, which is introduced in \Cref{sec:prelims}.
\end{proof}

Before formalizing the direct implications of \Cref{lem:noisy-GJCR-sequence-stable}, we interpret the provided bound. Recall that $\Delta$ measures how much approvals a single voter can change in the definition of TV-stability;
for the model in which only one approval changes, we have $\eps = O(\frac{k^3}{n} \log k \log n)$, while if $m$ approvals may change then we have $\eps = O(\frac{k^3}{n} \log k \log(mn))$. We state two corollaries which are direct consequences of \Cref{lem:noisy-GJCR-sequence-stable}. First, let $f$ be \noisygjcr{}. It might not be surprising that bounding the TV distance of the sequence distribution has direct implications for the TV distance of the committee distribution $\mathcal{D}_{f(E)}$ and the continuity of the selection probabilities. 
We get:
\begin{corollary} \label{lem:noisy-GJCR-stable}
The committee distribution of \noisygjcr{} for $a_{\ell}$ defined in \Cref{eq:gjcr-a-defn}
    is $\eps$-$\Delta$-TV-stable for $\eps = \gamma \cdot \frac{k^3}{n} \cdot \log k \cdot \left(\log n + \log \Delta \right)$, for some constant $\gamma$.
\end{corollary}

Lastly, as a direct consequence of \Cref{lem:noisy-GJCR-stable}, we have by \Cref{obs:marginalseasier} that
\begin{corollary} \label{thm:noisy-GJCR-continuous}
     The voting rule \noisygjcr{} for $a_{\ell}$ defined in \Cref{eq:gjcr-a-defn} is $L$-$\Delta$-continuous for $L= \gamma \cdot \frac{k^3}{n} \cdot \log k \cdot \left(\log n + \log \Delta \right)$, for some constant $\gamma$.
\end{corollary}

\subsection{Continuity and Stability Lower Bounds} \label{sec:lower-bounds}

We establish lower bounds on the continuity of any proportional rule, as well as on our own rules specifically.
For establishing the former bound, the key idea is to construct an election with $k+1$ candidates, each supported by $\frac{n}{k} - \frac{n}{k^2} < \frac{n}{k+1}$ voters.
For any rule, at least one candidate has a selection probability below \( \frac{k}{k+1} \).
By increasing that candidates' approvals to \( \frac{n}{k} \), a proportional rule must select them due to consistency with \jr{}, causing their probability to rise by \( \Omega(\frac{1}{k}) \) with just \( \frac{n}{k^2} \) additional approvals---thus proving the lower bound.

\begin{restatable}{theorem}{JRcandcontLB}
\label{thm:JR-cand-cont-lb}
    The continuity of any rule satisfying ex-post \jr{} is $\Omega\left(\frac{k}{n}\right)$, even when $\Delta = 1$. 
\end{restatable}

The lower bound from \Cref{thm:JR-cand-cont-lb} implies the same lower bound for any class of rules which is a superset of rules satisfying \jr{}, e.g., rules satisfying \pjr{}, \ejr{}, \ejrp{}, core stability or priceability \cite[Figure 4.1]{lackner23multi}.
In particular, our rules
satisfy \jr{} ex post,
but we can also directly establish a stronger continuity lower bound for them.
The approach is similar to above, but this time we construct an election with $k+1$ \emph{pairs} of candidates.
This forces our rule to split probability \( \frac{k}{(k+1)} \) among a pair of candidates, hence some candidate has probability at most $0.5$.
After altering the election by \( \frac{n}{k^2} \) approvals, the probability of that candidate has to increase by $0.5$.

\begin{restatable}{theorem}{noisyContLB}
\label{thm:noisy-rules-cand-cont-lb}
    The continuity of \noisygjcr{}
    is $\Omega\left(\frac{k^2}{n}\right)$, even when $\Delta = 1$. 
\end{restatable}
We remark that the same lower bound holds for \noisygjcrcap{}, an algorithm we introduce in the next section. Because $\eps$-TV stability implies $\eps$-continuity by \Cref{obs:marginalseasier}, our continuity lower bounds imply lower bounds for $\eps$-TV stability.

\section{Improved Continuity via Relaxations} \label{sec:more-stability}

In \Cref{sec:rando-rules}, we saw how to attain $\Tilde{O}(\frac{k^3}{n})$-continuity in conjunction with \ejrp{}. We also showed that no rule satisfying ex-post JR can achieve better than \(\Omega(\frac{k}{n})\)-continuity. Motivated by this gap, we explore ways to improve continuity guarantees by relaxing some of our other desiderata.

First, in \Cref{sec:gjcr-slack}, we present a simple adaptation of \noisygjcr{} that achieves \(\alpha\)-EJR+ while improving continuity guarantees. For instance, setting \(\alpha = 1 - 1/\log(k)\) leads to \(\tilde{O}(\frac{k^2}{n})\)-continuity. Second, we explore the spectrum between \jr{} and \ejrp{}, as introduced in \Cref{sec:prelims}. Here, we propose another adaptation of \noisygjcr{} that achieves -- for example -- a continuity improvement by a factor of $\log(k)$ when relaxing \ejrp{} to \jr{}. Finally, we pick up an idea discussed at the beginning of \Cref{sec:rando-rules}, namely, to randomize over all \ejrp{} committees weighted by a measure of proportionality (in our case the \emph{PAV} score). This yields a continuity guarantee that improves upon the one presented for \noisygjcr{} when the number of approvals changed by the voter (denoted by $\Delta$) is small. On the other hand, this algorithm is not polynomial-time.

A summary of our results and their corresponding bounds can be found in \Cref{tab:results}. To stay close to our motivation of candidate fairness, we state all of our stability results in this section in terms of candidate continuity. However, all of these results are achieved via proving the corresponding TV-stability bounds, which we state in \Cref{app:more-stability}.

\subsection{\texorpdfstring{A Relaxed \noisygjcr{} for $\alpha$-\ejrp{}}{A Relaxed \noisygjcr{} for a-\ejrp{}}}\label{sec:gjcr-slack}

We introduce a variant of \noisygjcr{}, parameterized by \(\alpha \in (0,1]\), that satisfies \(\alpha\)-\ejrp{}. The only modification compared to \Cref{alg:gjcr-noisy} is in \Cref{line:GJCR-noisy-sample}, where we replace the probability with \[P_c \defeq \frac{g_\ell(\scl)}{\max\left(\sum_{d \in \LW} g_\ell(\sdl), \: g_\ell\left(\frac{n\ell}{\alpha k}\right)\right)},\]  
where the difference is the inclusion of \(\alpha\) in the denominator. We refer to this variant as \noisygjcrslack{} and provide its formal definition in \Cref{app:gjcr-slack}. The following theorem follows from arguments that are closely parallel to those used for \noisygjcr{}.  

\begin{restatable}{theorem}{gcjrSlackThm} \label{thm:gjcr-slack-correct}
    \noisygjcrslack{}{} for $\alpha \in (0,1]$ satisfies 
    ex-post $\alpha$-\ejrp{}, ex-ante neutrality, and ex-ante monotonicity.
\end{restatable}

Increasing the denominator in $P_c$ helps us to interpolate between situations where we have few witnesses for \ejrp{} violations (who thus receive high selection probabilities), and situations where none of these candidates are witnesses. 

\begin{restatable}{theorem}{thmgjcrslackcontinuous} \label{thm:noisy-GJCR-slack-continuous}
    \noisygjcrslack{} for $\alpha \in (0,1]$ and some $a_{\ell}$ (defined in the appendix) is $L$-$\Delta$-continuous for some constant $\gamma$ and
    \[
        L= \gamma \cdot \frac{k^2}{n (1-\alpha + \frac{2-\alpha}{k})} \cdot (\log k) \cdot \left(\log n + \log \Delta \right).
    \]
\end{restatable}

For $\alpha = 1$, we recover our bound from \Cref{thm:noisy-GJCR-continuous} and when choosing $\alpha \leq 1 - 1/\log k$ this implies $\tilde O \left(\frac{k^2}{n}\right)$-continuity.

\subsection{A Relaxed \noisygjcr{} for JR} \label{sec:noisy-greedy-JR}

In the definition of \ejrp{} we have a collection of constraints, one for every $\ell \in [k]$. Our second approach to relaxing \ejrp{} is by asking for representation for groups that are $\ell$-large for $\ell \leq \ell_{max}$.

\begin{definition}[$\ell_{\max}$-capped \ejrp{}] \label{def:ejrp-capped}
    A committee $W$ provides \emph{$\ell_{\max}$-capped $\alpha$-\ejrp{}} if
    for every $\ell \in [\ell_{\max}]$
    there is no candidate $c \in C \setminus W$ and an 
    $\ell$-large subset of voters $N' \subseteq N$,
    such that
    \[
        c \in \bigcap_{i \in N'} A_i
        \quad \text{and} \quad
        \forall_{i \in N'} \; |A_i \cap W| < \ell.
    \]
\end{definition}

When $\ell_{\text{max}} = k$, we recover the definition for EJR+ and when $\ell_{\text{max}} = 1$, we recover the definition of JR. There exists a simple adaptation of \noisygjcr{} which satisfies $\ell_{\max}$-capped EJR+, namely, we simply start the algorithm in epoch $\ell_{\text{max}}$ instead of $k$. We refer to this algorithm as \noisygjcrcap{} and formalize it in \Cref{app:gjcr-capped}. In particular, we refer to \noisygjcrcap{} with $\lm=1$ as \noisygreedy{}, since it provides \jr{}.\footnote{\noisygreedy{} can be seen as a randomized variant of an algorithm known as \textsc{seqCC} or \textsc{GreedyCC}~\cite{lackner23multi,ElkindFIMSS24}.}

\begin{restatable}{theorem}{gcjrCappedThm}
    \noisygjcrcap{} for $\lm \in [k]$ satisfies 
    ex-post \lm-capped \ejrp{}, ex-ante neutrality, and ex-ante monotonicity.
\end{restatable}

Here, the continuity benefit is more subtle; relaxing \ejrp{} to JR yields an $\log(k)$ improvement. 

\begin{restatable}{theorem}{thmgjcrcappedcontinuous} \label{thm:noisy-GJCR-capped-continuous}
      \noisygjcrcap{} for $\lm \in [k]$ and $a_{\ell}$ defined in \Cref{eq:gjcr-a-defn} is $L$-$\Delta$-continuous for some constant $\gamma$ and
    \[
         L= \gamma \cdot \frac{k^3}{n} \cdot (\log (\lm)+1) \cdot \left(\log n + \log \Delta \right).
    \]

\end{restatable}

\subsection{Softmax Proportional Approval Voting (\noisypav)} \label{sec:exp-PAV}

We introduce a second rule that can either be seen as a randomized version of the well-established rule \emph{proportional approval voting (PAV)} \cite{Thie95a}, or as carefully randomizing over EJR+ committees via a weighting function dependent on the PAV score. 
As before, let $g(x) \defeq e^{a x}$ for some $a$ to be determined. 
Let $\pav{W} \defeq \sum_{i \in N} \sum_{r = 1}^{\abs{A_i \cap W}}r^{-1}$ denote the PAV score of committee $W$. %

\begin{algorithm}[H]
\caption{\noisypav{}}
\label{alg:exp-PAV}
$\W_E \gets \left\{W \in \binom{C}{k}: W \text{ satisfies \ejrp{} for } E\right\}$\;
Sample $W \sim \W_E$ with probability $P_W \defeq \frac{g(\pav{W})}{\sum_{W' \in \W_E} g(\pav{W'})}$ \;
\Return $W$\;
\end{algorithm}

While \noisypav{} trivially satisfies \ejrp{} and neutrality by definition, proving its monotonicity requires more arguments. The proof works for \emph{any} increasing $g(h)$. 

\begin{restatable}{theorem}{noisyPAVmonotone} \label{thm:exp-PAV-monotonicity}
    \noisypav{} 
    satisfies 
    ex-post \ejrp{}, ex-ante neutrality, and ex-ante monotonicity.
\end{restatable}

We now address the continuity of \noisypav{}. 
\begin{restatable}{theorem}{expPAVcontinuous}
\label{thm:exp-pav-continuous}
    \noisypav{} for some parameter $a$ (defined in the appendix) is $L$-$\Delta$-continuous for for some constant $\gamma$ and $L= \gamma \cdot \frac{k^2}{n} \cdot \hat\delta \cdot \log\left(\frac{m^kn}{k^2}\right)$, where $\hat\delta \defeq \harm(\min(\abs{\Delta}, \:k)) \leq \log k + 1$. 
\end{restatable}

For interpreting the above stated bound, we can rewrite it as 
    $O\left( \frac{k^3}{n} \cdot \hat \delta \cdot \log m + \frac{k^2}{n} \cdot \hat \delta \cdot \log n \right)$.
Comparing this to the stability of \noisygjcr{}, we observe that when $\log n = \Theta(\log m)$, the asymptotic stability guarantees for \noisygjcr{} and \noisypav{} are equivalent. However when $\Delta = O(1)$, our bound for \noisypav{} is more stable
by a factor of $\Theta(\log k)$.

\section{Differentially Private Committee Voting} \label{sec:dp-abc}

As mentioned in \Cref{sec:prelims}, the stronger guarantee of $(\eps, \delta)$-$\Delta$-max-KL stability with $\Delta = m$ corresponds to the concept of $(\eps,\delta)$-differential privacy (DP) \cite{dwork06calibrating}.
If a randomized function satisfies $(\eps,\delta)$-differential privacy then its behavior is similar for similar inputs. 
If each part of an input dataset is comprised of individual agents' data, a differentially private mechanism limits the amount of information an adversary can learn about any subset of agents' data from mechanism queries, even when equipped with the rest of the input.
In the study of randomized voting rules, DP has been motivated by the possibility of providing voters with stringent guarantees as to the privacy of their votes \cite{mohsin22learning,david23local}, as well as its possible implications for strategyproofness and manipulation-robustness \cite{lee15efficient,tao22local}.

The question of whether approval-based committee voting rules can satisfy differential privacy guarantees was first asked by \citet{li24differentially}.
They design rules which are $(\eps, 0)$-$m$-max-KL-stable.
However they observe no such rules satisfy \jr{} or stronger proportionality guarantees, and so they introduce novel relaxations of these proportionality axioms that their rules then satisfy. 
We discuss their and other possible approaches to `pure' DP in \Cref{sec:pure-DP-via-relaxations}. We first consider our rules.

\subsection{The Max-KL Stability of \noisypav{} and \noisygjcr{}}
To begin, it turns out that \noisypav{} can provide a range of max-KL stability guarantees.

\begin{restatable}{theorem}{noisyPAVmaxKLstable}
\label{thm:exp-pav-kl-stable}
    For any $a \geq 0$, \noisypav{} (\Cref{alg:exp-PAV}) satisfies $(\eps, \delta)$-$\Delta$-max-KL stability for $\eps = 2 a \hat\delta$ and $\delta = \exp\left(k \log m - a\frac{n}{k^2} + 2a \hat\delta \right)$, where $\hat\delta \defeq \harm(\min(\abs{\Delta}, \:k)) \leq \log k + 1$.
\end{restatable}
Recall that $a$ is a free parameter in \noisypav{} which specifies the `temperature' in the sampling from $\W_E$ in \Cref{alg:exp-PAV}.
We chose a particular value of $a$ in \Cref{sec:exp-PAV} in order to optimize TV stability; however no part of the subsequent proof depends on the choice of $a \geq 0$ and \Cref{thm:exp-pav-kl-stable} holds for any choice.
For instance, via $a$ near the value chosen in \Cref{thm:exp-pav-continuous}: 
\begin{restatable}{corollary}{noisyPAVmaxKLstableCorrolary}
\label{cor:exp-pav-kl-stable}
    \noisypav{} is $(\eps, \delta)$-$m$-max-KL-stable for $\eps = 2 a (\log k + 1) = \tilde O\left((k + \kappa)\frac{k^2}{n}\right)$ and $\delta \leq n^{-\kappa}$, when we choose $a = 2\frac{k^2}{n} \left( k \log m + \kappa \log n\right)$ and for any $\kappa > 0$.
\end{restatable}

It's worth pausing for a moment to interpret this claim.
By increasing the parameter $a$, we linearly increase the parameter $\eps$ which bounds the multiplicative change in the probability of \emph{any} event under the distribution output by \noisypav{}.
In exchange, the additive change can rapidly be made vanishingly small.

We now turn to \noisygjcr{}. 
To begin, we can derive some form of $(\eps,\delta)$-$\Delta$-max-KL-stability guarantee from the TV stability guarantee for \noisygjcr{}.
\Cref{lem:noisy-GJCR-stable} implies: 
\begin{restatable}{corollary}{noisyGJCRmaxKLstable}
\label{lem:noisy-GJCR-weak-DP-guarantee}
    \noisygjcr{} is $(\eps, \delta)$-$\Delta$-max-KL-stable for any $\eps \geq 0$ and $\delta = \gamma \cdot \frac{k^3}{n} \cdot \log k \cdot \log\left(n \cdot \Delta \right)$, for some constant $\gamma$.
\end{restatable}

We might hope that a trade-off similar to that in \Cref{thm:exp-pav-kl-stable} is also possible for our greedy, polynomial time rules; that by increasing the parameter $a$ (or $a_\ell$), we can rapidly drive the additive distributional shift towards $0$ while only slightly weakening multiplicative guarantees.
However, this is not the case.
While individual candidates' selections follow this trend, the event that \emph{no} candidate is selected can make a large additive probability jump, as we illustrate here:

\begin{restatable}{observation}{greedyrulesbadforDP}
\label{obs:greedy-bad-for-DP}
    If \noisygjcr{} or \noisygreedy{} with parameter choice of $a, a_\ell \leq 1$ satisfy $(\eps, \delta)$-$m$-max-KL stability, then $\delta \geq a/(1 - 1/e)$.
\end{restatable}

\subsection{Lower Bounds and Impossibilities} \label{sec:pure-DP-impossible}
What about general impossibilities for the max-KL-stability of rules that satisfy proportionality axioms?
We focus on \jr{}, but note that since \jr{} is weaker than \ejr{} or \ejrp{}, the stronger proportionality axioms inherit \jr{} impossibilities.
To begin, we observe that \Cref{ex:continuity} illustrates that no rule both guarantees \jr{} and provides $(\eps,0)$-$m$-max-KL stability.
\begin{restatable}{observation}{nopureDPwithJR}
\label{obs:no-pure-DP}
    No randomized rule satisfies \jr{} and $(\eps,0)$-$m$-max-KL stability, for any $\eps > 0$.
\end{restatable}

By adjusting this argument slightly, we may obtain a lower bound on the combinations of $(\eps, \delta)$ which are possible while satisfying \jr{} ex post (and hence also stronger proportionality axioms).

\begin{restatable}{claim}{DPlowerboundsforJR}
\label{claim:JR-DP-lowerbound}
    If a randomized rule $f$ satisfies \jr{} and $f$ is $(\eps,\delta)$-$m$-max-KL stable, then 
    \begin{equation} \label{eq:jr-dp-lower-bound}
        \delta \geq \eps \cdot \exp\left(- \eps \cdot \frac{n}{k}\right). 
    \end{equation}
\end{restatable}

We remark that \Cref{claim:JR-DP-lowerbound} is not quite tight for \noisypav{} (\Cref{thm:exp-pav-kl-stable}).
For example, our guarantees for \noisypav{} still satisfy \Cref{claim:JR-DP-lowerbound} when $\delta$ is replaced by $\delta^{2k}$ in \eqref{eq:jr-dp-lower-bound}.

\subsection{Pure Differential Privacy via Relaxed Proportionality} \label{sec:pure-DP-via-relaxations}
\Cref{obs:no-pure-DP} illustrates that \jr{} is incompatible with $(\eps,0)$-differential privacy.
Our stable rules provide $(\eps,\delta)$-differential privacy by allowing $\delta > 0$.
However if $(\eps,0)$-differential privacy is the goal, there are a number of plausible ways to relax the requirement that a rule satisfy \jr{} (or related axioms) ex post.
\citeauthor{li24differentially1} introduce $\theta$-\jr{} for $\theta \in [0,1]$, an ex-ante relaxation of \jr{} (and related axioms) for which each $1$-cohesive group (a $1$-large group of voters that approve a common candidate) has a voter who approves of a winning candidate with probability at least $\theta$, in expectation over the randomized rule \cite{li24differentially1}.
This is different from the approximate \jr{} of \Cref{def:JR}.
They consider the rule which chooses the PAV winner with some probability, and otherwise returns a uniformly random committee, which induces a trade-off between $\theta$ and $\eps$.
(In a subsequent version of this manuscript, \citet{li24differentially} evaluate the same rule against a distinct relaxation of \jr{}.)

Another way to circumvent \Cref{obs:no-pure-DP} while guaranteeing $(\eps,0)$-DP is to satisfy \jr{} or \ejrp{} merely with high probability (w.h.p.).
This is stronger than $\theta$-\jr{}, and indeed the randomized rule of \citeauthor{li24differentially} satisfies \ejrp{} w.h.p. also.
Modifying \noisypav{} (\Cref{alg:exp-PAV}) to sample exponentially from all committees would provide a failure probability proportional to $\delta$ in \Cref{thm:exp-pav-kl-stable}, and seems like a good candidate rule for this setting.
Finally, \Cref{obs:no-pure-DP} refers only to \jr{}.
We might also ask whether it is possible to construct $(\eps,0)$-max-KL stable rules for $\alpha$-\jr{} or $\alpha$-\ejrp{} (\Cref{def:JR,def:EJRp}) for $\alpha < 1$.
Following \Cref{obs:no-pure-DP} again shows this is impossible, since there are some combinations of profiles $A$ and committees $W$ for which $W$ does not provide $\alpha$-JR for any $\alpha >0$.

\section{Online Dynamic Committee Voting} \label{sec:dynamic}

TV-stable rules can also be leveraged to construct multiwinner voting rules that ensure proportionality guarantees in an online dynamic setting while incurring low recourse.
We show that any TV-stable rule can be used to design a dynamic rule maintaining a winning committee through a sequence of incremental changes to the approval profile, without changing the committee much.

We prove this for any TV-stable rule, and also demonstrate that, following this framework, \noisygjcr{} can be used to design a dynamic rule that maintains an \ejrp{} committee under incremental changes to the approval profile, makes only a small number of changes to the committee in expectation, and operates in polynomial time.

Recent work has studied multi-winner approval voting a number of online models.
\citet{do22online} pursue proportionality in a model where candidates and their approval data arrive online and must be irrevocably included or excluded from the winning committee.
\citet{halpern23representation} work with a random-order stream of voters, where only a few of each voter's approvals can be queried.
Their aim is to end with a committee that is proportional with high probability; \citet{brill23robust} also present results in this setting.
The model of this section is closest to that of \emph{perpetual voting}, wherein a rule faces a sequence of approval profiles and must choose a sequence of winning committees online \cite{lackner2020perpetual}.
However there arriving approval profiles can be arbitrarily different from one another, and the emphasis is on defining and satisfying inter-step notions of fairness and proportionality \cite{lackner23proportional}, rather than minimizing recourse.

\subsection{The Dynamic Model}\label{sec:dynamic-model}

We assume that an adversary presents a sequence of elections $E = (E^1, \ldots, E^t, \ldots, E^T)$ which are revealed one-at-a-time to the algorithm (dynamic rule), such that for every $t$ the approval profiles $A^t$ and $A^{t+1}$ differ only on the approvals of a single voter $v_{t+1}$.
Here we take $\Delta = m$ for convenience.

The dynamic rule's goal is to maintain a winning committee $W^t$, represented by the sequence $W = (W^1, \ldots, W^t, \ldots, W^T)$, such that $W^t$ satisfies \ejrp{} with respect to $A^t$ for all $t$, while minimizing the \emph{total recourse}
\begin{equation}
    R(W) \defeq \sum_{t\in [T-1]} \abs{W^t \oplus W^{t+1}}.
\end{equation}
This is the total change in the winning committee $W^t$ that the algorithm maintains over all steps.

The adversary may be \emph{oblivious}, which means that the sequence $E$ is fixed at the outset and arriving $E^t$ cannot change in response to the dynamic rule's prior choices, or \emph{adaptive}, meaning that they can.
We will assume an oblivious adversary.
The aim is to establish bounds on
the \emph{expected} recourse $\frac{1}{T}\expect{R(W)}$, where the expectation is taken over the randomness of the algorithm.

\subsubsection{Deterministic Rules for Low Recourse}
The score-based rules PAV and Local Search PAV are natural choices for designing low-recourse dynamic rules.
In particular, consider lazy versions of PAV and $\eps$-LS-PAV \cite{kraiczy23properties} for $\eps = \frac{n}{2k^2}$, say, which change the winning committee to some $W^{t'} \neq W^t$ for $t' > t$ only when $W^t$ no longer satisfies \ejrp{} for $E^{t'}$. 
Then the standard pivot-based argument %
also shows that the committees output by these rules satisfy \ejrp{}, and since for all committees $W$ the PAV score $\pav{W}$ changes by at most $\log k + 1$ in each step, at least $t' - t \geq \Omega(\frac{n}{k^2 \log k} )$ changes to $E^t$ must occur before $W^t$ no longer satisfies \ejrp{}.
Since this $\eps$-LS-PAV swaps only one member of the committee to improve the PAV score by $\Omega(n/k^2)$, its dynamic counterpart has expected (in fact, amortized) recourse $O\left(k^2 \log k/n \right)$.
A dynamic PAV then has recourse $O\left(k^3 \log k/n \right)$, since the committee may change by $2k$ candidates.

However, each leaves something to be desired; PAV is NP-hard to compute~\cite{AzizGGMMW15}
, while $\eps$-LS-PAV may violate candidate monotonicity between steps \cite{kraiczy23properties}.

\subsubsection{Randomized Rules for Low Recourse}
We now show that any TV-stable rule gives rise to a low-recourse dynamic rule; as a result, our stable rules imply low-recourse dynamic rules which satisfy proportionality ex post while retaining other desirable properties.
While the dynamic counterparts of our rules do not quite match the recourse guarantee of $\eps$-LS-PAV, they have other desirable properties. 
In particular, our monotone rules are still ex-ante monotone in the sense that if $A^{t-1}$ and $A^t$ differ only on the new approval of some candidate $c$, then the initial probability that $c$ is chosen increases from step $t-1$ to $t$.
The dynamic counterparts of our rules are also \emph{ex-ante memoryless}, meaning that the initial distribution of $W^t$ is independent of the approval profiles $A^{t'}$ for $t' \neq t$.

\subsection{Reducing Recourse to Stability}
Broadly speaking, the idea of using private randomness to initialize a low-recourse algorithm against an oblivious adversary is not a new idea.
But we have not encountered TV stability deployed in this way.
We state in general terms the insight which drives this result.
We formulate it as a reduction.

\begin{restatable}{proposition}{dynamictostablereduction}
\label{prop:reduction}
    Let $f: \cI \rightarrow \cR$ be a randomized algorithm which solves some problem ex post, meaning that $\supp(\cD_{f(I)}) \subseteq Q(I)$, where $Q(I)$ are feasible solutions to the problem for input $I$.
    Suppose $f$ is $\eps$-TV-stable, i.e., for all $(I, I') \in \cI_\square \defeq \{(I, I') \in \cI^2: I, I' \text{ differ by at most one element}\}$
    \[
        \dtv{\cD_{f(I)}}{\cD_{f(I')}} \leq \eps.
    \]
    Then \Cref{alg:dynamic-stable-reduction} produces a sequence $(R^1, \ldots, R^T)$ for which $R^t \subseteq Q(I^t)$ for all $t$, and with expected recourse at most $\eps \cdot K$, where $K \defeq \max_{(I,I') \in I_\square} \max_{(R,R') \in \supp(f(I)) \times \supp(f(I'))}\abs{R \oplus R'}$.
\end{restatable}
(Note that the TV stability here is that of \cite{bassily21algorithmic}, which coincides with ours for $\Delta = m$.)
This parameter $K$ captures the extent to which the possible outputs of $f$ can differ on neighboring inputs.
In our setting of dynamic committee rules, an output of $f(E)$ has size of at most $k$, and so $K \leq 2k$. 
For a given axiom, for instance \ejrp{}, we therefore have the following corollary:
\begin{corollary} \label{cor:voting-reduction}
    Let $f$ be a rule which satisfies \ejrp{} ex post.
    If $f$ is $\eps$-$m$-TV-stable, then \Cref{alg:dynamic-stable-reduction} produces a sequence $(W^t)_{t \in [T]}$ which satisfies \ejrp{} for all $t$, and with expected recourse at most $2\eps k$.
\end{corollary}

Recall the Coupling Lemma (\Cref{lem:coupling}). 
The ideal approach here would be to couple \emph{all} of the distributions $(\cD_{f(I)})_{I \in \cI}$, draw some $(R_I)_{I \in \cI} \sim \cR^\cI$, and for input sequence $(I^t)_{t \in [T]}$ produce the output sequence $(R_{I^t})_{t \in [T]}$. 
Unfortunately, this does not work: optimal coupling only holds for pairs of distributions. 
Fortunately, \Cref{alg:dynamic-stable-reduction} can use coupling to achieve a similar result.
At each step, it considers the optimal coupling between the distributions, and samples the next output from the coupling, conditioned on the realization of the previous value.

\begin{algorithm}[ht!]
\caption{Dynamic to Stable Reduction}
\label{alg:dynamic-stable-reduction}
\Input{$\eps$-TV-stable $f$ and $(A^t)_{t \in [T]}$}
\Output{$(R^1, \ldots, R^T)$}
$R^1 \gets f(A^1)$\; \label{line:reduction-initial-sample}
\For{$t \in (2, \ldots, T)$}{
    $\cE^t \leftarrow$ optimal coupling of $\cD_{f(A^{t})}$ and $\cD_{f(A^{t-1})}$\;
    $R^t \leftarrow x$ for $(x, y) \sim (\cE^t$ conditioned on the second coordinate being equal to $R^{t-1}$)\;\label{line:reduction-conditional-sample}
}
\Return $R$\;
\end{algorithm}

\subsection{Opening the Box: \dynamicgjcr{}}
This reduction (\Cref{prop:reduction}) does not make any claims about the runtime of the resulting dynamic algorithm.
When---as is the case here---randomized algorithms sample from large distributions, the conditional sampling on \cref{line:reduction-conditional-sample} of \Cref{alg:dynamic-stable-reduction} may take more than polynomial time.

By leveraging the specific structure and properties of our rules, we can construct low-recourse rules that run in polynomial time while retaining the desirable properties of their stable precursors.
We demonstrate this by way of example for the case of \noisygjcr (\Cref{alg:gjcr-noisy}), which we adapt to the dynamic setting %
in \Cref{app:dynamic}.

Let $\cN_s^{t-1}$ and $\cN_s^{t}$ denote the distributions from which \noisygjcr samples on \cref{line:GJCR-noisy-sample}, given a partial sequence $s$ of candidates and under inputs $E^{t-1}$ and $E^{t}$, respectively.
We correlate these distributions by sampling from their optimal coupling $\cM_s^{t}$.
Since the distributions $\cN_s^{t-1}$ and $\cN_s^{t}$ are constructed explicitly in polynomial time and each has support $m$, an optimal coupling $\cM_s^{t}$ can be explicitly computed in polynomial time as well.

\begin{restatable}{theorem}{dynamicgjcrpolytime}\label{thm:dynamicgjcr}
    \dynamicgjcr{} (see appendix) outputs a sequence of \ejrp{} consistent committees with expected recourse $\tilde O\left(\frac{k^4}{n}\right)$, provides ex-ante monotonicity, and runs in polynomial time.
\end{restatable}

\section{Conclusion}\label{sec:conclusion}
In this work we demonstrate how the `slack' present in many of the rules that provide proportionality in approval-based committee elections can be leveraged to create randomized rules that are algorithmically stable.
This confers continuity to individual candidate selection probabilities, which we argue is favorable from the candidate perspective. 
We observe that it is also useful for guaranteeing voter privacy, and for efficiently maintaining proportional committees in settings where voter approvals may change over time.

We leave open several exciting directions for future work.
First is to close the $\tilde\Theta(k^2)$ gap between our stability guarantees and impossibilities---we hesitate to conjecture whether either are tight. 
Tighter bounds could in particular reveal whether the trade-offs between stability and proportionality in \Cref{sec:more-stability} are innate.
Our results in \Cref{sec:dp-abc} also invite the question of whether it is possible to provide polynomial-time $(\eps,\delta)$-differential privacy guarantees for which $\delta$ can be made arbitrarily small. 
Finally, looking beyond multiwinner voting, is there a stable version of the method of equal shares (MES)~\cite{PierczynskiSP21equalshares}? 
Such a rule could enable similar guarantees in the more general setting of participatory budgeting.

\begin{acks}
We thank Piotr Skowron for helpful discussions in the initial stages of this project. We are also grateful to the Bellairs Research Institute, and to Adrian Vetta and Alexandra Lassota for organizing the 2024 Workshop on Fairness in Operations Research, where this project was initiated.
We also thank anonymous reviewers for their valuable feedback.
Ulrike Schmidt-Kraepelin thanks Ariel Procaccia, Paul Gölz, Markus Brill, and Jannik Peters for early discussions on the related concept of party-fairness, which helped shape some of the ideas developed in this paper.
Gregory Kehne thanks Anupam Gupta for communicating the reduction of \Cref{sec:dynamic} in a dream.

Ulrike Schmidt-Kraepelin was supported by the Dutch Research Council (NWO) under project number VI.Veni.232.254.
Krzysztof Sornat was supported by the European Research Council (ERC) under the European Union’s Horizon 2020 research and innovation programme (grant agreement No 101002854). 
\end{acks}

\bibliographystyle{ACM-Reference-Format}
\bibliography{refs}


\begin{thebibliography}{43}

%%% the Web address.

\ifx \showCODEN    \undefined \def \showCODEN     #1{\unskip}     \fi
\ifx \showDOI      \undefined \def \showDOI       #1{#1}\fi
\ifx \showISBNx    \undefined \def \showISBNx     #1{\unskip}     \fi
\ifx \showISBNxiii \undefined \def \showISBNxiii  #1{\unskip}     \fi
\ifx \showISSN     \undefined \def \showISSN      #1{\unskip}     \fi
\ifx \showLCCN     \undefined \def \showLCCN      #1{\unskip}     \fi
\ifx \shownote     \undefined \def \shownote      #1{#1}          \fi
\ifx \showarticletitle \undefined \def \showarticletitle #1{#1}   \fi
\ifx \showURL      \undefined \def \showURL       {\relax}        \fi
% invisible to TeX
\providecommand\bibfield[2]{#2}
\providecommand\bibinfo[2]{#2}
\providecommand\natexlab[1]{#1}
\providecommand\showeprint[2][]{arXiv:#2}

\bibitem[Aldous(1983)]%
        {aldous1983random}
\bibfield{author}{\bibinfo{person}{David Aldous}.}
  \bibinfo{year}{1983}\natexlab{}.
\newblock \showarticletitle{Random Walks on Finite Groups and Rapidly Mixing
  Markov Chains}. In \bibinfo{booktitle}{\emph{Proceedings of S{\'e}minaire de
  Probabilit{\'e}s XVII 1981/82}}. \bibinfo{pages}{243--297}.
\newblock


\bibitem[Aziz et~al\mbox{.}(2017)]%
        {aziz17justified}
\bibfield{author}{\bibinfo{person}{Haris Aziz}, \bibinfo{person}{Markus Brill},
  \bibinfo{person}{Vincent Conitzer}, \bibinfo{person}{Edith Elkind},
  \bibinfo{person}{Rupert Freeman}, {and} \bibinfo{person}{Toby Walsh}.}
  \bibinfo{year}{2017}\natexlab{}.
\newblock \showarticletitle{Justified Representation in Approval-based
  Committee Voting}.
\newblock \bibinfo{journal}{\emph{Soc. Choice Welf.}} \bibinfo{volume}{48},
  \bibinfo{number}{2} (\bibinfo{year}{2017}), \bibinfo{pages}{461--485}.
\newblock


\bibitem[Aziz et~al\mbox{.}(2018)]%
        {aziz18complexity}
\bibfield{author}{\bibinfo{person}{Haris Aziz}, \bibinfo{person}{Edith Elkind},
  \bibinfo{person}{Shenwei Huang}, \bibinfo{person}{Martin Lackner},
  \bibinfo{person}{Luis {S{\'{a}}nchez{-}Fern{\'{a}}ndez}}, {and}
  \bibinfo{person}{Piotr Skowron}.} \bibinfo{year}{2018}\natexlab{}.
\newblock \showarticletitle{On the Complexity of Extended and Proportional
  Justified Representation}. In \bibinfo{booktitle}{\emph{Proceedings of the
  32nd {AAAI} Conference on Artificial Intelligence, {AAAI} 2018}}.
  \bibinfo{pages}{902--909}.
\newblock


\bibitem[Aziz et~al\mbox{.}(2015)]%
        {AzizGGMMW15}
\bibfield{author}{\bibinfo{person}{Haris Aziz}, \bibinfo{person}{Serge
  Gaspers}, \bibinfo{person}{Joachim Gudmundsson}, \bibinfo{person}{Simon
  Mackenzie}, \bibinfo{person}{Nicholas Mattei}, {and} \bibinfo{person}{Toby
  Walsh}.} \bibinfo{year}{2015}\natexlab{}.
\newblock \showarticletitle{Computational Aspects of Multi-Winner Approval
  Voting}. In \bibinfo{booktitle}{\emph{the 2015 International Conference on
  Autonomous Agents and Multiagent Systems, {AAMAS} 2015}}.
  \bibinfo{pages}{107--115}.
\newblock


\bibitem[Aziz et~al\mbox{.}(2023)]%
        {aziz2023best}
\bibfield{author}{\bibinfo{person}{Haris Aziz}, \bibinfo{person}{Xinhang Lu},
  \bibinfo{person}{Mashbat Suzuki}, \bibinfo{person}{Jeremy Vollen}, {and}
  \bibinfo{person}{Toby Walsh}.} \bibinfo{year}{2023}\natexlab{}.
\newblock \bibinfo{booktitle}{\emph{Best-of-Both-Worlds Fairness in Committee
  Voting}}.
\newblock \bibinfo{type}{{T}echnical {R}eport}.
  \bibinfo{institution}{arxiv.org/2303.03642}.
\newblock


\bibitem[Aziz and Shah(2021)]%
        {aziz2021participatory}
\bibfield{author}{\bibinfo{person}{Haris Aziz} {and} \bibinfo{person}{Nisarg
  Shah}.} \bibinfo{year}{2021}\natexlab{}.
\newblock \showarticletitle{Participatory Budgeting: {M}odels and Approaches}.
\newblock \bibinfo{journal}{\emph{Pathways Between Social Science and
  Computational Social Science: Theories, Methods, and Interpretations}}
  (\bibinfo{year}{2021}), \bibinfo{pages}{215--236}.
\newblock


\bibitem[Bassily et~al\mbox{.}(2021)]%
        {bassily21algorithmic}
\bibfield{author}{\bibinfo{person}{Raef Bassily}, \bibinfo{person}{Kobbi
  Nissim}, \bibinfo{person}{Adam~D. Smith}, \bibinfo{person}{Thomas Steinke},
  \bibinfo{person}{Uri Stemmer}, {and} \bibinfo{person}{Jonathan~R. Ullman}.}
  \bibinfo{year}{2021}\natexlab{}.
\newblock \showarticletitle{Algorithmic Stability for Adaptive Data Analysis}.
\newblock \bibinfo{journal}{\emph{{SIAM} J. Comput.}} \bibinfo{volume}{50},
  \bibinfo{number}{3} (\bibinfo{year}{2021}).
\newblock


\bibitem[Beimel et~al\mbox{.}(2022)]%
        {beimel2022dynamic}
\bibfield{author}{\bibinfo{person}{Amos Beimel}, \bibinfo{person}{Haim Kaplan},
  \bibinfo{person}{Yishay Mansour}, \bibinfo{person}{Kobbi Nissim},
  \bibinfo{person}{Thatchaphol Saranurak}, {and} \bibinfo{person}{Uri
  Stemmer}.} \bibinfo{year}{2022}\natexlab{}.
\newblock \showarticletitle{Dynamic Algorithms Against an Adaptive Adversary:
  {G}eneric Constructions and Lower Bounds}. In
  \bibinfo{booktitle}{\emph{Proceedings of the 54th Annual {ACM} {SIGACT}
  Symposium on Theory of Computing, {STOC} 2022}}. \bibinfo{pages}{1671--1684}.
\newblock


\bibitem[Boehmer et~al\mbox{.}(2024)]%
        {boehmer2024approval}
\bibfield{author}{\bibinfo{person}{Niclas Boehmer}, \bibinfo{person}{Markus
  Brill}, \bibinfo{person}{Alfonso Cevallos}, \bibinfo{person}{Jonas Gehrlein},
  \bibinfo{person}{Luis {S{\'{a}}nchez Fern{\'{a}}ndez}}, {and}
  \bibinfo{person}{Ulrike Schmidt{-}Kraepelin}.}
  \bibinfo{year}{2024}\natexlab{}.
\newblock \showarticletitle{Approval-Based Committee Voting in Practice: {A}
  Case Study of (over-)Representation in the Polkadot Blockchain}. In
  \bibinfo{booktitle}{\emph{Proceedings of the 38th {AAAI} Conference on
  Artificial Intelligence, {AAAI} 2024}}. \bibinfo{pages}{9519--9527}.
\newblock


\bibitem[Boehmer et~al\mbox{.}(2023)]%
        {boehmer2023robustness}
\bibfield{author}{\bibinfo{person}{Niclas Boehmer}, \bibinfo{person}{Piotr
  Faliszewski}, \bibinfo{person}{Lukasz Janeczko}, {and}
  \bibinfo{person}{Andrzej Kaczmarczyk}.} \bibinfo{year}{2023}\natexlab{}.
\newblock \showarticletitle{Robustness of Participatory Budgeting Outcomes:
  {C}omplexity and Experiments}. In \bibinfo{booktitle}{\emph{Proceedings of
  the 16th International Symposium Algorithmic Game Theory, {SAGT} 2023}}.
  \bibinfo{pages}{161--178}.
\newblock


\bibitem[Brill and Peters(2023)]%
        {brill23robust}
\bibfield{author}{\bibinfo{person}{Markus Brill} {and} \bibinfo{person}{Jannik
  Peters}.} \bibinfo{year}{2023}\natexlab{}.
\newblock \showarticletitle{Robust and Verifiable Proportionality Axioms for
  Multiwinner Voting}. In \bibinfo{booktitle}{\emph{Proceedings of the 24th
  {ACM} Conference on Economics and Computation, {EC} 2023}}.
  \bibinfo{pages}{301}.
\newblock


\bibitem[Cevallos and Stewart(2021)]%
        {cevallos2021verifiably}
\bibfield{author}{\bibinfo{person}{Alfonso Cevallos} {and}
  \bibinfo{person}{Alistair Stewart}.} \bibinfo{year}{2021}\natexlab{}.
\newblock \showarticletitle{A Verifiably Secure and Proportional Committee
  Election Rule}. In \bibinfo{booktitle}{\emph{Proceedings of the 3rd {ACM}
  Conference on Advances in Financial Technologies, {AFT} 2021}}.
  \bibinfo{pages}{29--42}.
\newblock


\bibitem[Cheng et~al\mbox{.}(2019)]%
        {cheng19group}
\bibfield{author}{\bibinfo{person}{Yu Cheng}, \bibinfo{person}{Zhihao Jiang},
  \bibinfo{person}{Kamesh Munagala}, {and} \bibinfo{person}{Kangning Wang}.}
  \bibinfo{year}{2019}\natexlab{}.
\newblock \showarticletitle{Group Fairness in Committee Selection}. In
  \bibinfo{booktitle}{\emph{Proceedings of the 2019 {ACM} Conference on
  Economics and Computation, {EC} 2019}}. \bibinfo{pages}{263--279}.
\newblock


\bibitem[David et~al\mbox{.}(2023)]%
        {david23local}
\bibfield{author}{\bibinfo{person}{Bernardo David}, \bibinfo{person}{Rosario
  Giustolisi}, \bibinfo{person}{Victor Mortensen}, {and}
  \bibinfo{person}{Morten Pedersen}.} \bibinfo{year}{2023}\natexlab{}.
\newblock \showarticletitle{Local Differential Privacy in Voting}. In
  \bibinfo{booktitle}{\emph{Proceedings of the Italian Conference on Cyber
  Security, {ITASEC} 2023}}.
\newblock


\bibitem[Do et~al\mbox{.}(2022)]%
        {do22online}
\bibfield{author}{\bibinfo{person}{Virginie Do}, \bibinfo{person}{Matthieu
  Hervouin}, \bibinfo{person}{J{\'{e}}r{\^{o}}me Lang}, {and}
  \bibinfo{person}{Piotr Skowron}.} \bibinfo{year}{2022}\natexlab{}.
\newblock \showarticletitle{Online Approval Committee Elections}. In
  \bibinfo{booktitle}{\emph{Proceedings of the 31st International Joint
  Conference on Artificial Intelligence, {IJCAI} 2022}}.
  \bibinfo{pages}{251--257}.
\newblock


\bibitem[Dwork et~al\mbox{.}(2006)]%
        {dwork06calibrating}
\bibfield{author}{\bibinfo{person}{Cynthia Dwork}, \bibinfo{person}{Frank
  McSherry}, \bibinfo{person}{Kobbi Nissim}, {and} \bibinfo{person}{Adam~D.
  Smith}.} \bibinfo{year}{2006}\natexlab{}.
\newblock \showarticletitle{Calibrating Noise to Sensitivity in Private Data
  Analysis}. In \bibinfo{booktitle}{\emph{Proceedings of the 3rd Theory of
  Cryptography Conference, {TCC} 2006}}. \bibinfo{pages}{265--284}.
\newblock


\bibitem[Elkind et~al\mbox{.}(2024a)]%
        {ElkindFIMSS24}
\bibfield{author}{\bibinfo{person}{Edith Elkind}, \bibinfo{person}{Piotr
  Faliszewski}, \bibinfo{person}{Ayumi Igarashi}, \bibinfo{person}{Pasin
  Manurangsi}, \bibinfo{person}{Ulrike Schmidt{-}Kraepelin}, {and}
  \bibinfo{person}{Warut Suksompong}.} \bibinfo{year}{2024}\natexlab{a}.
\newblock \showarticletitle{The Price of Justified Representation}.
\newblock \bibinfo{journal}{\emph{{ACM} Trans. Economics and Comput.}}
  \bibinfo{volume}{12}, \bibinfo{number}{3} (\bibinfo{year}{2024}),
  \bibinfo{pages}{11:1--11:27}.
\newblock


\bibitem[Elkind et~al\mbox{.}(2024b)]%
        {elkind2024temporal}
\bibfield{author}{\bibinfo{person}{Edith Elkind}, \bibinfo{person}{Svetlana
  Obraztsova}, {and} \bibinfo{person}{Nicholas Teh}.}
  \bibinfo{year}{2024}\natexlab{b}.
\newblock \showarticletitle{Temporal Fairness in Multiwinner Voting}. In
  \bibinfo{booktitle}{\emph{Proceedings of the 38th {AAAI} Conference on
  Artificial Intelligence, {AAAI} 2024}}. \bibinfo{pages}{22633--22640}.
\newblock


\bibitem[Fish et~al\mbox{.}(2024)]%
        {fish2024generative}
\bibfield{author}{\bibinfo{person}{Sara Fish}, \bibinfo{person}{Paul
  G{\"{o}}lz}, \bibinfo{person}{David~C. Parkes}, \bibinfo{person}{Ariel~D.
  Procaccia}, \bibinfo{person}{Gili Rusak}, \bibinfo{person}{Itai Shapira},
  {and} \bibinfo{person}{Manuel W{\"{u}}thrich}.}
  \bibinfo{year}{2024}\natexlab{}.
\newblock \showarticletitle{Generative Social Choice}. In
  \bibinfo{booktitle}{\emph{Proceedings of the 25th {ACM} Conference on
  Economics and Computation, {EC} 2024}}. \bibinfo{pages}{985}.
\newblock


\bibitem[Flanigan et~al\mbox{.}(2023)]%
        {flanigan23smoothed}
\bibfield{author}{\bibinfo{person}{Bailey Flanigan}, \bibinfo{person}{Daniel
  Halpern}, {and} \bibinfo{person}{Alexandros Psomas}.}
  \bibinfo{year}{2023}\natexlab{}.
\newblock \showarticletitle{Smoothed Analysis of Social Choice Revisited}. In
  \bibinfo{booktitle}{\emph{Proceedings of the 19th International Conference
  Web and Internet Economics, {WINE} 2023}}. \bibinfo{pages}{290--309}.
\newblock


\bibitem[Gawron and Faliszewski(2019)]%
        {gawron2019robustness}
\bibfield{author}{\bibinfo{person}{Grzegorz Gawron} {and}
  \bibinfo{person}{Piotr Faliszewski}.} \bibinfo{year}{2019}\natexlab{}.
\newblock \showarticletitle{Robustness of Approval-Based Multiwinner Voting
  Rules}. In \bibinfo{booktitle}{\emph{Proceedings of the 6th International
  Conference on Algorithmic Decision Theory, {ADT} 2019}}.
  \bibinfo{pages}{17--31}.
\newblock


\bibitem[Halpern et~al\mbox{.}(2023)]%
        {halpern23representation}
\bibfield{author}{\bibinfo{person}{Daniel Halpern}, \bibinfo{person}{Gregory
  Kehne}, \bibinfo{person}{Ariel~D. Procaccia}, \bibinfo{person}{Jamie
  Tucker{-}Foltz}, {and} \bibinfo{person}{Manuel W{\"{u}}thrich}.}
  \bibinfo{year}{2023}\natexlab{}.
\newblock \showarticletitle{Representation with Incomplete Votes}. In
  \bibinfo{booktitle}{\emph{Proceedings of the 37th {AAAI} Conference on
  Artificial Intelligence, {AAAI} 2023}}. \bibinfo{pages}{5657--5664}.
\newblock


\bibitem[Jiang et~al\mbox{.}(2020)]%
        {jiang20approximately}
\bibfield{author}{\bibinfo{person}{Zhihao Jiang}, \bibinfo{person}{Kamesh
  Munagala}, {and} \bibinfo{person}{Kangning Wang}.}
  \bibinfo{year}{2020}\natexlab{}.
\newblock \showarticletitle{Approximately Stable Committee Selection}. In
  \bibinfo{booktitle}{\emph{Proceedings of the 52nd Annual {ACM} {SIGACT}
  Symposium on Theory of Computing, {STOC} 2020}}. \bibinfo{pages}{463--472}.
\newblock


\bibitem[Kraiczy and Elkind(2023)]%
        {kraiczy23properties}
\bibfield{author}{\bibinfo{person}{Sonja Kraiczy} {and} \bibinfo{person}{Edith
  Elkind}.} \bibinfo{year}{2023}\natexlab{}.
\newblock \showarticletitle{Properties of Local Search {PAV}}. In
  \bibinfo{booktitle}{\emph{Proceedings of the 9th International Workshop on
  Computational Social Choice, {COMSOC} 2023}}.
\newblock


\bibitem[Lackner(2020)]%
        {lackner2020perpetual}
\bibfield{author}{\bibinfo{person}{Martin Lackner}.}
  \bibinfo{year}{2020}\natexlab{}.
\newblock \showarticletitle{Perpetual Voting: {F}airness in Long-Term Decision
  Making}. In \bibinfo{booktitle}{\emph{Proceedings of the 34th {AAAI}
  Conference on Artificial Intelligence, {AAAI} 2020}}.
  \bibinfo{pages}{2103--2110}.
\newblock


\bibitem[Lackner and Maly(2023)]%
        {lackner23proportional}
\bibfield{author}{\bibinfo{person}{Martin Lackner} {and} \bibinfo{person}{Jan
  Maly}.} \bibinfo{year}{2023}\natexlab{}.
\newblock \showarticletitle{Proportional Decisions in Perpetual Voting}. In
  \bibinfo{booktitle}{\emph{Proceedings of the 37th {AAAI} Conference on
  Artificial Intelligence, {AAAI} 2023}}. \bibinfo{pages}{5722--5729}.
\newblock


\bibitem[Lackner and Skowron(2023)]%
        {lackner23multi}
\bibfield{author}{\bibinfo{person}{Martin Lackner} {and} \bibinfo{person}{Piotr
  Skowron}.} \bibinfo{year}{2023}\natexlab{}.
\newblock \bibinfo{booktitle}{\emph{Multi-Winner Voting with Approval
  Preferences}}.
\newblock \bibinfo{publisher}{Springer}.
\newblock


\bibitem[Lee(2015)]%
        {lee15efficient}
\bibfield{author}{\bibinfo{person}{David~Timothy Lee}.}
  \bibinfo{year}{2015}\natexlab{}.
\newblock \showarticletitle{Efficient, Private, and eps-Strategyproof
  Elicitation of Tournament Voting Rules}. In
  \bibinfo{booktitle}{\emph{Proceedings of the 24th International Joint
  Conference on Artificial Intelligence, {IJCAI} 2015}}.
  \bibinfo{pages}{2026--2032}.
\newblock


\bibitem[Levin and Peres(2017)]%
        {levin17markov}
\bibfield{author}{\bibinfo{person}{David~A Levin} {and} \bibinfo{person}{Yuval
  Peres}.} \bibinfo{year}{2017}\natexlab{}.
\newblock \bibinfo{booktitle}{\emph{Markov chains and mixing times}}.
  Vol.~\bibinfo{volume}{107}.
\newblock \bibinfo{publisher}{American Mathematical Soc.}
\newblock


\bibitem[Li et~al\mbox{.}(2024a)]%
        {li24differentially}
\bibfield{author}{\bibinfo{person}{Zhechen Li}, \bibinfo{person}{Zimai Guo},
  \bibinfo{person}{Lirong Xia}, \bibinfo{person}{Yongzhi Cao}, {and}
  \bibinfo{person}{Hanpin Wang}.} \bibinfo{year}{2024}\natexlab{a}.
\newblock \bibinfo{booktitle}{\emph{Differentially Private Approval-Based
  Committee Voting}}.
\newblock \bibinfo{type}{{T}echnical {R}eport}.
  \bibinfo{institution}{arxiv.org/2401.10122v2}.
\newblock


\bibitem[Li et~al\mbox{.}(2024b)]%
        {li24differentially1}
\bibfield{author}{\bibinfo{person}{Zhechen Li}, \bibinfo{person}{Zimai Guo},
  \bibinfo{person}{Lirong Xia}, \bibinfo{person}{Yongzhi Cao}, {and}
  \bibinfo{person}{Hanpin Wang}.} \bibinfo{year}{2024}\natexlab{b}.
\newblock \bibinfo{booktitle}{\emph{Differentially Private Approval-Based
  Committee Voting}}.
\newblock \bibinfo{type}{{T}echnical {R}eport}.
  \bibinfo{institution}{arxiv.org/2401.10122v1}.
\newblock


\bibitem[Mohsin et~al\mbox{.}(2022)]%
        {mohsin22learning}
\bibfield{author}{\bibinfo{person}{Farhad Mohsin}, \bibinfo{person}{Ao Liu},
  \bibinfo{person}{Pin{-}Yu Chen}, \bibinfo{person}{Francesca Rossi}, {and}
  \bibinfo{person}{Lirong Xia}.} \bibinfo{year}{2022}\natexlab{}.
\newblock \showarticletitle{Learning to Design Fair and Private Voting Rules}.
\newblock \bibinfo{journal}{\emph{J. Artif. Intell. Res.}}
  \bibinfo{volume}{75} (\bibinfo{year}{2022}), \bibinfo{pages}{1139--1176}.
\newblock


\bibitem[Munagala et~al\mbox{.}(2022)]%
        {MSW+22a}
\bibfield{author}{\bibinfo{person}{Kamesh Munagala}, \bibinfo{person}{Yiheng
  Shen}, \bibinfo{person}{Kangning Wang}, {and} \bibinfo{person}{Zhiyi Wang}.}
  \bibinfo{year}{2022}\natexlab{}.
\newblock \showarticletitle{Approximate Core for Committee Selection via
  Multilinear Extension and Market Clearing}. In \bibinfo{booktitle}{\emph{2022
  {ACM-SIAM} Symposium on Discrete Algorithms, {SODA} 2022}}.
  \bibinfo{pages}{2229--2252}.
\newblock


\bibitem[Peters et~al\mbox{.}(2021)]%
        {PierczynskiSP21equalshares}
\bibfield{author}{\bibinfo{person}{Dominik Peters}, \bibinfo{person}{Grzegorz
  Pierczynski}, {and} \bibinfo{person}{Piotr Skowron}.}
  \bibinfo{year}{2021}\natexlab{}.
\newblock \showarticletitle{Proportional Participatory Budgeting with Additive
  Utilities}. In \bibinfo{booktitle}{\emph{Proceedings of the 34th
  International Conference on Neural Information Processing Systems, {NeurIPS}
  2021}}. \bibinfo{pages}{12726--12737}.
\newblock


\bibitem[Peters and Skowron(2020)]%
        {PetersS20equalshares}
\bibfield{author}{\bibinfo{person}{Dominik Peters} {and} \bibinfo{person}{Piotr
  Skowron}.} \bibinfo{year}{2020}\natexlab{}.
\newblock \showarticletitle{Proportionality and the Limits of Welfarism}. In
  \bibinfo{booktitle}{\emph{Proceedings of the 21st {ACM} Conference on
  Economics and Computation, {EC} 2020}}. \bibinfo{pages}{793--794}.
\newblock


\bibitem[Rey and Maly(2023)]%
        {rey2023computational}
\bibfield{author}{\bibinfo{person}{Simon Rey} {and} \bibinfo{person}{Jan
  Maly}.} \bibinfo{year}{2023}\natexlab{}.
\newblock \bibinfo{booktitle}{\emph{The (Computational) Social Choice Take on
  Indivisible Participatory Budgeting}}.
\newblock \bibinfo{type}{{T}echnical {R}eport}.
  \bibinfo{institution}{arxiv.org/2303.00621}.
\newblock


\bibitem[{S{\'{a}}nchez{-}Fern{\'{a}}ndez}(2025)]%
        {SanchezFernandez25}
\bibfield{author}{\bibinfo{person}{Luis {S{\'{a}}nchez{-}Fern{\'{a}}ndez}}.}
  \bibinfo{year}{2025}\natexlab{}.
\newblock \showarticletitle{A Note on the Method of Equal Shares}.
\newblock \bibinfo{journal}{\emph{Inf. Process. Lett.}}  \bibinfo{volume}{190}
  (\bibinfo{year}{2025}), \bibinfo{pages}{106576}.
\newblock


\bibitem[S{\'{a}}nchez{-}Fern{\'{a}}ndez et~al\mbox{.}(2017)]%
        {FernandezEL17}
\bibfield{author}{\bibinfo{person}{Luis S{\'{a}}nchez{-}Fern{\'{a}}ndez},
  \bibinfo{person}{Edith Elkind}, {and} \bibinfo{person}{Martin Lackner}.}
  \bibinfo{year}{2017}\natexlab{}.
\newblock \bibinfo{booktitle}{\emph{Committees Providing {EJR} Can be Computed
  Efficiently}}.
\newblock \bibinfo{type}{{T}echnical {R}eport}.
  \bibinfo{institution}{arxiv.org/1704.00356v3}.
\newblock


\bibitem[{S{\'{a}}nchez-Fern{\'{a}}ndez} and Fisteus(2019)]%
        {fernandez19monotonicity}
\bibfield{author}{\bibinfo{person}{Luis {S{\'{a}}nchez-Fern{\'{a}}ndez}} {and}
  \bibinfo{person}{Jes{\'{u}}s~A. Fisteus}.} \bibinfo{year}{2019}\natexlab{}.
\newblock \showarticletitle{Monotonicity Axioms in Approval-based Multi-winner
  Voting Rules}. In \bibinfo{booktitle}{\emph{Proceedings of the 18th
  International Conference on Autonomous Agents and MultiAgent Systems, {AAMAS}
  2019}}. \bibinfo{pages}{485--493}.
\newblock


\bibitem[Suzuki and Vollen(2024)]%
        {suzuki2024maximum}
\bibfield{author}{\bibinfo{person}{Mashbat Suzuki} {and}
  \bibinfo{person}{Jeremy Vollen}.} \bibinfo{year}{2024}\natexlab{}.
\newblock \showarticletitle{Maximum Flow is Fair: {A} Network Flow Approach to
  Committee Voting}. In \bibinfo{booktitle}{\emph{Proceedings of the 25th {ACM}
  Conference on Economics and Computation, {EC} 2024}}.
  \bibinfo{pages}{964--983}.
\newblock


\bibitem[Tao et~al\mbox{.}(2022)]%
        {tao22local}
\bibfield{author}{\bibinfo{person}{Liangde Tao}, \bibinfo{person}{Lin Chen},
  \bibinfo{person}{Lei Xu}, {and} \bibinfo{person}{Weidong Shi}.}
  \bibinfo{year}{2022}\natexlab{}.
\newblock \showarticletitle{Local Differential Privacy Meets Computational
  Social Choice - Resilience under Voter Deletion}. In
  \bibinfo{booktitle}{\emph{Proceedings of the 31st International Joint
  Conference on Artificial Intelligence, {IJCAI} 2022}}.
  \bibinfo{pages}{3940--3946}.
\newblock


\bibitem[Thiele(1895)]%
        {Thie95a}
\bibfield{author}{\bibinfo{person}{Thorvald~N. Thiele}.}
  \bibinfo{year}{1895}\natexlab{}.
\newblock \showarticletitle{Om Flerfoldsvalg}.
\newblock In \bibinfo{booktitle}{\emph{Oversigt over Det Kongelige Danske
  Vidensk-Abernes Selskabs Forhandlinger}}. \bibinfo{pages}{415--441}.
\newblock


\bibitem[Xia(2020)]%
        {xia20smoothed}
\bibfield{author}{\bibinfo{person}{Lirong Xia}.}
  \bibinfo{year}{2020}\natexlab{}.
\newblock \showarticletitle{The Smoothed Possibility of Social Choice}. In
  \bibinfo{booktitle}{\emph{Proceedings of the 33rd International Conference on
  Neural Information Processing Systems, {NeurIPS} 2020}}.
\newblock


\end{thebibliography}

\newpage
\appendix

\section{\texorpdfstring{Proofs Elided from \Cref{sec:prelims}}{Proofs Elided from Section 2}} \label{app:prelims}

\etvdistanceToContinuity*

\begin{proof}%
    The proposition states that if a rule is $\eps$-$\Delta$-TV-stable, then $f$ is $\eps$-$\Delta$-continuous.
    Fix a randomized rule $f$, a pair of elections $E = (C,N,A,k)$, $E' = (C,N,A',k)$ such that $A$ and $A'$ differ for only one voter $v$ such that $\abs{A_v \oplus A_v'} \leq \Delta$. %
    
    Recall that $\W$ is the collection of all committees of size at most $k$ and let $\W_c = \{ W \in \W : c \in W\}$ be the set of committees containing $c$.
    Let $P_W \defeq \prob{f(E) = W}$ and $P_W' \defeq \prob{f(E') = W}$. 
    Also let $\mathcal{A} \subseteq \W$ be the collection of committees $W$ for which $P_W' \geq P_W$;
    this is equal to the set (an event) which witnesses the maximum in the first definition of $\dtv{\cdot}{\cdot}$ in \Cref{def:tv-dist}.
    Then
    \begin{align*}
        \abs{\pi_c(f(E)) - \pi_c(f(E'))} &= \abs{\sum_{W \in \W_c} P_W - \sum_{W \in \W_c} P_W'} \\
        &= \abs{\sum_{W \in \W_c \cap \mathcal{A}} (P_W-P_W') - \sum_{W \in \W_c\setminus \mathcal{A}} (P_W' - P_W)} \\
        &\leq \max \left( \sum_{W \in \W_c \cap \mathcal{A}} (P_W - P_W'), \:\sum_{W \in \W_c\setminus \mathcal{A}} (P_W' - P_W) \right)
        \\
        &\leq \max \left( \sum_{W \in \mathcal{A}} (P_W - P_W'), \:\sum_{W \in \W \setminus \mathcal{A}} (P_W' - P_W) \right)
        \\
        &= \sum_{W \in \mathcal{A}} (P_W-P_W')
        \\
        &= \dtv{f(E)}{f(E')} \leq \eps,
    \end{align*}
    Since $f$ is $\eps$-$\Delta$-TV-stable by assumption.
\end{proof}

\section{\texorpdfstring{Proofs Elided from \Cref{sec:noisy-GJCR}}{Proofs Elided from Section 3}} \label{app:GJCR}

\thmuniformEJR*

\begin{proof}
    Let $E=(C,N,A,k)$ be an approval-based committee election and let $E'=(C,N,A',k)$ differ from $E$ by one voter $v \in N$ approving one additional candidate $c \in C \setminus A_v$. Let $\Wf \subseteq \mathcal{W}$ (respectively $\Ws \subseteq \mathcal{W}$) be those committees satisfying EJR+ in $E$ (respectively $E'$). 

    First, we show that $W \in \Wf \setminus \Ws$ implies that $c \not\in W$. Assume for contradiction that $c \in W$. Note that $|A_i \cap W| \leq |A'_i \cap W|$ for all $i \in N$ and $N_{c'} = N'_{c'}$ for all $c'\in C \setminus W$, where $N'_{c'}$ denotes the supporter set of $c'$ in election $E'$. Therefore, if $c'$ witnesses an EJR+ violation for $E'$, the same EJR+ violation is present for the election $E$, a contradiction. Thus, $c \not\in W$. 

    Second, we show that $W \in \Ws \setminus \Wf$ implies that $c \in W$. Assume for contradiction that $c \notin W$. Then, $|A_i \cap W| = |A'
    _i \cap W|$ for all $i \in N$ and $N_{c'} \subseteq N'_{c'}$ for all $c' \in C \setminus W$. Consider a candidate $c'$ that provides an EJR+ violation for the election $E$. This EJR+ violation would still be an EJR+ violation for the election $E'$, a contradiction. Thus, $c \in W$. 

    Ex-ante monotonicity of \uejr follows now directly since $$\frac{|\Wf \cap \Wc|}{|\Wf|} \leq \frac{|\Ws \cap \Wc|}{|\Ws \cap \Wc| + |\Wf \setminus \Wc|} \leq \frac{|\Ws \cap \Wc|}{|\Ws \cap \Wc| + |\Ws \setminus \Wc|} = \frac{|\Ws \cap \Wc|}{|\Ws|},$$
    where the first inequality follows from the fact that $\Wf \cap \Wc \subseteq \Ws \cap \Wc$ and since the entire expression is smaller than $1$, and the second inequality follows from $\Ws \setminus \Wc \subseteq \Wf \setminus \Wc$. The first term equals the selection probability for candidate $c$ in election $E$ and the last term equals the selection probability for candidate $c$ in election $E'$. Thus, \textsc{UniformEJR+} is ex-ante monotone. 
\end{proof}

\gcjrEJRplus*

\begin{proof}
    First, we claim the output committee satisfies $\abs{W} \leq k$.
    This proceeds via a price system, where each candidate costs $1$, and each voter in $\{i \in N_c : \abs{A_i \cap W} < \ell\}$ is charged $(\scl)^{-1}$ for a candidate $c \in \LW$ sampled in round $\ell$.
    Fix a voter, and consider the last time in the course of \Cref{alg:gjcr-noisy} they contribute to the cost of a candidate, $c$. 
    If this transpires in round $\ell$, then since they approve of at most $\ell-1$ other candidates in $W$ and $\ell$ decreases over the course of the algorithm, their total expenditure is at most
    \begin{equation} \label{eq:gjcr-noisy-correctness}
        (\ell-1)\left(\frac{n\ell}{k+1}\right)^{-1} + (\scl)^{-1} < \ell\left(\frac{n\ell}{k+1}\right)^{-1} = \frac{k+1}{n}.
    \end{equation}
    Summing over all voters, the total expenditure is strictly less than $k+1$, and so $\abs{W} \leq k$. 

    Next we claim the output committee satisfies \ejrp{}.
    We will argue that, for any $\ell \in [k]$, after $k$ steps of the inner loop on \cref{line:GJCR-noisy-inner-loop}, no \ejrp{}-violating candidates with an $\ell$-large group of supporting voters remain.
    If for any round $r \in [k]$ we have $\LW = \emptyset$ or $\max_{d \in \LW} \sdl < n \frac{\ell}{k}$, then no such candidates remain and we are done.
    Otherwise for all rounds $r \in [k]$ we have $\sum_{d \in \LW} g_\ell(\sdl) \geq g_\ell(\frac{n\ell}{k})$, and so \Cref{alg:gjcr-noisy} samples a candidate in every round, and after the $r=k$ iteration we have $\abs{W} = k$. 
    
    At this point we may reuse the price system argument above.
    After these $k$ rounds in epoch $\ell$, how many voters $i$ remain for which $\abs{A_i \cap W} < \ell$? 
    Let the number of these voters be denoted $n'$, and observe that each of these voters was charged at most $\ell-1$ times 
    for the $k$ candidates chosen in round $\ell$, while all other voters 
    where charged at most $\ell$ times. And each time someone was charged, the amount was at strictly smaller than $\left(\frac{n\ell}{k+1}\right)^{-1}$. 
    Hence, the total amount of payments, $k$, can be upper bounded by 
    \[n' (\ell -1) \left(\frac{n\ell}{k+1}\right)^{-1} + (n - n') \ell \left(\frac{n\ell}{k+1}\right)^{-1} > k.\]

    This is equivalent to $n' < \frac{n\ell}{k+1}$, so $\LW = \emptyset$ at the end of epoch $\ell$. Therefore after the $r = k$ rounds all unchosen $c \not \in W$ have $\scl < \frac{n\ell}{k}$ and so do not violate \ejrp{} for $W$ at level $\ell$.
\end{proof}

\lemgcjrmonnext*

\begin{proof}
Let $c^*$ be the candidate that receives one additional approval in election $E'$ by voter $i^*$. 
Consider a fixed sequence $s$ of length $T$, and let $\ell$ be the epoch corresponding to iteration $T$. From now on, we assume that $\alg(E,T) = \alg(E',T) = s$. 
Given that $c^* \not \in s$, we establish the relationship between $\scsl$ and $\scslp$, where the latter refers to the number of $\ell$-underrepresented supporters in election $E'$ (at iteration $T$). For $c^*$ and any $d \neq c^*$, we claim that 
\begin{equation} \label{eq:gjcr-delta-vals}
    \scsl \leq \scslp \leq \scsl + 1
    \qquad \text{and} \qquad 
    \sdlp = \sdl.
\end{equation}
To see this, recall the definition of $\scl$ in \Cref{eq:gjcr-deltajell-defn}. 
Because $c^* \not \in s$, for all voters $i \in N$ we have that $A_i \cap W_s = A_i' \cap W_s$. And because $A$ and $A'$ differ only on the approval of $c^*$ by $i^*$, we have that $N_{c^*} \cup \{i^*\} =  N_{c^*}'$ and $N_d = N_d'$ for all $d \neq c$.

Since $g_\ell$ is increasing, it follows that $g_\ell(\scslp) \geq g_\ell(\scsl)$, and of course $g_\ell(\sdlp) = g_\ell(\sdl)$ for $d \neq c^*$.
Next consider the form of the probability $P_c$ in \cref{line:GJCR-noisy-sample}.
For $c^*$, the numerator grows at least at the same rate as the denominator, and so $P_{c^*}' \geq P_{c^*}$. 
For $d \neq c^*$, the numerator is constant while the denominator is non-decreasing, so $P_d' \leq P_d$. 
Finally, observe that $P_\bot = 1 - \sum_{c \in C} P_c$. Letting $G \defeq \sum_{c \in C} g_\ell(\scl)$, from \eqref{eq:gjcr-delta-vals} and since $g_\ell$ is increasing it also holds that $G' \geq G$.
Since 
\begin{equation} \notag
    \sum_{c \in C} P_c = \frac{G}{\max(G, g_\ell(\frac{n\ell}{k}))}
\end{equation}
is non-decreasing in $G$, it follows that $P_\bot' \leq P_\bot$.
The claim follows.
\end{proof}

\gcjrmonotone*

\begin{proof}    
    Let $E$ be any election and $E'$ be the election derived from $E$ where some voter $i^*$ additionally approves $c^* \in C \setminus A_{i^*}$. 
    We consider the tree of all possible sequences chosen by \Cref{alg:gjcr-noisy}, and prove that the probability of choosing $c^*$ increases at every node in the tree. 
    Our claim is that, for all $T$ and all sequences $s=(s_1, \ldots, s_T)$,
    \begin{equation} \label{eq:gjcr-mon-IH}
        \prob{c^* \in \alg(E') \given \alg(E',T) = s} \geq \prob{c^* \in \alg(E) \given \alg(E,T) = s}.
    \end{equation}
    We show this via induction on the tree of potential sequences $s$, where $s = ()$ is the root of the tree and parents are direct prefixes of children.

    \noindent\textbf{Leaf Case:} Consider $T = k^2 - 1$. Then for any $s$ such that $c^* \in s$ we have 
    \[
        \prob{c^* \in \alg(E') \given \alg(E',T) = s} = \prob{c^* \in \alg(E) \given  \alg(E,T) = s} = 1.
    \]
    Otherwise $c^* \not \in s$, and since this is the last step of the algorithm, for any such $s$ we have
    \begin{align}
        \prob{c^* \in \alg(E') \given \alg(E',T) = s} &= \prob{c^* = \alg(E')_{T+1} \given \alg(E',T) = s} \notag\\
        &\geq \prob{c^* = \alg(E)_{T+1} \given \alg(E,T) = s} \notag \\
        &= \prob{c^* \in \alg(E) \given \alg(E,T) = s}, \notag
    \end{align}
    by \Cref{lem:gjcr-mon-next}. This establishes \eqref{eq:gjcr-mon-IH} when $T = k^2-1$.

    \noindent\textbf{Inductive Step:} We now address $T < k^2 - 1$. 
    We assume that \eqref{eq:gjcr-mon-IH} holds for all sequences of length $T+1$.
    We will demonstrate that \eqref{eq:gjcr-mon-IH} then holds for $s$ with $\abs{s} = T$ also.
    Again if $c^* \in s$ we have 
    \[
        \prob{c^* \in \alg(E') \given \alg(E',T) = s} = \prob{c^* \in \alg(E) \given \alg(E,T) = s} = 1.
    \] 
    Otherwise $c^* \not\in s$. For the sequence $s$, let $s + u$ denote appending $u \in C \cup \{\bot\}$ to $s$. 
    Then, 
    \begin{align}
        &\prob{c^* \in \alg(E') \given \alg(E',T) =s} \notag \\
        &\qqqq= \sum_{u \in C\cup\{\bot\}} \prob{c^* \in \alg(E') \given \alg(E',T+1) = s+u} \cdot \prob{u = \alg(E')_{T+1} \given \alg(E',T) = s} \notag \\
        &\qqqq= \prob{c^* = \alg(E')_{T+1} \given \alg(E',T) = s}  \label{eq:continue-here}  \\
        &\qqqq\qq +  \sum_{u \neq c^*} \prob{c^* \in \alg(E') \given \alg(E',T+1) = s+u} \cdot \prob{u = \alg(E')_{T+1} \given \alg(E',T) = s}. \notag
        \intertext{Since $c^* \not \in s$, by \Cref{lem:gjcr-mon-next} we may write for any $u \neq c^*$
        $$\prob{u = \alg(E')_{T+1} \given \alg(E',T) = s} = \prob{u = \alg(E)_{T+1} \given \alg(E,T) = s} - \eps_u$$ for some $\eps_u \geq 0$ and
        $$\prob{c^* = \alg(E')_{T+1} \given \alg(E',T) = s} = \prob{c^* = \alg(E)_{T+1} \given \alg(E,T) = s} + \sum_{u \neq c} \eps_u.$$ 
        Therefore, we can lower bound \Cref{eq:continue-here} by}
        &\qqqq\geq \prob{c^* = \alg(E)_{T+1} \given \alg(E,T) = s}  \notag \\
        &\qqqq\qq +  \sum_{u \neq c^*} \prob{c^* \!\in \!\alg(E') \given \alg(E',T\!+\!1) = s\!+\!u} \cdot \left(\prob{u = \alg(E')_{T+1} \!\given \!\alg(E',T) = s} + \eps_u\right) \notag \\ 
        &\qqqq= \prob{c^* = \alg(E)_{T+1} \given \alg(E,T) = s}  \notag \\
        &\qqqq\qq +  \sum_{u \neq c^*} \prob{c^* \!\in \!\alg(E') \given \alg(E',T\!+\!1) = s\!+\!u} \cdot \left(\prob{u = \alg(E)_{T+1} \!\given \!\alg(E,T) = s}\right) \notag \\ 
        &\qqqq \geq \prob{c^* = \alg(E)_{T+1} \given \alg(E,T) = s}  \notag \\
        &\qqqq\qq +  \sum_{u \neq c^*} \prob{c^* \!\in \!\alg(E) \given \alg(E,T\!+\!1) = s\!+\!u} \cdot \left(\prob{u = \alg(E)_{T+1} \!\given \!\alg(E,T) = s}\right) \notag \\ 
        &\qqqq=\prob{c \in \alg(E) \given \alg(E,T) = s}, \notag 
    \end{align}
    where the last inequality follows from the inductive hypothesis \eqref{eq:gjcr-mon-IH} for $T+1$.
    
    Therefore \eqref{eq:gjcr-mon-IH} holds for all nodes of the `sequence tree' of the algorithm. Since 
    \[
        \prob{c^* \in \alg(E)} = \prob{c^* \in \alg(E) \given \alg(E,0) = ()},
    \]
    our main claim then corresponds to the root of the tree, where $s$ is the empty list.
\end{proof}

\gcjrsequencestable*

\begin{proof}[Proof of \Cref{lem:noisy-GJCR-sequence-stable}]
    As before, let $E$ and $E'$ be two elections that differ only in one voter $v$, and let $\Delta_v = |A_v \oplus A'_v|$. We use statistical coupling and the coupling lemma to make our proof rigorous. 
    (Recall \Cref{def:coupling} and \Cref{lem:coupling}.)
    For any (possibly partial) sequence, consider the distributions $\cN_s$ and $\cN_s'$ over $\Omega = C \cup \{\bot\}$. 
    Moreover, let $\cM_s$ be the optimal coupling of $\cN_s$ and $\cN_s'$ given by \Cref{lem:coupling}.
    We now define a coupling $\cT$ for $\cS_E$ and $\cS_{E'}$ via an algorithm which constructs a sample $(s, s') \sim \cT$:
    \begin{enumerate}
        \item For \emph{every} partial sequence $s$, let $(u_s, u_s')$ be a sample from $\cM_s$. %
        \item Starting at $s=()$ and $s'=()$, iteratively construct $s \leftarrow s + u_{s}$ and $s' \leftarrow s' + u_{s'}'$.
        \item Return $(s, s')$.
    \end{enumerate}
    
    This is a coupling of $\cS_E$ and $\cS_{E'}$ because for every partial sequence $s$, the sample which is appended to $s$ is the $u_s$ in the joint sample $(u_s, u_s')\sim \cM_s$, and so it is distributed according to $\cN_s$. Therefore the final sequence $s$ is distributed according to $\cS$, and the same is true for $s'$ and $\cS_{E'}$.

    Now, let $s\vert t$ denote the length-$t$ prefix of $s$. For this coupling $\cT$ we have
    \begin{align}
        \probover{(s,s')\sim\cT}{s \neq s'} &= \prob{\bigvee_{t \in [k^2]} \left[ u_{s\vert t} \neq u_{s'\vert t}' \right]} \notag \\
        &\leq \sum_{t \in [k^2]} \probover{\cM_{s \vert t}}{u_{s\vert t} \neq u_{s\vert t}'} \notag \\
        &= \sum_{t \in [k^2]} \dtv{\cN_{s\vert t}}{\cN_{s'\vert t}}, \notag 
        \intertext{by a union bound and using that the $\cM_s$ are optimal couplings.
        We now apply \Cref{lem:gjcr-stability-per-step} with our choice of $a_\ell$ from \eqref{eq:gjcr-a-defn}. We also rewrite the sum in terms of epoch $\ell$ and round $r$ instead of using the flattened index $t$. 
        }
        &\leq \sum_{r=1}^k \sum_{\ell = 1}^k 4(2 e - 1) \cdot a_\ell \notag \\
        &= 4(2 e - 1) \sum_{r=1}^k \sum_{\ell = 1}^k  \frac{k(k+1)}{n\ell}\left(\log n + \log \Delta \right) \notag \\
        &= 4(2 e - 1) \sum_{\ell = 1}^k \frac{k^2(k+1)}{n\ell}\left(\log n + \log \Delta \right) \notag \\
        &= O\left( \frac{k^3}{n} \cdot \log k \cdot \left(\log n + \log \Delta \right) \right) \notag, 
    \end{align}
    where the last equation follows from upper bounding the harmonic number by the logarithm. 
    To conclude, we have demonstrated a coupling with bounded `off-diagonal' probability, and by \Cref{lem:coupling} this `off-diagonal' probability is an upper bound on the TV distance of $\cS$ and $\cS_{E'}$. %
\end{proof}

\JRcandcontLB*

\begin{proof}[Proof of \Cref{thm:JR-cand-cont-lb}]
    In order to show an asymptotic lower bound $\Omega\left(\frac{k}{n}\right)$ we will use the following family of instances defined for every $n, k \in \mathbb{N}$ such that $n$ is divisible by $k^2$. We remark that all single-voter changes that we carry out below only change one approval per voter, hence we are in the continuity model where $\Delta = 1$.
    We define $n = h \cdot k^2$ voters and a set $C$ of $k+1$ candidates.
    We create $k+1$ groups of voters of size $\frac{n}{k} - h = \frac{n}{k} - \frac{n}{k^2}$ such that the $i$-th group of voters approves only the $i$-th candidate.
    In this way, the number of voters that approve at least one candidate is equal to
    $(k+1) \left(\frac{n}{k} - \frac{n}{k^2} \right)
    = n - \frac{n}{k^2} = n-h$.
    That means that exactly $h$ voters approve no candidate.
    In the instance every candidate $c$ has the same approval score, i.e., $|N_c| = \frac{n}{k} - h$.

    We consider a rule $f$ which satisfies \jr{} applied to any instance from the family.
    A rule $f$, by definition, outputs a random committee of size at most $k$
    which induces a set of selection probabilities $(\pi_{c})_{c \in C}$,
    where $\sum_{c \in C} \pi_{c} \leq k$.
    Since we have $k+1$ candidates
    there exists a candidate $c^*$ among them such that
    $\pi_{c^*} \leq \frac{k}{k+1}$.

    Now, we modify the considered instance by adding $h$ approvals to $c^*$ in the votes of the voters without approvals.
    In this way, $c^*$ has $\frac{n}{k}$ supporters,
    but each other candidate has only $\frac{n}{k} - \frac{n}{k^2}$ approvals.
    That means that there exists exactly one $1$-large group of voters approving a common candidate---the group of voters approving $c^*$.
    Since $f$ satisfies \jr{} and voters from the $1$-large group approve only $c^*$, it must be the case that $c^*$ is chosen to the committee by $f$ in the modified instance, i.e., its selection probability becomes $1$.
    
    Therefore, the selection probability change for selecting $c^*$ increased by at least $\frac{1}{k+1}$ while adding $h = \frac{n}{k^2}$ approvals.
    It means that there exists an approval added which changed the selection probability of $c^*$ by at least $\Big( \frac{1}{k+1} \Big) / \left( \frac{n}{k^2} \right) \leq \frac{k}{n}$.
    Hence the lowerbound follows.
\end{proof}

\noisyContLB*

\begin{proof}[Proof of \Cref{thm:noisy-rules-cand-cont-lb}]
    In order to show an asymptotic lower bound $\Omega\left(\frac{k^2}{n}\right)$ we will use a family of instances which is analogous to the family from the proof of \Cref{thm:JR-cand-cont-lb} but instead of $k+1$ candidates we define $k+1$ pairs of candidates. 
    In this way we will force our rules to split probability $k/(k+1)$ into a pair of candidates, hence some candidate has selection probability at most $0.5$ which later results in an increase of probability by a factor $k$ larger than in the construction with only $k+1$ candidates. Again, we remark that all single-voter changes that we carry out below only change one approval per voter, hence we are in the continuity model where $\Delta = 1$.

    Formally, we define a family of instances for every $n, k \in \mathbb{N}$ such that $n$ is divisible by $k^2$.
    We define $n = h \cdot k^2$ voters and $m = 2(k+1)$ candidates $\{ c_1, c_2, \dots, c_m\} = C$ as follows.
    We create $k+1$ groups of voters of size $\frac{n}{k} - h = \frac{n}{k} - \frac{n}{k^2}$ such that the $i$-th group of voters approves only candidates $c_{2i-1}$ and $c_{2i}$.
    In this way, the number of voters that approve at least one candidate is equal to
    $(k+1) \left(\frac{n}{k} - \frac{n}{k^2} \right)
    = n - \frac{n}{k^2} = n-h$.
    That means that exactly $h$ voters approve no candidate.
    In the instance every candidate $c \in C$ has the same approval score, i.e., $|N_c| = \frac{n}{k} - h$.
    
    First, we consider \noisygjcr~(\Cref{alg:gjcr-noisy}) on an instance from the family.
    From the definition of a rule,
    \noisygjcr outputs a random committee of size at most $k$
    which induces a set of selection probabilities $(\pi_{c})_{c \in C}$,
    where $\sum_{c \in C} \pi_{c} \leq k$.
    Since we have $2(k+1)$ candidates,
    there exists a candidate $c^* \in C$ such that
    $\pi_{c^*} \leq \frac{k}{2(k+1)} < \frac{1}{2}$.

    Now, we modify the considered instance by adding $h$ approvals to $c^*$ in votes of the voters without approvals.
    In this way, $c^*$ has $\frac{n}{k}$ supporters,
    but each other candidate has only $\frac{n}{k} - \frac{n}{k^2}$ approvals.
    It means that none of the candidates is included in $\cL$ (\cref{line:alg:noisy-gjcr:defining-l}) for all epochs with $\ell \in (k,k-1,\dots,2)$ (\cref{line:alg:noisy-gjcr:outer-loop}).
    Only when $\ell = 1$ we have $n_{c^*, 1} = \frac{n}{k} > \frac{n}{k+1}$ and $n_{c, 1} = \frac{n}{k} - \frac{n}{k^2} < \frac{n}{k+1}$ for every $c \neq c^*$.
    Therefore, \noisygjcr defines $\cL \leftarrow \{c^*\}$ (\cref{line:alg:noisy-gjcr:defining-l}) and
    since $\ell=1$ we have
    $\max \left(\sum_{d \in \cL} g_{\ell}(n_{d\ell}), \:g_{\ell}(\frac{n\ell}{k}) \right) = \max \left( g_1(n_{c^*,1}), \:g_1(\frac{n}{k}) \right) = g_1(n_{c^*,1})$.
    It implies that $P_{c^*} = 1$ (\cref{line:GJCR-noisy-sample}), i.e., $c^*$ is included in $W$ with probability~$1$.

    Therefore, the selection probability change for selecting $c^*$ increased from a value at most $\frac{1}{2}$ to $1$ while adding $h = \frac{n}{k^2}$ approvals.
    It means that there exists an approval added which changed the selection probability by at least $\frac{1-0.5}{h} = \frac{k^2}{2n}$.
    Hence the lowerbound follows for \noisygjcr.

    The arguments for \noisygjcrcap{}~(\Cref{alg:gjcr-noisy-cap}) for every $\lm \in [k]$ are the same since every such algorithm considers $\ell=1$ as the last epoch and this is the only epoch when $\cL$ is not empty (\cref{line:alg:noisy-gjcr:defining-l}).
\end{proof}

Finally, we remark that the same lower bound holds for \noisygjcrcap{}, an algorithm we introduce in \Cref{app:more-stability}. 
Because $\eps$-TV stability implies $\eps$-continuity by \Cref{obs:marginalseasier}, our continuity lower bounds imply lower bounds for $\eps$-TV stability.

\section{\texorpdfstring{Proofs Elided from \Cref{sec:more-stability}}{Proofs Elided from Section 4}} \label{app:more-stability}

\subsection{\texorpdfstring{Proofs Elided from \Cref{sec:gjcr-slack}}{Proofs Elided from Section 4.1}}

\label{app:gjcr-slack}

\begin{algorithm}
\caption{\noisygjcrslack for $\alpha \in (0,1]$}
\label{alg:gjcr-noisy-slack}
$W \gets \emptyset$\;
\For{epoch $\ell$ in $(k, k-1, \dots, 1)$}{
    \For{round $r \in [k]$} {
        $\LW \leftarrow \left\{c \in C \setminus W: \scl > \frac{n \ell}{k+1} \right\}$\;
        Sample $c\sim \LW$ with probability $P_c \defeq \frac{g_\ell(\scl)}{\max\left(\sum_{d \in \LW} g_\ell(\sdl), \: g_\ell(\frac{n\ell}{\alpha \cdot k})\right)}$, and $c \leftarrow \bot$ otherwise\;
        \If{$c \neq \bot$}{
            $W \gets W \cup \{c\}$\;
        }
    }
}
\Return $W$\;
\end{algorithm}

\gcjrSlackThm* 

\begin{proof}
We start by arguing why (i) holds. The proof that \noisygjcr{} satisfies EJR+ can be divided into three parts: Part (a) shows that \noisygjcr{} selects at most $k$ candidates. This part of the proof carries over exactly as is to \noisygjcrslack{}, since the two algorithms are the same besides the definition of the probabilities $P_c$. However, the pricing system which is used for part (a) of the proof is not dependent on these probabilities.  Part (b) shows that if in some epoch $\ell$ selecting $\bot$ receives non-zero probability, then we satisfy EJR+. This statement is not true for \noisygjcrslack{} and we will address it below. Part (c) shows that if for some epoch $\ell$, the dummy candidate $\bot$ never receives positive utility, then we select $k$ candidates in that epoch and satisfy EJR+ for that $\ell$. This is part of the proof is still true and caries over exactly as it is to  \noisygjcrslack{}, since again, it uses the pricing scheme which is independent from the probabilities $P_c$. 

We now turn to proving a relaxed version of part (b). We will argue that, for any epoch $\ell \in [k]$, after $k$ rounds of the inner loop on \cref{line:GJCR-noisy-inner-loop}, no $\alpha$-\ejrp{}-violating for candidates with an $\frac{\ell}{\alpha}$-large  group of supporting voters exists. If for any iteration $r \in [k]$ the dummy candidate $\bot$ receives positive probability, which means that $\LW = \emptyset$ or $\max_{c \in \LW} \scl <\frac{n\ell}{\alpha k}$, then, no candidate can witness a violation towards $\alpha$-\ejrp{} for an $\frac{\ell}{\alpha}$-large group. Otherwise for all rounds $r \in [k]$ we have $\sum_{d \in \LW} g_\ell(\scl) \geq g_\ell(\frac{n\ell}{\alpha k})$, and so \Cref{alg:gjcr-noisy} samples a candidate in every step, and after the $r=k$ iteration we have $\abs{W} = k$. In this case, we can continue with part (c) of the proof, as discussed above.

For claim (ii), i.e., \noisygjcrslack{} is neutral and anonymous, note that it follows directly from the definition of the algorithm, since it is only based on the approval sets and does not do any decisions depending on the identities of the agents or candidates.

Lastly, we turn towards arguing why (iii) holds. Recall that we showed this statement for \noisygjcr{} by first showing that, for any (partial) sequence $s$, the \emph{next candidate} distribution is monotone (see \Cref{lem:gjcr-mon-next}). Recall that the only difference between \noisygjcr{} and \noisygjcrcap{} is that the denominator in the definition of $P_c$ is defined differently, namely, $\frac{G}{\max(G,g_{\ell}(\frac{n\ell}{ \alpha k}))}$ instead of $\frac{G}{\max(G,g_{\ell}(\frac{n\ell}{k}))}$. However, this change does not change any of the statements regarding the behavior of these terms in the proof of \Cref{lem:gjcr-mon-next}. In fact, the proof would hold for any $P_c$ of this form, where the term in the maximum function is any constant. Thus, \Cref{lem:gjcr-mon-next} is true for \noisygjcrslack{} as well. Since \Cref{thm:gjcr-mon} uses the lemma as a blackbox, the fact that \noisygjcrslack{} is monotone follows as well. 
\end{proof}

Before we turn towards proving \Cref{thm:noisy-GJCR-slack-continuous}, we will formalize and proof a variant of \Cref{lem:gjcr-slack-stability-per-step} for \noisygjcrslack{}. Let $s$ be some sequence of samples and let $\cN_s$ and $\cN'_s$ be the \emph{next candidate} distributions of \noisygjcrslack{} given that $s$ was selected up to some point. 

\begin{lemma} \label{lem:gjcr-slack-stability-per-step}
    Let $E$ and $E'$ be two elections that differ only in the ballot of voter $v$, and let $\cN_s, \cN_s'$ and $\Delta_v$ be defined as described above. 
    For any partial sequence of samples $s$ in \Cref{alg:gjcr-noisy-slack}, it holds that  
    \begin{equation} \label{eq:gjcr-slack-noisy-round-continuity}
        \dtv{\cN_s}{\cN_s'} \leq \gamma \cdot a_\ell + 4 \cdot \exp\left(\log(\Delta_v) -a_{\ell}\frac{n\ell}{k(k+1)}\left( \frac{k+1}{\alpha} - k\right) + a_{\ell} \right).
    \end{equation}
    for some universal constant $\gamma$ (in particular, $\gamma = 4(2e-1)$).
\end{lemma}

The proof is largely analogous to one of \Cref{lem:gjcr-stability-per-step}. 

\begin{proof}[Proof of \Cref{lem:gjcr-slack-stability-per-step}]
    We begin with an observation about the regime of $a_\ell$ for which this statement requires proof.
    \begin{observation} \label{obs:a-ell-small-wlog2}
        If $a_\ell \geq 1$, then \Cref{lem:gjcr-slack-stability-per-step} is trivially satisfied (e.g., with $\gamma=1$).
    \end{observation}
    This holds because $\dtv{\cdot}{\cdot} \leq 1$. Therefore, we will assume that $a_{\ell} \leq 1$ holds in the following.

    We first establish a pair of useful bounds and notation. In the following we assume that the sequence $s$ of length $T$ is fixed and all of the following notation is under the assumption that the algorithm is in step $T+1$ and $\alg(E,T) = \alg(E',T) = s$, e.g., we denote by $\scl$ ($\sclp$, respectively) the variable in step $T+1$ of \Cref{alg:gjcr-noisy} when running it for input $E$ ($E'$, respectively). 
    Then, 
    \begin{equation} \label{eq:gjcr-delta-vals-general-app}
        \scl - 1 \leq \sclp \leq \scl + 1, 
    \end{equation}
    which holds because the only voter that can change membership in $\Scl$ when going from $E$ to $E'$ is voter $v$. 
    From \eqref{eq:gjcr-delta-vals-general-app}, the definition of $g_\ell$, and $a_\ell \leq 1$ (see \Cref{obs:a-ell-small-wlog2}) we have that for $c\in C$,
    \begin{equation} \label{eq:gjcr-gdelta-bound-app}
        \abs{g_\ell(\sclp)- g_\ell(\scl)} \leq  g_\ell(\scl) \cdot \max(1 - e^{-a_\ell}, \:e^{a_\ell} - 1)  \leq g_\ell(\scl) \cdot a_\ell \cdot (e-1).
    \end{equation}
    We will also apply the fact that 
    \begin{equation}
    |\LW \oplus \LW'| \leq \Delta_v, \label{eq:delta-app}
    \end{equation}
     which is due to the fact that $c \in \LW \oplus \LW'$ implies $\scl \neq \sclp$, which in turn implies that voter $v$ votes differently for $c$ in the two elections $E$ and $E'$. For notational convenience, we define 
    \begin{equation}
        G \defeq \max\left(\sum_{d \in \LW} g_\ell(\sdl), \: g_\ell\left(\frac{n\ell}{\alpha k}\right)\right)
    \end{equation}
    to be the denominator of the probabilities in \cref{line:GJCR-noisy-sample}, and define $G'$ analogously for $\sdlp$ and $\LW'$. We observe that \begin{equation}
        |G' - G| \leq \sum_{c \in \LW' \cap \LW} \abs{g_\ell(\sclp) - g_\ell(\scl)} + \sum_{c \in \LW' \oplus \LW} \max(g_\ell(\sclp),g_\ell(\scl)). \label{eq:Gdifference-app}
    \end{equation}

    {We are now ready to bound the difference between $\cN_s$ and $\cN_s'$. Since $P_\bot = 1 - \sum_{c \in C} P_c$, to bound $\dtv{\cN_s}{\cN_s'}$ we may incur a factor of two and focus on $c\in C$. Using this, we have

    \begin{align}
        \frac{1}{2}\dtv{\cN_s}{\cN_s'}
        &\leq  \sum_{c \in C} \abs{P_c' - P_c} %
        =  \sum_{c \in \LW' \cup \LW} \abs{P_c' - P_c} \notag \\
        &=  \sum_{c \in \LW' \cap \LW} \abs{\frac{g_\ell(\sclp)}{G'} - \frac{g_\ell(\scl)}{G}} +  \sum_{c \in \LW' \oplus \LW} \max(P_c', \: P_c). \notag
        \intertext{We now apply \Cref{fac:quotientdiff} to the first sum, which yields} 
        &\leq \sum_{c \in \LW' \cap \LW} g_\ell(\sdlp) \abs{\frac{1}{G'} - \frac{1}{G}} + \frac{1}{G} \sum_{c \in \LW' \cap \LW} \abs{g_\ell(\sclp) - g_\ell(\scl)} +  \sum_{c \in \LW' \oplus \LW} \max(P_c', \: P_c). \notag \\ 
        \intertext{For the next inequality, we rewrite the first sum as $\sum_{c \in \LW' \cap \LW} \frac{g_\ell(\sclp)}{G'} \frac{\abs{G' - G}}{G}$ which is smaller or equal to $ \frac{\abs{G' - G}}{G}$ since $\sum_{c \in \LW' \cap \LW} g_\ell(\sclp) \leq G'$. We also upper bound the second sum, which will become helpful later on. We get }
        &\leq \frac{\abs{G' - G}}{G} + \frac{1}{G} \sum_{c \in \LW' \cap \LW} \abs{g_\ell(\sclp) - g_\ell(\scl)} + \sum_{c \in \LW' \oplus \LW} \frac{\max(g_\ell(\scl),g_\ell(\sclp))}{\min(G,G')}. \notag \\ 
        \intertext{Next, we apply our previously established upper bound on $|G'-G|$ from \eqref{eq:Gdifference-app} and directly combine these terms with the right-hand-side of the expression.}
        &\leq 2\left(\frac{1}{G} \cdot \sum_{c \in \LW' \cap \LW} \abs{g_\ell(\sclp) - g_\ell(\scl)} + \sum_{c \in \LW' \oplus \LW} \frac{\max(g_\ell(\scl),g_\ell(\sclp))}{\min(G,G')} \right). \notag \\ 
        \intertext{Then, we apply \eqref{eq:gjcr-gdelta-bound-app} to the first sum. Moreover, we upper bound each summand of the second sum by the same upper bound, using the following  observation: When $c \in \LW' \oplus \LW$, then $\max(\scl,\sclp) \leq \frac{n\ell}{k+1} + 1$, which holds since \scl{} and \sclp{} can only differ by $1$. Also, $\min(G,G') \geq g_\ell(\frac{n\ell}{k})$. Lastly, by \eqref{eq:delta-app}, we know that the number of summands is bounded by $|\LW'\oplus\LW| \leq \Delta_v$. This yields }
        &\leq 2 \left((e-1) \cdot a_\ell \cdot \sum_{c \in \LW'\cap \LW} \frac{g_{\ell}(\scl)}{G} + \Delta_v \frac{g_{\ell}\left(\frac{n\ell}{k+1}+ 1\right) }{g_{\ell}(\frac{n\ell}{ \alpha k})} \right), \notag \\ 
        \intertext{where we can upper bound the sum by $1$ again and rewrite the right-hand-side by applying the definition of $g_{\ell}$. This yields}
        & \leq 2\left( (e-1) \cdot a_{\ell} + \exp\left(\log(\Delta_v) - a_{\ell} \frac{n\ell}{k(k+1)} \left( \frac{k+1}{\alpha} - k\right) + a_{\ell} \right) \right), \notag
    \end{align}    
    which, after multiplying the entire inequality by the factor of $2$, yields the claim.}
    \end{proof}

We now discuss our choice of $a_\ell$. 
We %
define 
\begin{equation} %
    a_\ell \defeq \frac{k+1}{n\ell \left(1 - \alpha + \frac{2-\alpha}{k} \right)}\log\left(n \cdot \Delta \right). 
\end{equation} 

Recall that by \Cref{obs:a-ell-small-wlog}, the lemma is trivially true for all $a_\ell \geq 1$. Therefore, we assume in the following that $a_\ell \leq 1$. Moreover, observe that $\frac{1}{\alpha} \geq 2 - \alpha$. Therefore, the second term can be upper bounded by

\begin{align*}
    \exp\left(\log(\Delta_v) -a_\ell\frac{n\ell}{k(k+1)}\left(\frac{k+1}{\alpha} - k\right) + a_\ell \right) &\leq \exp\left(\log(\Delta) -a_\ell\frac{n\ell}{k(k+1)}\left(\frac{k+1}{\alpha} - k\right) + a_\ell \right) \\
    & \leq \exp\left(\log(\Delta) -a_\ell\frac{n\ell}{k(k+1)}k\left(1 - \alpha + \frac{2-\alpha}{k} \right) + a_\ell \right) \\ 
    &=\exp \left( \log(\Delta) - \log\left(n \cdot \Delta \right) + a_\ell \right) \\
    &= \exp\left(\log n + a_\ell \right) = \frac{1}{n} \cdot e^{a_\ell}\\
    & \leq \frac{e}{n} \leq a_\ell \cdot e.
\end{align*}
Plugging this into \Cref{lem:gjcr-stability-per-step} yields a bound of $4(2e -1) a_{\ell}$  on the TV-distance of $\cN_s$ and $\cN'_s$.

\begin{theorem}\label{lem:GJCR-slack-sequence-stable}
    Assuming all single-voter profile changes satisfy $\Delta_v \leq \Delta$,
    the \noisygjcrslack{} sequence distribution $\mathcal{S}_E$ is $\eps$-TV-stable for $$\eps = \gamma \cdot \frac{k^2}{n\left(1 - \alpha + \frac{2-\alpha}{k} \right)} \cdot \log k \cdot \left(\log n + \log \Delta \right)\text{, for some universal constant }\gamma.$$
\end{theorem}

\begin{proof}
Again, this proof follows almost exactly like the proof of \Cref{lem:GJCR-slack-sequence-stable}. The first part of the proof which defines the coupling $\cT$ can be followed exactly as is. For completeness, we do state the last part of the proof here, where we apply the union bound for the new parameter $a_{\ell}$. 

Now, let $s\vert t$ denote the length-$t$ prefix of $s$. For this coupling $\cT$ we have
    \begin{align}
        \probover{(s,s')\sim\cT}{s \neq s'} &= \prob{\bigvee_{t \in [k^2]} \left[ u_{s\vert t} \neq u_{s'\vert t}' \right]} \notag \\
        &\leq \sum_{t \in [k^2]} \probover{\cM_{s \vert t}}{u_{s\vert t} \neq u_{s\vert t}'} \notag \\
        &= \sum_{t \in [k^2]} \dtv{\cN_{s\vert t}}{\cN_{s'\vert t}}, \notag 
        \intertext{by a union bound and using that the $\cM_s$ are optimal couplings.
        We now apply \Cref{lem:gjcr-stability-per-step} with our choice of $a_\ell$ from \eqref{eq:gjcr-a-defn}. We also rewrite the sum in terms of epoch $\ell$ and round $r$ instead of using the flattened index $t$. 
        }
        &\leq \sum_{r=1}^k \sum_{\ell = 1}^k 4(2 e - 1) \cdot a_\ell \notag \\
        &= 4(2 e - 1) \sum_{r=1}^k \sum_{\ell = 1}^k  \frac{k+1}{n\ell \left(1 - \alpha + \frac{2-\alpha}{k} \right)}\left(\log n + \log \Delta \right) \notag \\
        &= 4(2 e - 1) \sum_{\ell = 1}^k \frac{k(k+1)}{n\ell\left(1 - \alpha + \frac{2-\alpha}{k} \right)}\left( \log n + \log \Delta \right) \notag \\
        &= O\left( \frac{k^2}{n \left(1 - \alpha + \frac{2-\alpha}{k} \right)} \cdot \log k \cdot \left(\log n + \log \Delta\right) \right) \notag, 
    \end{align}
    where the last equation follows from upper bounding the harmonic number by the logarithm. 
\end{proof}

Given the proof of \Cref{lem:GJCR-slack-sequence-stable}, the following theorem follows as a direct corollary. 

\thmgjcrslackcontinuous*

\subsection{\texorpdfstring{Proofs Elided from \Cref{sec:noisy-greedy-JR}}{Proofs Elided from Section 4.2}}

\label{app:gjcr-capped}

\begin{algorithm}
\caption{\noisygjcrcap}
\label{alg:gjcr-noisy-cap}
$W \gets \emptyset$\;
\For{epoch $\ell$ in $(\ell_{max}, \ell_{max}-1, \dots, 1)$}{ \label{line:alg:gjcr-noisy-cap:outer-loop}
    \For{round $r \in [k]$} { \label{line:gjcr-noisy-cap:inner-loop}
        $\LW \leftarrow \left\{c \in C \setminus W: \scl > \frac{n \ell}{k+1} \right\}$\;\label{line:alg:gjcr-noisy-cap:defining-l} 
        Sample $c\sim \LW$ with probability $P_c \defeq \frac{g_\ell(\scl)}{\max\left(\sum_{d \in \LW} g_\ell(\sdl), \: g_\ell(\frac{n\ell}{k})\right)}$, and $c \leftarrow \bot$ otherwise\; \label{line:gjcr-noisy-cap:sample}
        \If{$c \neq \bot$}{
            $W \gets W \cup \{c\}$\;
        }
    }
}
\Return $W$\;
\end{algorithm}

\gcjrCappedThm*

\begin{proof}
    We first note that (ii) follows directly from the definition of the algorithm, since it is only based on the approval sets and does not do any decisions depending on the identities of the agents or candidates.

    We now turn towards arguing why (i) holds. This proof is a straightforward adaptation of our proof of \Cref{thm:gjcr-noisy-correct}.
    The fact that the output committee satisfies $\abs{W} \leq k$ follows because of the exact same arguments as in the proof of \Cref{thm:gjcr-noisy-correct}. 
    
    Next for showing that \noisygjcrcap{} satisfies \lm-capped \ejrp{}, we observe that the proof of \Cref{thm:gjcr-noisy-correct} argues that after any epoch $\ell \in [k]$ (which consists of $k$ rounds), no EJR+ violation of an $\ell$-large group can exist. The very same arguments go through for \Cref{alg:gjcr-noisy-cap} and for all $\ell \in [\lm]$. It follows immediately that \noisygjcrcap{} satisfies $\lm$-capped EJR+.

    Lastly, we argue why \noisygjcrcap{} is monotone. Recall that we showed this statement for \noisygjcr{} by first showing that, for any (partial) sequence $s$, the \emph{next candidate} distribution is monotone (see \Cref{lem:gjcr-mon-next}). With respect to this proof, the only difference between \noisygjcr{} and \noisygjcrcap{} is that a complete sequence for \noisygjcr{} is of length $k^2$, while a complete sequence for \noisygjcrcap{} is of length $\lm \cdot k$. Other than that, the proof of \Cref{lem:gjcr-mon-next} and the proof of \Cref{thm:gjcr-mon} work exactly the same for \noisygjcrcap{}. 
\end{proof}

\begin{theorem}\label{lem:GJCR-capped-sequence-stable}
    Assuming all single-voter profile changes satisfy $\Delta_v \leq \Delta$,
    the \noisygjcrcap{} sequence distribution $\mathcal{S}_E$ is $\eps$-TV-stable for $$\eps = \gamma \cdot \frac{k^3}{n}(\log(\lm) + 1) \cdot \left(\log n + \log \Delta \right)\text{, for some universal constant }\gamma.$$
\end{theorem}

\begin{proof}
    Recall how we proved the continuity guarantee for \noisygjcr{}: We consider two elections $E$ and $E'$ that only differ in one voter $v$ who differs in $\Delta_v = |A_v \oplus A'_v|$ approvals. Then, we showed in \Cref{lem:gjcr-stability-per-step} that if we fix a sequence of samples $s$, then the \emph{next candidate} distributions $\mathcal{N}_s$ and $\mathcal{N}'_s$ have a bounded TV-distance. Again, the proof of this lemma goes through exactly as is for the case of \noisygjcrcap{}, with the only difference again that the maximum value of $\ell$ is $\lm$ and sequences are of length at most $k \lm$.

    Now, we turn towards proving a variant of \Cref{lem:noisy-GJCR-sequence-stable} for \noisygjcrcap{}. That is, we want to argue about the stability of the sequence distribution $\cS_{E}$. As before, let $E$ and $E'$ be two elections that differ only in one voter $v$, and let $\Delta_v = |A_v \oplus A'_v|$. We use statistical coupling and the coupling lemma to make our proof rigorous. 
    (Recall \Cref{def:coupling} and \Cref{lem:coupling}.)
    For any (possibly partial) sequence, consider the distributions $\cN_s$ and $\cN_s'$ over $\Omega = C \cup \{\bot\}$. 
    Moreover, let $\cM_s$ be the optimal coupling of $\cN_s$ and $\cN_s'$ given by \Cref{lem:coupling}. We now define a coupling $\cT$ for $\cS_E$ and $\cS_{E'}$ via an algorithm which constructs a sample $(s, s') \sim \cT$:
    \begin{enumerate}
        \item For \emph{every} partial sequence $s$, let $(u_s, u_s')$ be a sample from $\cM_s$. %
        \item Starting at $s=()$ and $s'=()$, iteratively construct $s \leftarrow s + u_{s}$ and $s' \leftarrow s' + u_{s'}'$.
        \item Return $(s, s')$.
    \end{enumerate}
    
    This is a coupling of $\cS_E$ and $\cS_{E'}$ because for every partial sequence $s$, the sample which is appended to $s$ is the $u_s$ in the joint sample $(u_s, u_s')\sim \cM_s$, and so it is distributed according to $\cN_s$. Therefore the final sequence $s$ is distributed according to $\cS_E$, and the same is true for $s'$ and $\cS_{E'}$.

    Now, let $s\vert t$ denote the length-$t$ prefix of $s$. For this coupling $\cT$ we have
    \begin{align}
        \probover{(s,s')\sim\cT}{s \neq s'} &= \prob{\bigvee_{t \in [k \cdot \lm]} \left[ u_{s\vert t} \neq u_{s'\vert t}' \right]} \notag \\
        &\leq \sum_{t \in [k \cdot \lm]} \probover{\cM_{s \vert t}}{u_{s\vert t} \neq u_{s\vert t}'} \notag \\
        &= \sum_{t \in [k \cdot \lm]} \dtv{\cN_{s\vert t}}{\cN_{s'\vert t}}, \notag 
        \intertext{by a union bound and using that the $\cM_s$ are optimal couplings.
        We now apply \Cref{lem:gjcr-stability-per-step} with our choice of $a_\ell$ from \eqref{eq:gjcr-a-defn}. We also rewrite the sum in terms of epoch $\ell$ and round $r$ instead of using the flattened index $t$. 
        }
        &\leq \sum_{r=1}^k \sum_{\ell = 1}^{\lm} 4(2 e - 1) \cdot a_\ell \notag \\
        &= 4(2 e - 1) \sum_{r=1}^k \sum_{\ell = 1}^{\lm}  \frac{k(k+1)}{n\ell}\left(\log n + \log \Delta \right) \notag \\
        &= 4(2 e - 1) \sum_{\ell = 1}^{\lm} \frac{k^2(k+1)}{n\ell}\left(\log n + \log \Delta \right) \notag \\
        &= 4(2 e - 1) \frac{k^3}{n} \cdot (\log(\lm) + 1) \cdot \left(\log n + \log \Delta \right) \notag, 
    \end{align}
    where the last equation follows from bounding the harmonic number of $\lm$ by $\log(\lm) + 1$. 
\end{proof}

Given the proof of \Cref{lem:GJCR-capped-sequence-stable}, the following theorem follows as a direct corollary. 

\thmgjcrcappedcontinuous*

\subsection{\texorpdfstring{Proofs Elided from \Cref{sec:exp-PAV}}{Proofs Elided from Section 4.3}}

\label{app:softmax-pav}
\noisyPAVmonotone* 

\begin{proof}[Proof of \Cref{thm:exp-PAV-monotonicity}]
    Since (i) and (ii) follows immediately from the definition of the algorithm, we focus on (iii). As in the proof of \Cref{thm:exp-pav-continuous}, we start with an arbitrary election $E$ and let $E'$ be the same election with one approval $E_{ij}$ flipped from $0$ to $1$. Let $c=j$ denote the candidate for which this switch occurs. 
    Let $\pavp{W}$ denote the PAV score of $W$ under $E'$. 
    For convenience, let $\mathcal{W}$ and $\mathcal{W}'$ denote $\W_E$ and $\W_{E'}$, respectively, and let $\W_c$ and $\W_{\overline c}$ denote the committees in $\W$ containing $c$ and not containing $c$, respectively. (Likewise for $\W_c'$ and $\W_{\overline c}'$.)

    For the property \ejrp{} an additional approval of some candidate $c$ only `helps' committees containing $c$, and only `hurts' committees not containing $c$; in particular, $\W_c \subseteq \W_c'$ and $\W_{\overline c}' \subseteq \W_{\overline c}$.
    As a result, we have
    \begin{align*}
        g(\pavp{W}) &= g(\pav{W}) %
        \qquad &\text{for $W \in \W_{\overline c}' \cap \W_{\overline c} = \W_{\overline c}'$}, \\
        g(\pavp{W}) &= g(\pav{W})\cdot e^{a\abs{A_v' \cap W}^{-1}} \qquad &\text{for $W \in \W_c' \cap \W_c = \W_c$.}
    \end{align*}
    
    Then for candidate $c$, we bound the probability a committee containing it is chosen under $E'$ by
    \begin{align*}
        \pi_c' &= \frac{\sum_{W \in \W_c'} g(\pavp{W})}{\sum_{W \in \W'} g(\pavp{W})} \\
        &= \frac{\sum_{W \in \W_c} g(\pav{W}) \cdot e^{a\abs{A_v' \cap W}^{-1}} + \sum_{W \in \W_c' \setminus \W_c} g(\pavp{W})} 
        {\sum_{W \in \W_c} g(\pav{W}) \cdot e^{a\abs{A_v' \cap W}^{-1}} + \sum_{W \in \W_c' \setminus \W_c} g(\pavp{W}) + \sum_{W \in \W_{\overline c}'} g(\pav{W})} \\
        &\geq \frac{\sum_{W \in \W_c} g(\pav{W}) \cdot e^{a\abs{A_v' \cap W}^{-1}}} 
        {\sum_{W \in \W_c} g(\pav{W}) \cdot e^{a\abs{A_v' \cap W}^{-1}} + \sum_{W \in \W_{\overline c}'} g(\pav{W})} \\
        &\geq \frac{\sum_{W \in \W_c} g(\pav{W})} 
        {\sum_{W \in \W_c} g(\pav{W}) + \sum_{W \in \W_{\overline c}'} g(\pav{W})} \\
        &\geq \frac{\sum_{W \in \W_c} g(\pav{W})} 
        {\sum_{W \in \W_c} g(\pav{W}) + \sum_{W \in \W_{\overline c}'} g(\pav{W}) + \sum_{W \in \W_{\overline c}\setminus \W_{\overline c}'} g(\pav{W})} \\
        &= \frac{\sum_{W \in \W_c} g(\pav{W})}{\sum_{W \in \W} g(\pav{W})} \\
        &= \pi_c,
    \end{align*}
    where we repeatedly applied the facts that $\frac{a+c}{b+c} \geq \frac{a}{b}$ and $\frac{a}{b} \geq \frac{a}{b+c}$ for $a, b, c > 0$ and $a \leq b$.
\end{proof}

\begin{restatable}{lemma}{noisyPAVstable} \label{lem:exp-pav-stable}
    \noisypav{} (\Cref{alg:exp-PAV}) is $\eps$-TV-stable for $\eps=a\cdot \hat\delta \cdot \gamma$, where $a = \frac{k^2}{n}\log(\frac{m^k n}{k^2})$, $\gamma = 5e^2$ is a constant, and the bound on the change to the approvals of voter $v$ is is measured by $\hat\delta \defeq \harm(\min(\abs{\Delta}, \:k)) \leq \log k + 1$.
\end{restatable}

Let $\pav{}^* \defeq \max_{W \in C} \pav{W}$ be the maximum PAV score. In order to show stability for this algorithm, we will need the fact that a committee $W$ only enters (or leaves) the set $\mathcal{W}_E$ when its PAV score $\pav{W}$ is far from $\pav{}^*$. 
This is captured by the following: 
\begin{proposition}[\cite{brill23robust}]\label{prop:EJRp-buffer}
    For committee $W$, if $\pav{W} \geq \pav{}^* - \frac{n}{k^2}$ then $W$ satisfies \ejrp{}.
\end{proposition}
\begin{proof}
     This is effectively a corollary of \citet[Theorem 1]{aziz18complexity}.
     They show that a committee $W$ satisfies \ejr{} if no single-candidate pivots improve the PAV score by at least $\frac{n}{k^2}$. 
     \citet{brill23robust} observe this holds for \ejrp{} also.
     Finally, if $\pav{W} \geq \pav{}^* - \frac{n}{k^2}$ then this is certainly the case.
\end{proof}

We are now ready to prove stability.
We will use the choice 
\begin{equation} \label{eq:exp-PAV-a-defn}
    a \defeq \frac{k^2}{n} \left(k\log m + \log \frac{n}{k^2}\right).
\end{equation}

\begin{proof}[Proof of \Cref{lem:exp-pav-stable}]
    We begin by registering some assumptions we can make on the problem parameters without loss of generality.
    First, since $\eps$-TV-stability always holds for $\eps \geq 1$, we have:
    \begin{observation} \label{obs:a-leq-deltav-for-PAV-stability}
        We may assume without loss of generality that $a \leq 1/\hat\delta$.
    \end{observation} 

    Next, we also have:
    \begin{observation} \label{obs:n-versus-k-for-PAV-stability}
        We may assume the parameters for election $E$ satisfy $k \log m + \log n - 2 \log k \geq 1$. 
    \end{observation} 
    This is without loss of generality because it holds provided that $n \geq 2$, $m \geq 2$, and $m \geq k + 1$. 
    Observe that if $m = 1$ or $m = k$ then the output of \noisypav{} is constant, and therefore $0$-TV-stable. 
    And if $m \geq k + 1$ but $n = 1$ then we may confirm that the claimed $\eps$ in the statement of \Cref{lem:exp-pav-stable} evaluates to more than $1$, and therefore holds trivially.

    We now proceed to the proof.

    For a given election $E$, let $E'$ be the election derived by changing the approval votes of voter $v$.
    Let $A_v$ and $A_v'$ denote the voter's approval sets before and after the change, and recall that $\Delta_v = A_v \oplus A_v'$ is the set of candidates that $v$ changes their approval for. 
    Let $\D, \D'$ be the distributions over committees resulting from \noisypav on $E$ and $E'$, so for each committee $W$ satisfying \ejrp{} for $E$, $\probover{\D}{W} = P_W= \frac{g(\pav{W})}{\sum_{W'} g(\pav{W'})}$, and for each committee satisfying \ejrp{} for $E'$ we have $\probover{\D'}{W} = P_W'= \frac{g(\pavp{W})}{\sum_{W'} g(\pavp{W'})}$. 
    We will assume $a \leq 1/\delta_v$, since otherwise the claim is trivial.
    
    By \Cref{fac:quotientdiff}, we have 
    \begin{align}
        &d_{TV}(\D, \D') %
        = \sum_{W \in \W \cap \W'} \abs{\frac{g(\pav{W})}{\sum_{W' \in \W} g(\pav{W'})} - \frac{g(\pavp{W})}{\sum_{W' \in \W'} g(\pavp{W'})}} + \sum_{W \in \W \setminus \W'} P_W + \sum_{W \in \W' \setminus \W} P_W' \notag \\
        &\quad\leq \sum_{W \in \W \cap \W'} \left[ g(\pav{W})\cdot \abs{\frac{1}{G} - \frac{1}{G'}} + \frac{1}{G'} \cdot \abs{g(\pav{W}) - g(\pavp{W})} \right] + \sum_{W \in \W \setminus \W'} P_W + \sum_{W \in \W' \setminus \W} P_W', \label{eq:continuity-halfway}
    \end{align}
    where we let $G \defeq \sum_{W' \in \W} g(\pav{W'})$ and $G' \defeq \sum_{W' \in \W'} g(\pavp{W'})$ for convenience.
    
    We will bound these terms separately. 
    It will be more manageable to separate them into lemmas.
    We start with the last terms.
    \begin{lemma} \label{lem:pav-stable-support-diff-bound}
        \[
            \sum_{W \in \W \setminus \W'} \frac{g(\pav{W})}{G} + \sum_{W \in \W' \setminus \W} \frac{g(\pavp{W})}{G'} \leq 2\cdot \exp\left( k \log m - a\frac{n}{k^2} + 2a \delta_v \right).
        \]
    \end{lemma}
    \begin{proof}
        We focus on the first term.
        By \Cref{prop:EJRp-buffer}, we know that for all $W \in \W \setminus \W'$ it holds that
        \[
            \pavp{W} \leq \pavp{}^* - n/k^2,
        \]
        Using the fact that $\pav{W} \leq \pavp{W} + \delta_v$ and $\pavp{}^* \leq \pav{}^* + \delta_v$, we therefore have
        \begin{align}
            \sum_{W \in \W \setminus \W'} g(\pav{W})
            &\leq \abs{\W \setminus \W'} \cdot g\left( \pav{}^* - n/k^2 + 2\delta_v \right) \notag \\
            &\leq \binom{m}{k} \cdot g(\pav{}^*) \cdot e^{a (- n/k^2 + 2\delta_v)} \notag \\
            &\leq g(\pav{}^*) \cdot \exp\left( k \log m - a\frac{n}{k^2} + 2a \delta_v \right). \notag
        \end{align}
        Finally, we observe that $G \geq g(\pav{}^*)$ since the PAV maximizer satisfies \ejrp{}.
        The same holds for the second term, and the claim follows.
    \end{proof}
    
    \begin{lemma} \label{lem:g-diff-bound}
        If $E$ and $E'$ differ by at most one voter's approvals $A_v \neq A_v'$ and $a \leq 1/\delta_v$, then
        \begin{align*}
            \abs{g(\pav{W}) - g(\pavp{W})} &\leq g(\pavp{W}) \cdot a \cdot \delta_v \cdot (e-1).
        \end{align*}
    \end{lemma}
    \begin{proof}
        If the approvals of only one voter $v$ change, we have
        \[
            \abs{\pav{W} - \pavp{W}} 
            = \abs{\harm\left(\abs{A_i \cap W}\right) - \harm\left(\abs{A_i' \cap W}\right)} \leq \harm\left(\min(k, \Delta_v)\right).
        \]
        
        Then since we assume that $a \leq 1/\delta_v$, we may use that $e^{ax} -1 \leq (e-1) ax$ for $x \in [0, 1/a]$ and likewise $1 - e^{-ax} \leq ax$.
        For a given committee $W$ we bound $\abs{g(\pav{W}) - g(\pavp{W})}$ by
        \begin{align}
            e^{-a \cdot \delta_v} \cdot g(\pavp{W}) &\leq g(\pav{W})\leq g(\pavp{W}) \cdot e^{a \cdot \delta_v} \notag \\
            \abs{g(\pav{W}) - g(\pavp{W})} &\leq g(\pavp{W})  \cdot \max\left(1 - e^{-a\cdot \delta_v}, e^{a \cdot \delta_v} - 1 \right) \notag\\
            &\leq g(\pavp{W}) \cdot a\cdot \delta_v \cdot (e-1). \qedhere
        \end{align}
    \end{proof}

    We now bound the difference in reciprocal total weights:
    \begin{lemma} \label{lem:g-reciprocal-bound}
        For $a \leq 1/\hat\delta$,
        \[
            \abs{\frac{1}{G} - \frac{1}{G'}} \leq \frac{1}{G} \cdot \left( \delta_v \cdot a \cdot (e-1) + 2\exp\left( k \log m - a\frac{n}{k^2} + 2 a \delta_v \right) \right) . 
        \]
    \end{lemma}
    \begin{proof}
    Again, let $\delta_v \defeq \harm\left(\min(k, \Delta_v)\right)$.
    Rewriting the left-hand side and passing the absolute value inside the sum, we have
    \begin{align}
        \frac{\abs{G' - G}}{GG'} 
        &\leq \frac{1}{GG'} \left( \sum_{W \in \W \cap \W'} \abs{g(\pavp{W}) - g(\pav{W})} +  \sum_{W \in \W \setminus \W'} g(\pav{W}) + \sum_{W \in \W' \setminus \W} g(\pavp{W}) \right)  \notag \\
        &\leq \frac{1}{GG'} \left( \delta_v \cdot a \cdot (e-1) \cdot \sum_{W \in \W'} g(\pavp{W}) +  \sum_{W \in \W \setminus \W'} g(\pav{W}) + \sum_{W \in \W' \setminus \W} g(\pavp{W}) \right)  \notag \\
        & = \delta_v \cdot a \cdot (e-1) \cdot \frac{1}{G} + \frac{1}{GG'} \left(\sum_{W \in \W \setminus \W'} g(\pav{W}) + \sum_{W \in \W' \setminus \W} g(\pavp{W}) \right), \label{eq:reciprocal-halfway}
    \end{align}
    where we applied \Cref{lem:g-diff-bound} to each term in $\sum_{W \in \W \cap \W'} \abs{g(\pav{W}) - g(\pavp{W})}$.

    Next we upper bound the remaining sums in \eqref{eq:reciprocal-halfway}.
    The proof of \Cref{lem:pav-stable-support-diff-bound} directly yields 
    \[
         \sum_{W \in \W' \setminus \W} g(\pavp{W}) \leq G' \cdot \exp\left( k \log m - a\frac{n}{k^2} + 2a \delta_v \right).
    \]
    We deploy a similar argument for the first sum.
    By \Cref{prop:EJRp-buffer}, we have $\pav{W} - \delta_v \leq \pavp{W} \leq \pavp{}^* - n/k^2$ for all $W \in \W \setminus \W'$, so
    \begin{align}
        \sum_{W \in \W \setminus \W'} g(\pav{W})
        &\leq \abs{\W \setminus \W'} \cdot g\left( \pavp{}^* - n/k^2 + \delta_v \right) \notag \\
        &\leq \binom{m}{k} \cdot g(\pavp{}^*) \cdot e^{a (- n/k^2 + \delta_v)} \notag \\
        &\leq G' \cdot \exp\left( k \log m - a\frac{n}{k^2} + a \delta_v \right) \notag
    \end{align}
    since $G' \geq g(\pavp{}^*)$. 
    Substituting these bounds into \eqref{eq:reciprocal-halfway} gives the stated claim.
    \end{proof}
    
    We are now ready to complete our proof of \Cref{lem:exp-pav-stable}. 
    Plugging \Cref{lem:pav-stable-support-diff-bound}, \Cref{lem:g-diff-bound}, and \Cref{lem:g-reciprocal-bound} into the prior bound \eqref{eq:continuity-halfway}, we have
    \begin{align}
        d_{TV}(\D, \D') 
        &\leq \sum_{W \in \W \cap \W'} \left[\frac{g(\pav{W})}{G} \cdot \left( a \delta_v \cdot (e-1) + 2\exp\left( k \log m - a\frac{n}{k^2} + 2 a \delta_v \right) \right) + \frac{g(\pavp{W})}{G'}\cdot a \delta_v \cdot (e-1) \right] \notag \\
        &\qquad\qquad + 2\cdot \exp\left( k \log m - a\frac{n}{k^2} + 2a \delta_v \right) \notag \\
        &\leq a \delta_v \cdot 2(e-1) + 4\cdot \exp\left( k \log m - a\frac{n}{k^2} + 2 a \delta_v \right) .  \label{eq:pav-stability-bound-intermediate}
    \end{align}
    The claim holds for $a = \frac{k^2}{n} \left( k\log m + \log(n/k^2)\right)$ since 
    by \Cref{obs:a-leq-deltav-for-PAV-stability},
    \begin{align}
        \exp\left( k \log m - a\frac{n}{k^2} + 2 a \delta_v \right) &\leq \exp\left( k \log m - a\frac{n}{k^2} + 2 \right) \notag \\
        &= \exp\left( k \log m - k \log m - \log(n/k^2) + 2 \right) \notag \\
        &= e^2 \cdot \exp\left(- \log(n/k^2)\right) \notag \\
        &= e^2 \cdot \frac{k^2}{n} \notag \\
        &\leq e^2 \cdot a, \notag 
    \end{align}
    Since we assume that $k \log m + \log n - 2 \log k \geq 1$ by \Cref{obs:n-versus-k-for-PAV-stability}, and so $\frac{k^2}{n} \leq a$. 
    Substituting this into \eqref{eq:pav-stability-bound-intermediate}, we have that $d_{TV}(\D, \D') \leq a\delta_v \cdot 5e^2$, as claimed.
\end{proof}

With \Cref{lem:exp-pav-stable} in hand, the candidate continuity of \noisypav{} follows.

\expPAVcontinuous*

\begin{proof}[Proof of \Cref{thm:exp-pav-continuous}]
    By \Cref{lem:exp-pav-stable}, we know that \noisypav{} is $\eps$-TV-stable for $\eps = \hat\delta \cdot \gamma \cdot \frac{k^2}{n}\log(\frac{m^k n}{k^2})$ and $\gamma = 5e^2$.
    By \Cref{obs:marginalseasier}, \noisypav{} is therefore also $\eps$-continuous.
\end{proof}

\section{\texorpdfstring{Proofs Elided from \Cref{sec:dp-abc}}{Proofs Elided from Section 5}}
\label{app:dp-abc}

\greedyrulesbadforDP*

\begin{proof}[Proof of \Cref{obs:greedy-bad-for-DP}]
    Consider an election $E$ with $k$ candidates $c_1 \ldots, c_k$, and voter groups $V_1, \ldots, V_k$ of size $n/k$, where the voters $v\in V_r$ approve only of $A_v = \{c_r\}$.
    Next consider $E'$ derived by removing one voter from $V_k$, who now approves of no candidates.

    For $E$, both rules output $W = \{c_1, \ldots, c_k\}$ with probability 1, which is good since it is the only committee that satisfies \jr{}. 
    For $E'$, both $W$ and $W' = \{c_1, \ldots, c_{k-1}\}$ satisfy \jr{}. 
    In this case, for both rules the probability that $W'$ is chosen is equal to the probability that $c_k$ is not picked in the last round:
    \begin{equation} \notag
        \prob{\alg(E') = W'} = 1-P_{c_k} = 1 - \frac{g(\frac{n}{k} - 1)}{g(\frac{n}{k})} = 1 - e^{-a} \geq \frac{1}{1-1/e} \cdot a
    \end{equation}
    for $a \leq 1$. 
    Meanwhile, a $(\eps, \delta)$-max-KL stability guarantee would imply that
    \begin{equation*}
        \prob{\alg(E') = W'} \leq \prob{\alg(E) = W'}\cdot e^\eps + \delta = \delta. \qedhere
    \end{equation*}
\end{proof}

\noisyPAVmaxKLstable*

\begin{proof}[Proof of \Cref{thm:exp-pav-kl-stable}]
    Recall that by the definition of max-KL stability (\Cref{def:kl-stability}), our aim is to prove that for all subsets of committees $R \subseteq \binom{C}{k}$, and for all pairs of approval profiles $A, A' \in \cA$ that differ only on $A_v$ for some $v \in V$,
    \begin{equation} \label{eq:kl-stabililty-pav-statement}
        \prob{\alg(A,k) \in R} \leq e^\eps \cdot \prob{\alg(A',k) \in R} + \delta,
    \end{equation}
    where \alg{} is \noisypav{}.

    We will consider each $W$ individually, and then derive a guarantee for \eqref{eq:kl-stabililty-pav-statement} by summing over the $W \in R$.
    Drawing inspiration from the proof of \Cref{lem:exp-pav-stable}, for a given $W$ we consider two cases: if $W \in \W \cap \W'$, and if $W \in \W' \setminus \W$.
    (Clearly if $W \not \in \W \cup \W'$ then $\prob{\alg(A,k) = W} = \prob{\alg(A',k) = W} = 0$.)
    Throughout we will for compactness let $P_W \defeq \prob{\alg(A,k) = W}$ and $P_W' \defeq \prob{\alg(A',k) = W}$, as well as $P_R \defeq \prob{\alg(A',k) \in T}$ for a set of committees $T$.

    If $W \in \W \setminus \W'$, we follow the argument of \Cref{lem:pav-stable-support-diff-bound}.
    We first recall that for all $W \in \W \setminus \W'$
    \[
        \pavp{W} \leq \pavp{}^* - n/k^2,
    \]
    by \Cref{prop:EJRp-buffer}.
    Using the fact that $\pav{W} \leq \pavp{W} + \delta_v$ and $\pavp{}^* \leq \pav{}^* + \delta_v$, we therefore have
    \begin{align}
        g(\pav{W})
        &\leq g\left( \pav{}^* - n/k^2 + 2\delta_v \right) \notag \\
        &= g(\pav{}^*) \cdot \exp\left( - a\frac{n}{k^2} + 2a \delta_v \right). \notag
    \end{align}
    Finally, we observe that $\sum_{W' \in \W} g(\pav{W'}) \geq g(\pav{}^*)$ since the PAV maximizer satisfies \ejrp{}, and so
    \begin{equation} \label{eq:kl-stability-pav-case-1}
        \prob{\alg(A,k) = W} = \frac{g(\pav{W})}{\sum_{W' \in \W} g(\pav{W'})} \leq \exp\left( - a\frac{n}{k^2} + 2a \delta_v \right).
    \end{equation}
    Summing over $W \in T$ for a set of committees $T\subseteq \W \setminus \W'$, this guarantees
    \begin{align}
        \prob{\alg(A,k) \in T} &\leq \abs{T}  \cdot \exp\left( - a\frac{n}{k^2} + 2a \delta_v \right). \label{eq:kl-stability-pav-case-1-aggregate}
    \end{align}

    If $W \in \W \cap \W'$, then letting $G \defeq \sum_{W' \in \W} g(\pav{W'})$ and $G' \defeq \sum_{W' \in \W} g(\pav{W'})$, we have 
    \begin{align}
        \frac{\prob{\alg(A,k) = W}}{\prob{\alg(A',k) = W}}
        &= \frac{g(\pav{W})}{g(\pavp{W})} \cdot \frac{G'}{G} \notag\\
        &\leq e^{a\delta_v} \cdot \frac{G'}{G}, \label{eq:kl-stability-pav-case-2-intermediate}
    \end{align}    
    where we used that $\pav{W} \leq \pavp{W} + \delta_v$ and the definition of $g$.
    Next we bound $G'/G$ by 
    \begin{align}
        \frac{G'}{G} &\leq \frac{\sum_{W' \in \W} g(\pavp{W'}) + \sum_{W' \in \W' \setminus \W} g(\pavp{W'})}{G} \notag \\
        &\leq \frac{G}{G}\cdot e^{a\delta_v} + \frac{\sum_{W' \in \W' \setminus \W} g(\pavp{W'})}{G} \notag \\
        &\leq e^{a\delta_v} + \abs{\W' \setminus \W} \cdot \exp\left( - a\frac{n}{k^2} + a \delta_v \right), \notag
    \end{align}
    where we used that $\pavp{W'} \leq \pav{W'} + \delta_v$ and the definition of $g$, followed by an argument similar to how we arrived at \eqref{eq:kl-stability-pav-case-1} that $\pavp{W'} \leq \pav{W'} + \delta_v \leq \pav{}^* - n/k^2 + \delta_v$ for $W' \not \in \W$, and $G \geq g(\pav{}^*)$.
    Combining this with \eqref{eq:kl-stability-pav-case-2-intermediate}
    , we have
    \begin{align}
        \prob{\alg(A,k) = W} &\leq \prob{\alg(A,k) = W}\cdot \left( e^{2a\delta_v} + \abs{\W' \setminus \W} \cdot \exp\left( - a\frac{n}{k^2} + 2a \delta_v \right)\right). \label{eq:kl-stability-pav-case-2}
    \end{align}
    Summing over $W \in T$ for a set $T$, this becomes 
    \begin{align}
        \prob{\alg(A,k) \in T} &\leq \prob{\alg(A,k) \in T}\cdot \left( e^{2a\delta_v} + \abs{\W' \setminus \W} \cdot \exp\left( - a\frac{n}{k^2} + 2a \delta_v \right)\right) \notag \\
        &\leq \prob{\alg(A,k) \in T}\cdot e^{2a\delta_v} + \abs{\W' \setminus \W} \cdot \exp\left( - a\frac{n}{k^2} + 2a \delta_v \right). \label{eq:kl-stability-pav-case-2-aggregate}
    \end{align}
    
    Finally we combine \eqref{eq:kl-stability-pav-case-1-aggregate} and \eqref{eq:kl-stability-pav-case-2-aggregate} to obtain a guarantee for an arbitrary subset of committees $R$ of the form of \eqref{eq:kl-stabililty-pav-statement}.
    Since $\prob{\alg(A,k) \in R\setminus \W} = 0$, we have
    \begin{align}
        \prob{\alg(A,k) \in R} 
        &= \prob{\alg(A,k) \in R\cap \W \cap \W'} + \prob{\alg(A,k) \in R\cap \W \setminus \W'} \notag \\
        &\leq \prob{\alg(A,k) \in R\cap \W \cap \W'} + \abs{R\cap \W \setminus \W'}  \cdot \exp\left( - a\frac{n}{k^2} + 2a \delta_v \right), \notag
        \intertext{by \eqref{eq:kl-stability-pav-case-1-aggregate}. Applying \eqref{eq:kl-stability-pav-case-2-aggregate} for the first term and using that $\abs{R\cap \W \setminus \W'} \leq \abs{\W \setminus \W'}$ yields}
        &\leq \prob{\alg(A,k) \in R\cap \W \cap \W'} \cdot e^{2a\delta_v} + \abs{\W \oplus \W'}  \cdot \exp\left( - a\frac{n}{k^2} + 2a \delta_v \right). \notag \\
        &\leq \prob{\alg(A,k) \in R} \cdot e^{2a\delta_v} + \abs{\W \oplus \W'}  \cdot \exp\left( - a\frac{n}{k^2} + 2a \delta_v \right). \notag 
    \end{align}
    Taking $\abs{\W \oplus \W'} \leq m^k$ gives our final $(\eps,\delta)$ claim.
\end{proof}

\noisyPAVmaxKLstableCorrolary*

\begin{proof}[Proof of \Cref{cor:exp-pav-kl-stable}]
    We proceed by evaluating $\delta$ from \Cref{thm:exp-pav-kl-stable} for our choice of $a$ and assuming $\Delta = m$, and so $\hat\delta \leq \log k + 1$. 
    Then
    \begin{align*}
        \delta &= \exp\left(k \log m - a\frac{n}{k^2} + 2a \delta_v \right) \\
        &\leq \exp\left(k \log m - a\frac{n}{k^2} + k \log m + \kappa \log n \right) \\
        &= \exp\left(2 k \log m + \kappa \log n - a\frac{n}{k^2}  \right) \\
        &= \exp\left(2 k \log m + \kappa \log n - 2k \log m - 2 \kappa \log n  \right) \\
        &= \exp\left(- \kappa \log n  \right) \\
        &= n^{-\kappa} \qedhere
    \end{align*}
\end{proof}

\noisyGJCRmaxKLstable*

\begin{proof} [Proof of \Cref{lem:noisy-GJCR-weak-DP-guarantee}]
    This follows from the definition of $TV$ distance (\Cref{def:tv-dist}). 
    Consider committees $R \subseteq \W$.
    Then for $\eps \geq 0$, letting $f$ denote \noisygjcr{}, we have
    \begin{align*}
        \prob{f(E) \in R} - e^\eps \cdot \prob{f(E') \in R} &\leq \prob{f(E) \in R} - \prob{\alg(E') \in R} \\
        &\leq \max_{A \subseteq \W} \left(\prob{f(E) \in A} - \prob{f(E') \in A} \right) \\
        & = \dtv{\cD_{f(E)}}{\cD_{f(E')}} \\
        & \leq \gamma \cdot \frac{k^3}{n} \cdot \log k \cdot \log\left(\frac{n}{k}\Delta_v \right), 
        \intertext{by \Cref{lem:noisy-GJCR-stable}. Then rearranging yields}
        \prob{f(E) \in R}& \leq e^\eps \cdot \prob{f(E') \in R} + \gamma \cdot \frac{k^3}{n} \cdot \log k \cdot \log\left(\frac{n}{k}\Delta_v \right), 
    \end{align*}
    proving the claim.
\end{proof}

\nopureDPwithJR*

\begin{proof}[Proof of \Cref{obs:no-pure-DP}]
    In \Cref{ex:continuity} the approval profile $A$ consists of $k$ equal-size groups of voters $N_1, \ldots, N_k$, approving distinct candidates $c_1, \ldots, c_k$ (and no others). 
    The only \jr{} committee is $W \defeq \{c_1, \ldots, c_k\}$. 

    Now consider changing the approvals of the voters in $N_1$ so that they approve only of some candidate $d_1\not\in W$. 
    For this approval profile $A'$ the only \jr{} committee is $W' \defeq \{d_1, \ldots, c_k\}$.
    
    If a randomized rule $f$ satisfies \jr{} and $(\eps,0)$-max-KL stability, then by chaining together its guarantee across these $n/k$ changes from $A$ to $A'$, we have
    \[
        1 = \prob{f(A) = W} \leq e^{\eps n/k} \cdot \prob{f(A') = W} = 0.
    \]
    This is unsatisfiable for any $\eps$.
\end{proof}

\DPlowerboundsforJR*

\begin{proof}[Proof of \Cref{claim:JR-DP-lowerbound}]
    We simply adapt the proof of \Cref{obs:no-pure-DP} to include $\delta$. 
    Considering the same profiles $A$ and $A'$, applying $(\eps,\delta)$-max-KL stability (\Cref{def:kl-stability}) $n/k$ times to get from one to the other and recalling that $\prob{f(A') = W} = 0$ yields
    \begin{align*}
        1 = \prob{f(A) = W} &\leq \delta + e^\eps \left( \delta + e^\eps \left( \ldots e^\eps \left( \delta + e^\eps \cdot 0\right) \right) \right) \\
        & = \delta \cdot \sum_{i = 0}^{n/k-1} e^{\eps i} = \delta \cdot \frac{e^{\eps n/k} - 1}{e^\eps - 1} 
        \leq \frac{\delta}{\eps} \cdot e^{\eps n/k},
    \end{align*}
    since $e^x - 1 \geq x$ for all $x \geq 0$. 
    Rearranging gives the stated claim.
\end{proof}

\section{\texorpdfstring{Proofs Elided from \Cref{sec:dynamic}}{Proofs Elided from Section 6}} 
\label{app:dynamic}

\dynamictostablereduction*

\begin{proof}[Proof of \Cref{prop:reduction}]
    We will argue that all consecutive pairs of outputs $(R^t, R^{t-1})$ are distributed according to the optimal couplings of $\cD^t \defeq \cD_{f(A^t)}$ and $\cD^{t-1} \defeq \cD_{f(A^{t-1})}$. 
    By the definition of the optimal coupling (\Cref{lem:coupling}) and $\eps$-TV-stability, it then follows that $\prob{R^t \neq R^{t+1}} \leq \eps$, where the probability is taken over the randomness in \cref{line:reduction-initial-sample} and \cref{line:reduction-conditional-sample} of \Cref{alg:dynamic-stable-reduction}.
    Then the expected recourse of $(R^t)_{t \in [T]}$ is upper-bounded by $\eps$ times the maximum possible value of $|R^t \oplus R^{t-1}|$, which is what $K$ reflects in the statement of \Cref{prop:reduction}.

    It remains to argue that for all $t$, $(R^t, R^{t-1})$ is distributed according to the optimal coupling $\cE^t$ of $\cD^t$ and $\cD^{t-1}$.
    If $R^{t-1}$ is distributed according to $\cD^{t-1}$ then this is the case, because then in \cref{line:reduction-conditional-sample} for each $x,y \in \cR$
    \begin{align}
        \prob{R^{t} = x \:\wedge\: R^{t-1} = y}
        &= \probover{(X, Y) \sim \cE^t}{X = x \given Y = y} \cdot \probover{Y \sim \cD^{t-1}}{Y = y} \notag \\
        &= \probover{(X, Y) \sim \cE^t}{X = x \given Y = y} \cdot \probover{(X, Y) \sim \cE^t}{Y = y} \notag \\
        &= \probover{(X, Y) \sim \cE^t}{X = x \:\wedge\: Y = y} \notag
    \end{align}
    by the definition of a coupling of $\cD^t$ and $\cD^{t-1}$ and Bayes' rule.

    We finally prove that all $R^t$ are distributed according to $\cD^t$ by induction.
    \Cref{line:reduction-initial-sample} comprises the base case.
    Given that $R^t \sim \cD^t$, for all $x \in \cR$ we have
    \begin{align}
        \prob{R^{t+1} = x}
        &= \sum_{y \in \cR} \probover{(X, Y) \sim \cE^{t+1}}{X = x \given Y = y} \cdot \probover{(X, Y) \sim \cE^{t+1}}{Y = y} \label{eq:dynamic-reduction-condnl-induction} \\
        &= \sum_{y \in \cR} \probover{(X, Y) \sim \cE^{t+1}}{X = x \given Y = y} \cdot \probover{Y \sim \cD^t}{Y = y} \notag \\
        &= \sum_{y \in \cR} \probover{(X, Y) \sim \cE^{t+1}}{X = x \:\wedge\: Y = y} \notag \\
        &= \probover{X \sim \cD^{t+1}}{X = x} \notag
    \end{align}
    again by the definition of a coupling (\Cref{def:coupling}) and Bayes’ rule.
\end{proof}

\begin{algorithm}[ht!]
\caption{\dynamicgjcr{}}
\label{alg:gjcr-dynamic}
\Input{$(E^t)_{t \in [T]} = (C, N, k, A^t)_{t \in [T]}$}
\Output{$(W^1, \ldots, W^T)$, where $W^t$ satisfies \ejrp{} for $E^t$ for all $t$}
$W^1 \gets \noisygjcr(E^1)$\; \label{line:dynamic-gjcr-initial-sample}
$s^1 \gets$ candidate sequence used to construct $W^1$ \;
\For{$t \in (2, \ldots, T)$}{
    $ s \leftarrow s^{t-1}$ and $s' \leftarrow ()$ \;
    Generate distributions $\cN_{s\vert r}^{t-1}, \cN_{s\vert r}^{t}$ from \cref{line:GJCR-noisy-sample} of \Cref{alg:gjcr-noisy} on inputs $E^{t-1}$ and $E^t$,\newline for all prefixes $s\vert r$ of $s$ for $r \in [k^2]$\; %
    Construct $\cM_{s\vert r}^t \gets $ optimal couplings of $\cN_{s\vert r}^{t-1}, \cN_{s\vert r}^{t}$, for $r \in [k^2]$ \;
    \For{$r \in (1, \ldots, k^2)$}{ \label{line:gjcr-dynamic-sample}
        \uIf{$s'\vert r = s \vert r$}{ %
            Sample $(c',c)\sim (\cM_{s' \vert r}^t$ conditioned on the first coordinate being equal to $s(r)$)\; \label{line:gjcr-dynamic-condnl-sample} %
        }
        \Else{ 
            Sample $c \sim \cN_{s'\vert r}^{t}$\; \label{line:gjcr-dynamic-fresh-sample}
        }
        $s' \leftarrow$ ($s'$ with $c$ appended)\;
    }
    $W^t \leftarrow \{c \in C: c \in s'$\}\;
    $s^t \leftarrow s'$\;
}
\Return $W$\;
\end{algorithm}

\dynamicgjcrpolytime*

\begin{proof}[Proof of \Cref{thm:dynamicgjcr}]
    Recall the couplings and distributions used in the proof of \Cref{thm:noisy-GJCR-continuous}. 
    In particular, let $\cN_s$ and $\cN_s '$ denote the distributions from which \noisygjcr samples on \cref{line:GJCR-noisy-sample}, given a partial sequence $s$ and under inputs $E$ and $E'$, respectively. 
    We let $\cM_s$ denote an optimal coupling of $\cN_s$ and $\cN_s '$.

    The key insight is that although coupling the output distributions of \noisygjcr for $E^t$ and $E^{t+1}$ using these couplings $\cM_s$ of $\cN_s$ and $\cN_s '$ is perhaps suboptimal from the perspective of TV distance, it is computationally tractable.
    The distributions $\cN_s$ and $\cN_s '$ are constructed explicitly in \cref{line:GJCR-noisy-sample} in polynomial time. 
    Since each has support $m$, an optimal coupling $\cM_s$ can be explicitly computed.
    For each $c \in C \cup \{\bot\}$, let $P_c$ and $P_c'$ be the probability of selecting $c$ from distributions $\cN_s$ and $\cN_s'$, respectively.
    Then to construct $\cM_s$,
    \begin{itemize}
        \item first set $\cM_s(c,c) = \min(P_c, \:P_c')$ for all $c \in C \cup \{\bot\}$,
        \item then match excess probability mass for $c \in \cup \{\bot\}$ for which $P_c > P_c'$ to excess probability mass for $c' \in \cup \{\bot\}$ for which $P_{c'} < P_{c'}'$.
    \end{itemize}
    We claim $\cM_s$ is in fact an optimal coupling of $\cN_s$ and $\cN_s'$ as defined in \Cref{lem:coupling}, but for our purposes it suffices to observe that, by \Cref{lem:gjcr-stability-per-step} and the choice of $a_\ell$ in \eqref{eq:gjcr-a-defn}, these distributions satisfy 
    \[
        \dtv{\cN_s}{\cN_s'} = O(a_\ell).
    \]
    By our construction of $\cM_s$, the total probability of selecting a pair of different elements from $\cM_s$ is
    \[
        \probover{(X,Y) \sim \cM_s}{X \neq Y}
        = \sum_{c \in C \cup \{\bot\}: P_c \geq P_c'} (P_c - P_c') = \dtv{\cN_s}{\cN_s'} = O(a_\ell)
    \]
    by the definition of TV distance (\Cref{def:tv-dist}).

   For a partial sequence of samples $s$ in \noisygjcr, we use $s \vert r$ to denote the sequence up to index $r$ (exclusive).
    We claim that at all stages $r$ of the loop on \cref{line:gjcr-dynamic-sample}, the next candidate $c$ is a sample from $\cN_{s'\vert r}^t$. This is true outright in the case of \cref{line:gjcr-dynamic-fresh-sample}.
    Otherwise, following the proof of \Cref{prop:reduction} this follows by induction on $t$, and it holds over randomness over prior rounds and samples.
    In particular, the base case is step $r$ of the initial sample performed on \cref{line:dynamic-gjcr-initial-sample}. 
    Provided that $s(r)$ on \cref{line:gjcr-dynamic-condnl-sample} is a sample from $\cN_{s' \vert r}^{t-1}$, which is the inductive hypothesis, following \eqref{eq:dynamic-reduction-condnl-induction} and using that $\cM_{s' \vert r}^t$ is a coupling proves that the $c$ we draw on \cref{line:gjcr-dynamic-condnl-sample} is indeed a sample from $\cN_{s'\vert r}^t$.

    We have established that this runs in polynomial time.
    Since each output $W^t$ is a (correlated) sample from the distribution of $\noisygjcr(E^t)$, we have anonymity, neutrality, and monotonicity.
    Finally, by the analysis of the coupling above, for any $t$ the probability that $W^{t-1} \neq W^{t}$ is at most the probability that any of the joint samples on \cref{line:gjcr-dynamic-condnl-sample} disagree, which is
    \[
        \sum_{r \in [k^2]} \probover{(X,Y) \sim \cM_s}{X \neq Y} = O\left( \sum_{r \in [k^2]} a_\ell \right).
    \]
    By the analysis in \Cref{thm:noisy-GJCR-continuous} this is $\tilde{O}\left(k^3/n\right)$.
    We can finally conclude that the expected recourse of the output is upper bounded by the probability $\tilde{O}\left(k^3/n\right)$ times the maximum possible value of a single-step recourse $|W^{t-1} \oplus W^{t}| \leq 2k$.
    Hence, the recourse bound $\tilde{O}\left(k^4/n\right)$ follows.
\end{proof}

\end{document}